\documentclass{article}
\usepackage{style, style_ma}

\title{On the Structure of Replicable Hypothesis Testers}
\author{
Anders Aamand  \\
BARC, University of Copenhagen \\
\texttt{aa@di.ku.dk}
\and
Maryam Aliakbarpour\thanks{Department of Computer Science and
Ken Kennedy Institute}
\\ Rice University 
\\ \texttt{maryama@rice.edu}
\and
Justin Y.\ Chen \\
MIT \\
\texttt{justc@mit.edu}
\and
Shyam Narayanan \\
Citadel Securities \\
\texttt{shyam.s.narayanan@gmail.com}
\and
Sandeep Silwal \\
University of Wisconsin-Madison\\
\texttt{silwal@cs.wisc.edu}
}

\date{}

\begin{document}

\maketitle

\begin{abstract}
A hypothesis testing algorithm is \emph{replicable} if, when run on two different samples from the same distribution, it produces the same output with high probability. This notion, defined by by Impagliazzo, Lei, Pitassi, and Sorell [STOC'22], can increase trust in testing procedures and is deeply related to algorithmic stability, generalization, and privacy. We build general tools to prove lower and upper bounds on the sample complexity of replicable testers, unifying and quantitatively improving upon existing results.

We identify a set of canonical properties, and prove that any replicable testing algorithm can be modified to satisfy these properties without worsening accuracy or sample complexity. A canonical replicable algorithm computes a deterministic function of its input (i.e., a test statistic) and thresholds against a uniformly random value in $[0,1]$. It is invariant to the order in which the samples are received, and, if the testing problem is ``symmetric,'' then the algorithm is also invariant to the labeling of the domain elements, resolving an open question by Liu and Ye [NeurIPS'24]. We prove new lower bounds for uniformity, identity, and closeness testing by reducing to the case where the replicable algorithm satisfies these canonical properties.

We systematize and improve upon a common strategy for replicable algorithm design based on test statistics with known expectation and bounded variance. Our framework allow testers which have been extensively analyzed in the non-replicable setting to be made replicable with minimal overhead. As direct applications of our framework combined with existing analyses of non-replicable testers, we obtain constant-factor optimal bounds for coin testing and closeness testing and get replicability for free for uniformity testing in a large parameter regime. As replicable coin testing can be used as a black-box to turn any tester into a replicable tester, our results directly imply improved replicable sampling bounds for myriad applications beyond the ones specifically studied in this paper.

We also give a state-of-the-art algorithm for replicable Gaussian mean testing. We present a $\rho$-replicable algorithm for testing whether samples come from $\mathcal{N}(0, I)$ or from $\mathcal{N}(\mu, I)$ where $\|\mu\|_2 \ge \alpha$. Our algorithm improves over the previous best sample complexity of Bun, Gaboardi, Hopkins, Impagliazzo, Lei, Pitassi, Sivakumar, and Sorrell [STOC'23] and runs in polynomial time.
\end{abstract}

\newpage
{\small \tableofcontents}
\newpage
\setcounter{page}{1}

\section{Introduction}
Statistical testing is often modeled as follows: we receive samples from an unknown distribution $\mathcal{D}$ and the goal is to decide if $\mathcal{D}$ belongs to one of two pre-determined classes of hypotheses: the null or alternate hypothesis ($H_0$ or $H_1$). Collecting samples abstracts real world data curation such as running a physical experiment, and outputting $H_0$ or $H_1$ represents a data analyst's task of deciding which hypothesis or explanation best fits the observed data. The goal is to use the available data as efficiently as possible to derive meaningful inferences.

The following outcomes are desirable in practice: (a) if we resample data or repeat our experiments, then the same conclusion should likely be reached (reproducibility). This is crucial for scientific validity, as it implies that independent researchers investigating the same phenomenon using different (but similarly generated) data samples should ideally arrive at consistent findings. (b) Our data analysis procedure should not overfit to any particular data point observed in the sample (generalization). This helps ensure that conclusions drawn from data are robust and not mere artifacts of specific random samples. 

However, these two ideals are sometimes at odds with the statistical testing process in practice, which is oftentimes inherently unstable. For example, the data collection could involve inaccurate or noisy measurements, leading to model misspecification. Thus, the two hypothesis classes $H_0$ and $H_1$ may not capture all possible underlying data generation processes, and estimators designed by the analyst for distinguishing between $H_0$ and $H_1$, e.g. the data analysis performed under the assumption $\mathcal{D} \in \{H_0, H_1\}$ may fail in arbitrary, unexpected ways. Furthermore, the data analysis procedure itself could introduce variability unrelated to the data or the scientific question at hand, for example practices related to anti-patterns such as P-hacking or data dredging. Such instability in statistical testing poses fundamental problems: it may undermine the trustworthiness and reliability of statistical results, complicate efforts to verify findings, hinder direct comparisons between scientific studies, and contribute to the broader concerns often termed the ``replicability crisis" in empirical sciences \cite{baker20161}. 

Motivated by this, \cite{impagliazzo2022reproducibility} introduced a notion capturing both ``reproducibility'' and ``stability'', providing a quantitative framework of replicability in the entire algorithm design pipeline for statistical testing. Intuitively, a replicable statistical tester addresses the two issues of reproducibility and generalization by requiring the tester to have stable outputs even under the variability of the underlying data generation process (e.g. collecting a new set of data), as well as the internal randomness of the tester (e.g. re-running the data analysis code).  

\begin{definition}[Replicability \cite{impagliazzo2022reproducibility}]\label{def:replicability}
 A randomized algorithm $\A(X; r): \mathcal{X} \rightarrow \mathcal{Y}$ is $\rho$-replicable if for all distributions $\mathcal{D}$ over $\mathcal{X}$, 
 \[\Prma[X,X', r]{\A(X; r) = \A(X';r)} \ge 1-\rho, \]
where $X, X'$ denote sequences of i.i.d. samples from $\mathcal{D}$ and $r$ denotes the internal randomness used by $\mathcal{A}$.
\end{definition}
We note that the stability condition $\A(X; r) = \A(X';r)$ is not tied to the \emph{accuracy} of the statistical task, and we require it to hold for any choice of $\mathcal{D}$, the data generation process. (Of course as we will see later, the stability condition is usually paired with the usual requirement of statistical accuracy as well, e.g. if data is actually generated from $H_0$ or $H_1$ we should confidently detect it).\footnote{This is similar in spirit to differential privacy where data privacy must always hold (and correctness is a separate issue).}

We briefly remark that in addition to stability and generalization, replicability also captures desirable properties such as data privacy, as ``reproducible algorithms are prevented from memorizing anything that is specific to the training data, similar to differentially private algorithms" \cite{impagliazzo2022reproducibility}. Furthermore, unlike differential privacy in general \cite{gilbert2018property, gaboardiicalp}, replicability can be efficiently tested (in time polynomial in $1/\rho$ and the dimension of the data universe \cite{impagliazzo2022reproducibility}).

In addition to conforming to the practical motivations outlined above, the notion of replicability is also \emph{theoretically rich}. \cite{impagliazzo2022reproducibility} and subsequent work have shown that replicability is intimately connected to many other technical notions of algorithmic stability, including differential privacy \cite{kalavasis2023statistical, moran2023bayesian, bun2023stability, kalavasis2024computational}, generalization in adaptive data analysis \cite{impagliazzo2022reproducibility, moran2023bayesian, bun2023stability}, TV-stability \cite{kalavasis2023statistical, moran2023bayesian}, Local Computation Algorithms~\cite{canonne2025lca}, and other notions of stability \cite{chase1, chase2024local}. For many of these notions, there exist reductions between algorithms satisfying stability in one sense to replicability, e.g. any approximate differentially private algorithm can be turned into a replicable algorithm, albeit with polynomial blow-up in the sample complexity (the reduction is not computationally efficient).  We refer to Figure 1 in \cite{moran2023bayesian}, Figure 1 in \cite{bun2023stability}, and Figure 1.1 in \cite{kalavasis2024computational} for a web of reductions.

In addition to strong theoretical connections to various  notions of algorithmic stability,  \cref{def:replicability} has been quite influential and has inspired replicable algorithms for a wide range of statistical tasks, including reinforcement learning \cite{eaton2023replicable, karbasi2023replicability}, online learning \cite{esfandiari2023replicable, ahmadi2024replicable, kalavasis2024computational, larsen2025improved}, learning half-spaces  \cite{kalavasis2024replicable, blondal2025borsuk}, data clustering \cite{esfandiari2023replicable}, high-dimensional statistical estimation \cite{hopkins2024replicability}, and distribution testing \cite{liu2024replicableuniformity}, oftentimes with algorithm design and analysis that is tailored to the specific task at hand.

In this work, we are interested in principled approaches to designing algorithms and hardness results for replicable statistical testers. To reiterate, ensuring replicability when designing algorithms for statistical testing is challenging since the traditional notion of correctness is not sufficient: even if we can successfully distinguish between the cases where $\mathcal{D} \in H_0$ or $\mathcal{D} \in H_1$ (the usual notion of accuracy in statistical testing), we need to be stable even if $\mathcal{D}$ is arbitrary. Towards this, we give general structural results and tools to analyze and design replicable algorithms for hypothesis testing. 

\subsection{Our Contributions}\label{sec:contributions}

\begin{table*}[!ht]
\centering
{\renewcommand{\arraystretch}{5.0}

\begin{tabular}{c|c|c}
\textbf{Problem} & \textbf{Prior Bounds} & \textbf{Our Results} \\ \hline
\makecell{Coin \\ Testing} &
\makecell{$O\p{\frac{\log(1/\delta)}{\eps^2\rho^2}}$ (D), $O\p{\frac{\log(1/\delta)}{\eps^2\rho}}$ (E)~\cite{hopkins2024replicability}  \\
$\Omega\p{\frac{1}{\eps^2 \rho^2}}$ (D),
$\Omega\p{\frac{1}{\eps^2 \rho}}$ (E)~\cite{hopkins2024replicability}}  &
\makecell{$O\p{\frac{\log(1/\delta)}{\eps^2} + \frac{1}{\eps^2 \rho^2}}$ (D) \\$O\p{\frac{\log(1/\delta)}{\eps^2} + \frac{1}{\eps^2 \rho}}$ (E)
} \\ \hline
\makecell{Uniformity \\ Testing} &
\makecell{
$O\p{\frac{\sqrt{n \log(n/\rho)} \log(1/\rho)}{\eps^2\rho} + \frac{\log(1/\rho)}{\eps^2\rho^2}}$ (D)~\cite{liu2024replicableuniformity} \\
$\Omega\p{\frac{\sqrt{n}}{\eps^2\rho \log^2(n)} + \frac{1}{\eps^2\rho^2}}$ (D) ~\cite{liu2024replicableuniformity}$^*$
} &
\makecell{
$O\p{\frac{\sqrt{n}\log(1/\delta)}{\eps^2} + \frac{\sqrt{n}}{\eps \rho} + \frac{1}{\eps^2\rho}}$ (E) \\
$\Omega\p{\frac{\sqrt{n}}{\eps^2\rho \log^2(n)} + \frac{1}{\eps^2\rho^2}}$ (D) \\
$\Omega\p{\frac{\sqrt{n}}{\eps^2\rho} + \frac{1}{\eps^2\rho^2}}$ (D)$^{**}$
}\\ \hline
\makecell{Closeness \\ Testing} &
\makecell{
$O\p{\frac{n^{2/3}}{\eps^{4/3} \rho^2} + \frac{\sqrt{n}}{\eps^2 \rho^2}}$(D)~[Coin]$+$\cite{chan2014optimal} \\
$\Omega\p{\frac{n^{2/3}}{\eps^{4/3}} + \frac{\sqrt{n}}{\eps^2}}$ (D)~\cite{chan2014optimal}
} &
\makecell{
$O\p{\frac{n^{2/3}}{\eps^{4/3} \rho^{2/3}} + \frac{\sqrt{n}}{\eps^2\rho} + \frac{1}{\eps^2\rho^2}}$ (D) \\
$\Omega\p{\frac{n^{2/3}}{\eps^{4/3}\rho^{2/3}} + \frac{\sqrt{n}}{\eps^2\rho \log^2(n)} + \frac{1}{\eps^2\rho^2}}$ (D)
}\\ \hline
\makecell{Gaussian \\ Mean \\ Testing} &
\makecell{
$\tilde{O}\p{\frac{\sqrt{d}}{\alpha^2 \rho^2}}$ (D)~\cite{bun2023stability}$+$\cite{narayanan2022private}$^{***}$ \\
$\Omega\p{\frac{\sqrt{d}}{\alpha^2}}$ (D)~\cite{ingster2003nonparametric}
} &
\makecell{
$\tilde{O}\p{\frac{\sqrt{d}}{\alpha^2\rho} + \frac{\sqrt{d}}{\alpha\rho^2} + \frac{1}{\alpha^2\rho^2}}$ (D) \\
$\Omega\p{\frac{\sqrt{d}}{\alpha^2\rho} + \frac{1}{\alpha^2\rho^2}}$ (D)
} \\ \hline
\makecell{Hypothesis \\ Selection} &
\makecell{
$O\p{\frac{\log^2(n) \log(1/\rho)}{\eps^2 \rho^2}}$ (D)~\cite{bun2023stability}$+$\cite{bun2019privateselect}$^{***}$ \\
$\Omega\p{\frac{\log^2(n)}{\eps^2 \rho^2 \log(1/\eps)}}$ (D)~\cite{hopkins2024replicability}
}&
\makecell{
$O\p{\frac{\log^5(n)}{\eps^2\rho^2}}$ (D) \\
$O\p{\frac{\log^5(n)}{\eps^2\rho}}$ (E)
} \\ 
\end{tabular}

}
\caption{Summary of applications.\protect\footnotemark{}
Parentheticals (D) and (E) are used to indicate deterministic and in-expectation sample complexity, respectively.
All deterministic lower bounds can be translated to in-expectation lower bounds by multiplying by a factor of $\rho$, and all in-expectation upper bounds can be translated to deterministic upper bounds by multiplying by a factor of $1/\rho$ via \cref{prop:markov-translate}.
The single asterisks * indicates that the lower bound from \cite{liu2024replicableuniformity} only holds against symmetric algorithms and not in general.
Our first lower bound result comes from using canonical properties of replicable testers to show that it suffices to only show lower bounds against symmetric algorithms, extending their result to the general setting.
The double asterisks ** indicates that this lower bound only holds when $n \geq \frac{1}{\eps^6\rho^2}$.
The triple asterisks *** indicates that the result is \emph{computationally inefficient}, as it follows from the (inefficient) black-box differential privacy to replicability transformation of \cite{bun2023stability}.
}
\label{tab:results}
\end{table*}

\footnotetext{For all problems, $\delta \leq \rho$ is the failure probability. Some prior works do not give bounds for general $\delta$. In these cases, $\delta$ dependence is omitted and the failure probability is $\rho$.}

We first overview our major contributions and defer the technical summary to \cref{sec:technical_overview}.
The quantitative bounds given as applications of our main results are summarized in \cref{tab:results}.

\paragraph{Contribution 1: A Framework for Characterizing Optimal Replicable Algorithms.}
 Our first contribution explores the question: ``What is the canonical structure of replicable testers?" Informally, our contribution shows that all replicable algorithms for testing discrete distributions can be assumed to be of a specific form. Our description helps an algorithm designer simplify the algorithm design process by imposing rigid conditions on the structure of the algorithms, which we show hold without loss of generality.

 First, we define the rigid properties we impose on our algorithms. For clarity, we give intuitive, informal definitions here and defer the formal definition to Section \ref{sec:canonical-tester}. A tester is called a canonical threshold algorithm if it accepts or rejects based on comparing a statistic to a random threshold drawn uniformly from the interval $[0, 1]$ (\cref{def:canonical}). A tester satisfies permutation-robust replicability if its performance is stable even if the data is sampled again from a permuted distribution (\cref{def:perm-rubust-replicability}). Note that this is a much stronger notion of replicability since the underlying distribution changes. We prove the following. 

\begin{restatable}[Canonical properties of replicable testers]{thm}{canonical}
\label{thm:main_canonical}
  Let $\A(X; r)$ be a $\rho$-replicable algorithm for testing a symmetric property $\PP$ of discrete distributions over $[n]$, using $\ns$ i.i.d.\ samples $X = (X_1, \dots, X_{\ns})$ drawn from an underlying distribution $p$, and randomness $r$. The algorithm outputs a binary decision in $\{\text{\accept}, \text{\reject}\}$ and satisfies:
    \begin{itemize}
        \item If $p \in \PP$, then 
        $$\Prma[X\sim p^{\otimes \ns}, r]{\A(X;r) = \text{\accept}} \geq 1-\delta\,.$$
        \item If $p$ is $\epsilon$-far from $\PP$, then 
        $$\Prma[X\sim p^{\otimes \ns}, r]{\A(X;r) = \text{\reject}} \geq 1-\delta\,.$$
    \end{itemize}
    Then, there exists an algorithm $\A'(X; r)$ that achieves the same accuracy with $\ns$ samples and has the following canonical properties:
    \begin{itemize}
        \item It operates in the canonical format of comparing a deterministic function to a random threshold (\cref{def:canonical}).
        \item It is invariant to both the order and the labels of the samples (\cref{def:sample_invariant}).
        \item It is $\rho$-replicable and satisfies $\rho$-permutation robust replicability (\cref{def:perm-rubust-replicability}).
    \end{itemize}
\end{restatable}

The label invariance and $\rho$-permutation robust replicability only hold for symmetric properties while random thresholding and order invariance hold generally.
Our characterization above is particularly powerful when combined with our following generic tool for proving sample complexity \emph{lower bounds} for replicable testing. 

\begin{restatable}[Chaining lower bound]{thm}{chaininglb}
\label{thm:LB_chaining}
Let $\epsilon\in(0,1]$, $\delta\in(0,1/3]$, and $\rho\in(0,0.001]$ be arbitrary parameters, and let $n$, $k$ be positive integers and $t\leq 1/(300\rho)$ be a positive integer. Also, let $\PP$ be a symmetric property. Consider a collection of $t+1$ distributions over $[n]$, namely $p_0, p_1, \ldots, p_t$, with the following properties:
\begin{itemize}
    \item $p_0$ belongs to $\PP$. That is, any $(\epsilon,\delta)$-tester for $\PP$ must output \accept on $p_0$ with probability at least $1-\delta$.
    \item $p_t$ is $\epsilon$-far from $\PP$. That is, any $(\epsilon,\delta)$-tester for $\PP$ must output \reject on $p_t$ with probability at least $1-\delta$.
    \item There exist $t+1$ priors $\DD_0, \DD_1, \ldots, \DD_t$ with the following properties:
    \begin{itemize}
        \item For every $i$, $\DD_i$ is a prior distribution over $p_i$ and all of its permutations $p^{\pi}$.
        \item For every $i$, if we draw two sample sets of size $k$, namely $X^{(i)} \sim \DD_i$ and $X^{(i-1)} \sim \DD_{i-1}$\footnote{Here we are abusing the notation slightly: by writing $X \sim \DD$, we mean that $X$ is a sample set drawn from a distribution $d$ that is selected according to $\DD$.}, they are statistically close to each other, by which we mean that the total variation distance between the overall distributions of $X^{(i)}$ and $X^{(i-1)}$ is at most $0.5$. 
    \end{itemize}
\end{itemize}
Then, no $\rho$-replicable algorithm exists for $(\epsilon,\delta)$-testing of $\PP$ that uses $k$ samples.
\end{restatable}

In conjunction, these two contributions immediately imply new lower bound results. We show a lower bound for the classic problem of uniformity testing (which implies a lower bound for identity testing). 
\begin{restatable}[Uniformity lower bound]{thm}{uniflb}
\label{thm:LB_unif}
For parameters $\epsilon \in [0,1/4]$, $\rho \leq 0.001$, and $\delta \in (0,1/3]$, suppose $\AA$ is a $\rho$-replicable algorithm that uses $m$ samples drawn from an underlying distribution $p$ over $[n]$, and that, with probability at least $1-\delta$, distinguishes whether $p$ is the uniform distribution over $[n]$ or is $\epsilon$-far from it. Then, it must be:
\begin{equation*}
m = \Tilde{\Omega}\p{\max\bc{
    \frac{\sqrt{n}}{\eps^2 \rho},
    \frac{1}{\eps^2 \rho^2}
}}\,.
\end{equation*}
For a sufficiently large $n \geq (\epsilon^6\,\rho^2)^{-1}$, we can further show that it must be: 
\begin{equation*}
m = \Omega\p{\max\bc{
    \frac{\sqrt{n}}{\eps^2 \rho},
    \frac{1}{\eps^2 \rho^2}
}}\,.
\end{equation*}
\end{restatable}

A lower bound for uniformity testing was given in the previous work of \cite{liu2024replicableuniformity}, but their lower bound only holds for ``label invariant algorithms". Our canonical framework immediately implies that their lower bound holds for general algorithms as well without loss of generality, answering an open question of \cite{liu2024replicableuniformity}. Furthermore, our chaining lower bound framework circumvents several technical calculations for understanding the Lipschitz continuity of the acceptance probability of a tester, leading to an arguably simpler lower bound proof overall (see 
\cref{sec:unif_lb} for details). We also improve their lower bound by a logarithmic factor when $n$ is very large. 

In addition, we also give a sample complexity lower bound for the harder testing problem of replicable closeness testing (see \cref{def:replicable_closeness}. Note we also give almost matching upper bound described in the next section). This answers another open question raised in  \cite{liu2024replicableuniformity}. 

\begin{restatable}[Closeness testing lower bound]{thm}{closelb}
\label{thm:LB_close}
Assume that $\eps < 0.99$.
Any $\rho$-replicable 0-vs-$\eps$ closeness tester has sample complexity at least
\begin{equation*}
\Omega\p{\max\bc{
    \frac{n^{2/3}}{\eps^{4/3}\rho^{2/3}},\ \Tilde{\Theta}\left(\frac{\sqrt{n}}{\eps^2 \rho}\right),\  
    \frac{1}{\eps^2 \rho^2}
}}.
\end{equation*}
\end{restatable}

Lastly, we re-derive a known sample complexity lower bounds for replicable coin testing (which was originally proved in \cite{impagliazzo2022reproducibility}). While this is not a new result, we believe it further demonstrates the generality of our framework.

\paragraph{Contribution 2: A Framework for Designing Replicable Testers from Non-Replicable Testers}

Many statistical testers, both with and without replicability, have a ``expectation-gap'' structure.
In these algorithms, a real-valued statistic $Z$ is computed from the set of samples and the algorithm outputs \accept or \reject depending on which side of a threshold the empirical estimate of $Z$ falls on.
Analysis of the success of such an algorithm depends on (a) upper bounding $\Ema{Z}$ under the null hypothesis, (b) lower bounding $\Ema{Z}$ under the null hypothesis, and (c) bounding the variance of $Z$.

Many existing replicable estimators (see, e.g., \cite{impagliazzo2022reproducibility, hopkins2024replicability, liu2024replicableuniformity} adopt this framework, taking more samples and picking a random threshold to ensure replicability (see a detailed explanation in \cref{sec:technical_overview}).
Importantly, this is exactly the analysis framework for the replicable coin testing problem, where the hypothesis testing problem is to distinguish between samples from $\Ber(p)$ where $p=p_0$ or $p \geq q_0 = p_0 + \eps$.
By a standard argument, replicable coin testing as a black box can turn any non-replicable algorithm into a replicable one~\cite{impagliazzo2022reproducibility, hopkins2024replicability}, leading to a multiplicative overhead of $\frac{\log(1/\rho)}{\rho^2}$ in the samples needed for replicability.

While powerful, this approach is lossy in two ways. First, there is a $\log(1/\rho)$ gap between the best upper and lower bounds for replicable coin testing~\cite{impagliazzo2022reproducibility, hopkins2024replicability}.
Second, and more importantly, this black box approach does not make use of application-specific analysis of the estimator.
For many hypothesis testing problems (such as uniformity testing or closeness testing),  existing works have developed a sharp understanding of the expectation and variance of the test statistics.
In a recent work on replicable uniformity testing~\cite{liu2024replicableuniformity}, the authors show that, for large domain sizes, only a $1/\rho$ dependence is needed in the sample complexity by carefully adapting analysis of a known statistic to the replicable setting.

We develop general purpose estimators for expectation-gap statistics which \textbf{quantitatively improve over existing algorithms} and \textbf{makes it simple to port existing analyses from non-replicable setting}.

As applications of our framework, we get \emph{optimal bounds for replicable coin testing}, in both the expected number of samples needed as well as the worst-case sample complexity, up to constant factors in all terms. As a comparison, the prior state-of-the-art algorithm was given in \cite{hopkins2024replicability} which obtains a sample complexity of $O\p{\frac{q_0 \log(1/\delta)}{\eps^2 \rho}}$ in expectation and $O\p{\frac{q_0 \log(1/\delta)}{\eps^2 \rho^2}}$, where $\delta$ is the failure probability of the tester (see Theorem \ref{thm:prior_coin} for a formal statement of their guarantees). 

In contrast, we obtain the following:
\begin{theorem}[Informal; see \cref{thm:optimal-coin-testing}]
  There exists a $\rho$-replicable coin testing algorithm which succeeds with probability $1-\delta$ for any $\delta \leq \rho$ and uses $O\p{\frac{q_0}{\eps^2 \rho} + \frac{q_0 \log(1/\delta)}{\eps^2}}$ samples in expectation and $O\p{\frac{q_0}{\eps^2 \rho^2} + \frac{q_0 \log(1/\delta)}{\eps^2}}$
samples in the worst-case. All terms are necessary up to constant factors.
\end{theorem}
Notably, our tight dependency separates the replicable parameter $\rho$ with the parameter controlling the failure probability $\log(1/\delta)$.
As coin testing allows for black-box replicability of non-replicable algorithms, this immediately implies improved replicable algorithms in myriad applications.
We use our new result to design an algorithm for replicable hypothesis selection via a reduction to coin testing, see \cref{subsec:hypothesis-selection-overview} for details.

Adopting analyses in the non-replicable setting for uniformity testing~\cite{diakonikolas2019collisiontesters} and closeness testing~\cite{chan2014optimal} into our framework, we immediately get the following bounds for the replicable versions of those problems.
\begin{theorem}[Informal; see \cref{thm:uniformity-testing}]
There exists a $\rho$-replicable uniformity testing algorithm which succeeds with probability $1-\delta$ for any $\delta \leq \rho$ and uses $
O\p{\frac{\sqrt{n}\log(1/\delta)}{\eps^2} + \frac{\sqrt{n}}{\eps \rho} + \frac{1}{\eps^2 \rho}}
$
samples in expectation and \\
$
O\p{\frac{\sqrt{n}\log(1/\delta)}{\eps^2}
+ \frac{\sqrt{n}}{\eps^2 \rho} +\frac{\sqrt{n}}{\eps \rho^2} + \frac{1}{\eps^2 \rho^2}}
$
samples in the worst-case.
\end{theorem}
The prior work on replicable uniformity testing from \cite{liu2024replicableuniformity} requires $O\p{\frac{\sqrt{n}\log(1/\rho)\sqrt{\log(n/\rho)}}{\eps^2\rho} + \frac{\log(1/\rho)}{\eps^2\rho^2}}$ samples and succeeds with probability $1-\rho$.
Our worst-case sampling bound is worse than this bound in some parameter regime due to the extra $\sqrt{n}/\eps\rho^2$ factor but otherwise improves upon this work.
Our in-expectation sampling bound is always superior, and there are no non-trivial prior bounds for sample complexity in expectation to our knowledge.
Intriguingly, our in-expectation bound implies that replicability comes \emph{for free} if $n$ is sufficently large and $\rho \ll \frac{\eps}{\log(1/\delta)}$.

\begin{theorem}[Informal; see \cref{thm:closeness-upperbound}]
For any constant $C$, there exists a $\rho$-replicable closeness testing algorithm which succeeds with probability $1-\rho^C$ and uses $O\p{\frac{n^{2/3}}{\eps^{4/3}\rho^{2/3}} + \frac{\sqrt{n}}{\eps^2 \rho} + \frac{1}{\eps^2 \rho^2}}$ samples in the worst-case. 
\end{theorem}
To our knowledge, there has been no prior work on replicable closeness testing. Our bounds are \emph{optimal}, they match our lower bound in all terms up to constant factors.

\paragraph{Contribution 3: High-Dimensional Gaussian Mean Testing.} 
In the previous two contributions, we mainly focused our attention to the study of discrete distributions over a finite domain. For the third contribution, we turn our attention to high-dimensional distributions and focus on one of the most classical high-dimensional testing problems: Given sample access to a distribution $\mathcal{D}$ over $\R^d$ and a parameter $\alpha$, design a $\rho$-replicable algorithm which accepts if $\mathcal{D} = \mathcal{N}(0, I)$ and rejects if $\mathcal{D} = \mathcal{N}(\mu, I)$  for any $\mu$ satisfying $\|\mu\| \ge \alpha$. Note that we require replicability even if $\mathcal{D}$ is not among the null or alternate hypotheses and is an arbitrary high-dimensional distribution.  

Without replicability, rejecting if and only if the norm of the empirical mean exceeds a fixed threshold is an efficient and information-theoretical optimal algorithm, requiring $\Theta(\sqrt{d}/\alpha^2)$ samples~\cite{SrivastavaD08}. Under replicability, this problem was previously studied in \cite{bun2023stability}, which gave an inefficient algorithm (not running in polynomial time) with sample complexity $\tilde{O}\left( \frac{\sqrt{d}}{\rho^2 \alpha^2}\right)$ by appealing to a differentially private (DP) mean testing algorithm of \cite{narayanan2022private} and using reductions between DP and replicability. (We remark, however, that by using replicable coin testing to turn any non-replicable algorithm into a replicable one~\cite{impagliazzo2022reproducibility, hopkins2024replicability} and the non-replicable algorithm of~\cite{SrivastavaD08}, one can obtain the $\tilde{O}\left(\frac{\sqrt{d}}{\rho^2 \alpha^2}\right)$ sample complexity efficiently.)

Our algorithm obtains the following improved guarantees.

\begin{restatable}[Replicable Gaussian Mean Testing]{thm}{repgaussian}
\label{thm:gaussian-testing-main}
Let $\mathcal{D}$ be a distribution over $\R^d$ which we have sample access to, and fix parameters $\alpha \in (0, 1]$\footnote{While our algorithm can extend to $\alpha > 1$, we do not focus on this setting as our bound is smaller only when $\alpha \le 1$.} and $\rho \in (0, 1)$. There exists a polynomial-time algorithm $\mathcal{A}$ taking $\ns = \tilde{O}\left(\frac{\sqrt{d}}{\alpha^2 \rho} + \frac{\sqrt{d}}{\alpha \rho^2} + \frac{1}{\alpha^2 \rho^2}\right)$ samples from $\mathcal{D}$ which satisfies the following properties:
\begin{itemize}
    \item $\mathcal{A}$ is $\rho$-replicable.
    \item If $\mathcal{D} = \mathcal{N}(0, I)$ then $\mathcal{A}$ outputs \text{\accept} with probability at least $0.99$.
    \item If $\mathcal{D} = \mathcal{N}(\mu, I)$ for any $\mu$ satisfying $\|\mu\| \ge \alpha$, then $\mathcal{A}$ outputs \text{\reject} with probability at least $0.99$.
\end{itemize}
\end{restatable}
Note that Theorem \ref{thm:gaussian-testing-main} improves upon the guarantees of \cite{bun2023stability} in two distinct ways. First, our algorithm runs in polynomial time (the reduction from replicability to DP used in \cite{bun2023stability} is inherently inefficient). Secondly and perhaps more importantly, we decouple the $1/\rho^2$ term from the standard $\sqrt{d}/\alpha^2$ term, showing that a $1/\rho^2$ overhead is not needed for this fundamental problem.

We also prove the first nontrivial lower bound on the sample complexity on any replicable Gaussian mean testing algorithm. Our upper and lower bounds do not match (as the upper bound has an additional $\frac{\sqrt{d}}{\alpha \rho^2}$ additive term), and a natural open question is to resolve this gap.

\begin{restatable}[Replicable Gaussian Mean Testing Lower Bound]{thm}{repgaussianlb}
\label{thm:gaussian-testing-lb}
Let $\mathcal{A}$ be a $\rho$-replicable algorithm that distinguishes between samples from $\NN(0, I)$ and $\NN(\mu, I)$ for any $\|\mu\| \ge \alpha$. (I.e., it satisfies the three guarantees in \Cref{thm:gaussian-testing-main}). Then, $\A$ must use $\ns = \Omega\left(\frac{\sqrt{d}}{\alpha^2 \rho} + \frac{1}{\alpha^2 \rho^2}\right)$ samples in the worst case.
\end{restatable}

We remark that our lower bound builds on the techniques we establish in this paper, such as our reduction to the canonical properties of replicable testers and our chaining lower bound tool, though we will require some modifications (see \Cref{sec:technical_overview} for more details).

\paragraph{Concurrent Work}
In independent and concurrent work, Diakonikolas, Gao, Kane, Liu, and Ye \cite{diakonikolas2025replicable} also study replicable hypothesis testing.
They show similar lower bounds for uniformity testing (with no limitation to symmetric algorithms as in \cite{liu2024replicableuniformity}) and upper and lower bounds for closeness testing.
Our bounds are quantitatively tighter  by logarithmic factors in some settings and never worse.
Intriguingly, they prove these lower bounds without showing that there exist canonical permutation-robust/label-invariant replicable algorithms for symmetric properties, but rather their proof directly deals with asymmetric algorithms.
They also provide upper bounds for replicable independence testing which we do not study in this work.
They do not have results on canonical properties of replicable testers, general tools for making existing testers replicable, or results for Gaussian mean testing or hypothesis selection.

\section{Technical Overview}\label{sec:technical_overview}

\subsection{Canonical Properties of a Replicable Tester}

Replicability measures the stability of an algorithm with respect to each individual distribution. In contrast, standard techniques for proving lower bounds in distribution testing typically operate over families of distributions, creating a misalignment that complicates lower bound proofs for replicable algorithms. To bridge this gap, we impose well-structured properties on our algorithms, enabling rigorous analysis from both upper-bound and lower-bound perspectives.

To this end, we show
that the existence of a replicable testing algorithm implies the existence of another replicable algorithm for the same problem with a well-defined structural form. These structural assumptions will prove essential for deriving lower bounds later. In particular, we identify the following properties:

\begin{itemize}
    \item \textbf{Canonical random threshold algorithm:}  
    If a replicable algorithm exists, then without loss of generality we may assume it computes a deterministic function of its input, 
    $f: \mathcal{X}^n \rightarrow [0,1],$
    and compares the value of this function to a random variable \( r \) drawn uniformly from \([0,1]\). In this setting, the function \( f(X) \) is determined by the probability that \(\A(X; r')\) outputs \(\text{\accept}\) over the random choices of \( r' \). This crucial observation helps us fully separate the notion of randomness from the input sample set.

    \item \textbf{Sample order invariant algorithm:}  
    If a replicable algorithm exists for distribution testing, then there also exists an equivalent replicable algorithm whose output remains invariant under any permutation of the input samples, while maintaining the same performance.

    \item \textbf{Symmetric property and sample label invariance:} We show that for symmetric properties---where membership and distance remain unchanged under any permutation---any replicable algorithm can be assumed to be label invariant. Specifically, if a replicable algorithm exists for testing a symmetric property (e.g., uniformity or closeness testing), then there exists one whose output does not depend on the sample labeling while maintaining the same performance. The main idea to prove this fact is to use the canonical deterministic function 
$f$: the performance on a sample $X$ 
can be equated to the average performance over all permutations of 
$X,$ 
thereby eliminating label dependence.

    \item \textbf{Permutation-Robust Replicability:}  
    We show that the label-invariant algorithm (in the previous property) satisfies an even stronger replicability condition. Specifically, its outcome remains stable even if the underlying distribution is replaced by another obtained by permuting the labels. An algorithm is said to satisfy \(\rho\)-\emph{permutation robust replicability} if, for any prior distribution $\DD$ over a given distribution and all of its permutations, we have:
    $$\Prma[{r\sim\Unif[0,1]},\ p, p^\pi \sim \DD, X\sim p^{\otimes \ns}, X'\sim (p^\pi)^{\otimes \ns} ]{\A_3(X;r) \not = \A_3(X';r)} \leq \rho\,,$$
    where $p^\pi \coloneqq p \circ \pi^{-1}$.
    This last property is one of the key observations that brings us closer to the standard setting for proving lower bounds for replicable testers. The replicability assumption works not only on a single distribution \(p\), but also on any prior over \(p\) and its permutations.

\end{itemize}

\subsection{Lower Bounds via Chaining}\label{sec:tech_chaining}
We give a technical overview of our Theorem \ref{thm:LB_chaining}. Intuitively, our approach boosts standard indistinguishability lower bounds to replicable ones, which can be then applied to various downstream problems. 

First we review standard hypothesis testing, without replicablity. In standard distribution testing, one typically considers two distributions, $p^+$ and $p^-$, with the promise that the tester can distinguish between them. If $p^+$ and $p^-$ appear very similar based on $k$ samples, then no tester using $k$ samples can reliably distinguish them. We enhance this argument as follows.

Suppose there are $t = \Theta(1/\rho)$ distributions $p_1, p_2, \ldots, p_t$, where $p_0$ corresponds to $p^+$ and $p_t$ corresponds to $p^-$. While a standard tester only needs to distinguish between $p_0$ and $p_t$, we show that any replicable algorithm must distinguish between {\em some} consecutive pair $p_i$ and $p_{i+1}$. This observation implies that if one can construct a \emph{chain} of indistinguishable pairs, then an impossibility result for replicability immediately follows. The improvement in the sample complexity lower bound is a result of packing $t$ distributions much closer together (e.g., between $p_0$ and $p_1$), which significancy increases the difficulty of the problem.

Another challenge in converting standard lower bounds to replicable ones is that standard settings often provide lower bounds for distinguishing between two families of distributions rather than two individual distributions. This poses a problem because the replicability guarantee does not extend across different distributions. For instance, if given datasets $X \sim \DD$ and $X' \sim \DD'$ where $\DD, \DD'$ are part of the same family of distributions, the replicability of an algorithm $\AA$ does not state any relationship between $\AA(X)$ and $\AA(X')$. Therefore, one cannot directly replace the individual $p_i$'s with families drawn from the priors $\DD_i$'s.

We overcome this difficulty by leveraging our canonical characterization of a replicable tester given in \cref{thm:main_canonical}. For symmetric properties, we show that if a $\rho$-replicable algorithm exists, then there is one that is also $\rho$-permutation robust replicable. The label invariant property ensures that the replicability guarantee holds even when the second sample set is drawn from a distribution that is a permutation of the first. Consequently, if the priors $\DD_i$ are supported on $p_i$ and all its permutations, our lower bound theorem is applicable, affording us a significant advantage.

Most importantly, this approach reframes the problem in terms of statistical indistinguishability and distribution packing,
allowing us to apply established techniques for non-replicable distribution testing lower bounds.
Consequently, our lower bound immediately implies several key lower bounds for uniformity, identity, and closeness testing (see Section~\ref{sec:chaining_applications}).

\paragraph{Sketch of Proof of \cref{thm:LB_chaining} :} Here, we use yet another property of the canonical tester. Let $h(X)$ denote the deterministic function used by our canonical tester to determine its output (by comparing it to a random threshold $r\in[0,1]$). We prove that for any $X\sim\DD_i$, the value $h(X)$ is highly concentrated in an interval $\II_i$ of length $O(\rho)$—a direct consequence of $\rho$-permutation robust replicability.

Using the accuracy assumptions, we argue that $\II_0$ is centered near $1/3$, while $\II_t$ is centered near $2/3$. Moreover, the indistinguishability between $\DD_i$ and $\DD_{i-1}$ forces the intervals $\II_i$ and $\II_{i-1}$ to overlap; otherwise, membership in $\II_i$ would serve as an effective distinguisher. 

Thus, the sequence of intervals $\II_0, \II_1, \ldots, \II_t$ must cover an interval of constant length while each interval has length $O(\rho)$ and overlaps with its neighbors. However, if $t\ll\Theta(1/\rho)$, such a covering is impossible, leading to the desired lower bound. The formal details are given in \cref{sec:chaining}.
\subsection{Generalizing and Improving Expectation-Gap Estimators}

We systematize and quantitatively improve a ubiquitous strategy for designing replicable testing algorithms from non-replicable testing algorithms. As direct applications of our general estimation technique, we get the first constant-factor optimal bounds for replicable coin testing as well as new and improved bounds for replicable uniformity testing and replicable closeness testing with simpler analyses.

A classic strategy in hypothesis testing is the so-called expectation-gap approach, which we now describe.
Given $m$ samples from a distribution, consider a one-dimensional test statistic $Z(m) \in \R$.
Two thresholds are defined: $\tau_0(m) \geq \Ema{Z(m) | H_0}$ upper bounds the expectation under any distribution belonging to the null hypothesis, and $\tau_1(m) \leq \Ema{Z(m) | H_1}$ lower bounds the expectation under any distribution belonging to the alternate hypothesis.
Finally, let $\sigma(m)$ be an upper bound on the standard deviation of $Z(m)$.\footnote{We actually only require $\sigma(m)$ to upper bound the standard deviation when $\Ema{Z(m)} \in [\tau_0(m), \tau_1(m)]$ and allow it to smoothly degrade as the expectation moves away from this interval.}
For simplicity of notation, we drop the argument $m$ from $Z, \tau_0, \tau_1, \sigma$ for the remainder of this subsection.

Given these parameters, take enough samples such that $\Delta := \tau_1 - \tau_0 > 0$ and $\sigma \leq \Delta/4$.
Then, Chebyshev's bound implies that the algorithm which thresholds at the midpoint $\tau_0 + \Delta/2$, outputting \accept if the empirical statistic $Z$ is below the threshold and \reject otherwise, is correct with probability $3/4$.

Classic examples of this framework are testers based on the empirical mean (e.g., for coin testing problems), collision statistics for uniformity testing, $\chi^2$ statistics, among many others (we refer to the surveys \cite{canonne2020survey, Canonne_book_22}).

The design of replicable algorithms often makes use of the same framework \cite{impagliazzo2022reproducibility, hopkins2024replicability, liu2024replicableuniformity, eaton2023replicable, esfandiari2023replicable, esfandiari2023replicable_clustering, ahmadi2024replicable, larsen2025improved, kalavasis2024replicable}.
The key difference is that rather than thresholding at the midpoint of the interval, the threshold is chosen as $\tau_0 + r\Delta$ for $r \sim \Unif\br{\frac{1}{4}, \frac{3}{4}}$ (clearly, correctness still holds as long as twice the number of samples are taken).
Enough samples are taken so that the standard deviation is a $\rho$ fraction of the interval: $\sigma \leq \rho \Delta$.
Repeating $\log(1/\rho)$ times and taking the median $V$ (or simply taking $\log(1/\rho)$ times more samples for well-concentrated statistics), a $\rho \Delta$-sized interval around the expectation $\Ema{Z}$ contains the median with probability $1-\rho$. 
The algorithm will only fail to be replicable if, over the randomness of $r$, two separate samples lead to estimates $V_1$ and $V_2$ which are on opposite sides of the threshold $\tau_0 + r \Delta$.
Assuming that both of the median estimates are contained in the $\rho \Delta$-sized interval, which happens with probability at least $1-2\rho$, the algorithm will be replicable as long as $\tau_0 + r \Delta$ does not fall in this interval,
which also occurs with probability $1-O(\rho)$. Overall, this implies $O(\rho)$-replicability.
Without additional structure about the estimator, this algorithm has a multiplicative $O\p{\frac{\log(1/\rho)}{\rho^2}}$ overhead on the sample complexity required to solve the non-replicable problem. This approach is used in many of the aforementioned works.

We generalize and improve upon this approach in four ways:
\begin{itemize}
    \item \textbf{Worst-Case Sample Complexity:} We show that the $\log(1/\rho)$ overhead described above is unnecessary.
    At the point where the statistic has standard deviation $\sigma = O(\rho \Delta)$, it suffices to simply compare the statistic $Z$ with the threshold $\tau_0 + r \Delta$ to decide whether to \accept or \reject.
    The median step used in prior work is not needed.
    This leads to black-box replicability overhead of $O\p{1/\rho^2}$ samples, removing a $\log(1/\rho)$ factor.
    
    The key observation is to define a ``replicably correct'' answer depending on which side of the threshold $\tau_0 + r\Delta$ contains the expectation $\Ema{Z}$ (this concept appears in prior works on replicability \cite{impagliazzo2022reproducibility, hopkins2024replicability}). 
    Importantly, this version of correctness is defined even if the distribution belongs neither to the null or alternate hypothesis.
    If an algorithm $\A$ is replicably correct with probability $1-\rho$ over the randomness in $r$ and the sample, then it is easy to see that it is $O(\rho)$-replicable as correctness is consistently defined for any fixed $r$.

    Given this concept, the algorithm fails to be replicable if the empirical statistic $Z$ and the expectation $\Ema{Z}$ land on opposite sides of the threshold.
    We consider geometrically increasing buckets of size $\{\rho, 2\rho, 4\rho, \ldots, 1/4\}$ corresponding to the magnitude of the distance $\abs{\frac{\Ema{Z} - \tau_0}{\Delta}-r}$.
    The probability of $r$ landing in a bucket of size $2^{-i}$ is $O\p{2^{-i}}$ as $r$ is chosen uniformly from a constant-sized interval.
    On the other hand, Chebyshev's bound implies that the event of $Z$ deviating from its expectation by more than $2^{-i-1}$ occurs with probability at most $O\p{\rho^2 2^{2i}}$.
    Roughly, the probability of failing at a given level is $O\p{\rho^2 2^i}$.
    This geometric sum over all levels converges to $O(\rho)$, as required. As mentioned in \cref{sec:contributions}, this idea already leads to the right worst-case upper bounds for the replicable coin testing, uniformity, and closeness testing.
    
    \item \textbf{Expected Sample Complexity:} The idea of considering geometric levels of the distance $\abs{\frac{\Ema{Z} - \tau_0}{\Delta}-r}$ appeared previously in \cite{hopkins2024replicability} to show that, surprisingly, the overhead of replicablity can be significantly improved if we are concerned with expected sample size rather than worst-case sample size.
    Specifically, they show that the quadratic worst-case sampling overhead can be improved to $O\p{\frac{\log(1/\rho)}{\rho}}$ for the expected sample size.

    We generalize their argument beyond coin testing and improve it using the same underlying techniques as in our analysis for worst-case sample size.
    For the in-expectation results, we focus on a specific kind of expectation-gap statistic we refer to as size-invariant.
    As one example, this class of statistics captures the empirical mean as used in the coin testing problem studied in prior work~\cite{hopkins2024replicability} and thus apply in a black-box fashion to give replicability results.
    Our analysis improves the overhead to $O\p{1/\rho}$ in expectation, again removing a $\log(1/\rho)$ factor from prior work.

    Size-invariant statistics are those which are normalized such that $\tau_0$, $\tau_1$, and $\Ema{Z}$ are constant functions in $m$---the expectation of the statistic does not vary with the number of samples.
    This holds for several natural statistics such as the empirical mean or collision probability, but does not necessary hold for other statistics such as $\chi^2$-statistics.

    As the expectation $z = \Ema{Z}$ is independent of the sample size, we can define a ``replicably correct'' answer as above which \emph{does not depend on the sample size}.
    The algorithm then proceeds by taking geometrically increasing number of samples.
    In the case of coin testing, for a given level $i \in [\ceil{\log(1/\rho)}]$, the sampling overhead is $O\p{2^i \p{\log(1/\rho) - i + 1}}$.
    The algorithm only terminates if the estimate is roughly $2^{-i}$ far from the threshold. As we will see, a similar argument can be implies for other problems. e.g., in to uniformity testing when the domain is large.
    
    The structure of the replicability analysis is similar to the worst-case sample analysis.
    We show that at a given sampling level, the probability of failing replicably by the statistic deviating too far from the expectation is exponentially small in $\Omega(\log(1/\rho) - i + 1)$: at the first level, the probability failing is $\poly(\rho)$, and at the final level, the probability of failing is constant.
    On the other hand, the probability of reaching level $i$ (i.e., not terminating before round $i$) is $O\p{2^{-i}}$ as $r$ is chosen uniformly and the algorithm only does not terminate if it is roughly within $2^{-i}$ of the threshold defined by $r$.
    Combined, the probability of failing at a given level is a geometrically increasing sequence and is dominated by the final level, where the probability of deviation is constant, but the probability of ever reaching the level is at most $\rho$.
    
    \item \textbf{High Probability of Correctness:} Independent of replicability, it may be desirable that the algorithm returns $\accept$ under the null hypothesis and $\reject$ under the alternate hypothesis with probability $1-\delta$ for some $\delta > 0$.
    Many $\rho$-replicable algorithms are automatically guaranteed to also be correct with probability $1-\rho$.
    In prior works, to achieve high probability guarantees for $\delta < \rho$, the sample complexity has to be multiplied by an additional factor of $O\p{\frac{\log(1/\delta)}{\log(1/\rho)}}$.

    We show that algorithms in this framework already succeed with probability much greater than $1-\rho$ with at most a constant-factor more samples.
    We achieve this by separating the algorithmic steps and analysis which guarantee the correctness of the algorithm from the steps which are required for the algorithm to be replicable.
    
    Our algorithm for general expectation-gap statistics succeeds with probability at least $1-\rho^C$ for any constant $C$.
    For some statistics, such as the empirical mean for coin testing, the success probability is at least $1 - \exp\p{-\Omega(1/\rho^2)}$.

    In the case of size-invariant expectation-gap statistics, $\delta$ can be specified to the algorithm with an additive sample complexity term that depends on $\log(1/\delta)$ and not at all on $\rho$.
    This means that if the multiplicative overhead for replicability is $c(\rho)$, then the algorithm succeeds with probability $1-\exp(-c(\rho))$ by doubling the sample size.
    
    \item \textbf{General Framework:} We formalize this general strategy, so that after specifying valid $Z$, $\tau_0$, $\tau_1$, and $\sigma$ for any particular estimator, one can immediately get an algorithm which is provably replicable and correct with high probability and has tight sample complexity bounds (for this style of analysis).
    Furthermore, these bounds improve with better analysis of the size of the expectation gap or of the variance.
    
    Using expectation-gap statistics and their analyses which are folklore or appear in prior literature on non-replicable testing, we immediately get the following results:
    \begin{enumerate}[label=(\alph*)]
        \item Optimal worst-case and in-expectation sampling bounds for replicable coin testing (and thus for black-box replicable testing).
        \item Improved sampling bounds in the worst-case in some regimes and state-of-the-art in-expectation sampling bounds for replicable uniformity testing.
        \item Near-optimal worst-case sampling bounds for replicable closeness testing. These are the first non-trivial\footnote{To our knowledge, the only known prior bounds for this problem come from the black-box strategy which has multiplicative overhead $O\p{\frac{\log(1/\rho)}{\rho^2}}$ from prior work and $O\p{\frac{1}{\rho^2}}$ with our new bounds.} sampling bounds for replicable closeness testing and match our lower bound up to constant or logarithmic factors depending on the parameter regime.
    \end{enumerate}
\end{itemize}

The details of these results are in \cref{sec:expectation_gap_tester}.

\subsection{Gaussian Mean Testing}

For fixed parameters $\alpha > 0$, $\rho \in (0, 1)$, and $d \in \N$, we recall that a $\rho$-replicable Gaussian mean testing algorithm in $\R^d$ is a $\rho$-replicable algorithm that accepts with at least $0.99$ probability if given $\ns$ i.i.d. samples from $\NN(0, I)$ and rejects with at least $0.99$ probability if given $\ns$ i.i.d. samples from $\NN(\mu, I)$ for any $\|\mu\| \ge \alpha$.

We give an overview of our algorithm (\Cref{thm:gaussian-testing-main}) that uses only $\ns = \tilde{O}\left(\frac{\sqrt{d}}{\alpha^2 \rho} + \frac{\sqrt{d}}{\alpha \rho^2} + \frac{1}{\alpha^2 \rho^2}\right)$ samples, followed by an overview of our lower bound (\Cref{thm:gaussian-testing-lb}) showing $\ns = \tilde{O}\left(\frac{\sqrt{d}}{\alpha^2 \rho} + \frac{1}{\alpha^2 \rho^2}\right)$ samples are necessary.

\paragraph{Algorithm:}
A first attempt at a replicable Gaussian mean testing algorithm that ``almost works'' is to directly combine the known algorithm for Gaussian mean testing in the non-replicable setting with the expectation-gap estimator. To explain this idea further, we recall the non-replicable algorithm~\cite{SrivastavaD08} for Gaussian mean testing. Given samples $X_1, \dots, X_{\ns}$, the algorithm simply computes the statistic $T = \|\sum X_i \|^2 - s \cdot d,$ and accepts as long as the statistic is below a certain threshold. To see why this works, first note that for any distribution $\cD$ with mean $\mu$, the expectation of $T = \left\|\sum_{i=1}^{\ns} X_i \right\|^2 - s \cdot d$, where $X_1, \dots, X_{\ns} \overset{i.i.d.}{\sim} \cD$, equals $\ns^2 \cdot \|\mu\|^2.$ This can be verified by writing $\left\|\sum_{i=1}^{\ns} X_i \right\|^2 = \left\langle \sum_{i=1}^{\ns} X_i, \sum_{i=1}^{\ns} X_i \right\rangle = \sum_{i, j=1}^{\ns} \langle X_i, X_j\rangle$, and for $X_i, X_j \sim \cN(\mu, I)$, $\Ema[]{\langle X_i, X_j \rangle} = \langle \Ema[X_i \sim \cN(\mu, I)]{X_i}, \Ema[X_j \sim \cN(\mu, I)]{X_j} \rangle = \|\mu\|^2$ and $\Ema[X_i \sim \cN(\mu, I)]{\|X_i\|^2} = d + \|\mu\|^2$. Also, the variance of $T$ can be effectively bounded if $\cD$ is a Gaussian with identity covariance. Hence, as long as the standard deviation is much smaller than the discrepancy $\ns^2 \cdot \alpha^2$ of the mean between the null and alternative hypotheses, we can use Chebyshev's inequality.

In the replicable setting, based on the canonical tester, we can sample a random seed $r \sim \Unif([0, 1])$, and accept if $r \le \frac{T}{\ns^2 \alpha^2}$. The idea is that, if the statistic has standard deviation $\ns^2 \alpha^2 \cdot \rho$, then for probability that $r$ was below the threshold for $X_1, \dots, X_{\ns}$ but above the threshold for a fresh set of samples $X_1', \dots, X_{\ns}'$ (or vice versa) is at most $\rho$. So, we just need to make sure that the variance of $\|X_1 + \cdots + X_{\ns}\|^2$ is at most $\rho \cdot \alpha^2 \ns^2$ for Gaussians. A similar calculation to that done for the non-replicable algorithm~\cite{SrivastavaD08} will tell us that $\ns = O\left(\frac{\sqrt{d}}{\alpha^2 \rho} + \frac{1}{\alpha^2 \rho^2}\right)$ samples suffice.

Unfortunately, there is one major problem with this approach, which is that the algorithm must be replicable for \emph{any} distribution $\cD$ over $\R^d$, not just Gaussian distributions. Some of these distributions may have terrible variance for the statistic described. To fix this, we design a replicable tester that will reject certain families of ``bad'' distributions, and then apply the thresholding algorithm.

The main ways that a distribution can be bad are the following: either the distribution has large covariance in one or a few directions (as opposed to a spherical Gaussian which has evenly spread out covariance), or the distribution has a high probability of two sampled data points $X, Y$ having large inner product. We prove that, as long as these do not hold, the variance of $\left\|\sum_{i=1}^{\ns} X_i \right\|^2$ is small. Specifically, we will design a replicable algorithm that eliminates \emph{bad} distributions, i.e., distributions $\cD$ with either of the following properties.
\begin{enumerate}
    \item The operator norm of $\Ema[X_i \sim \cD]{X_iX_i^\top}$ is too large, i.e., for some direction the ``variance'' (where we do not subtract the mean) is too large.
    \item If we sample $2 \cdot \ns$ data points $X_1, \dots, X_{\ns}$ and $Y_1, \dots, Y_{\ns}$ from $\cD$, with reasonable probability there is a reasonably large bipartite matching where for each edge $(i, j)$ the inner product $\langle X_i, Y_j \rangle$ is much larger than $\tilde{O}(\sqrt{d})$ (which is expected for random Gaussian vectors). 
\end{enumerate}
To eliminate distributions satisfying either property, we again apply the thresholding technique, as we initially tried to test Gaussians. But now, we can in fact successfully accomplish this for any distribution. For the first property, we use the Matrix Chernoff inequality to show that the operator norm of the empirical covariance $\frac{1}{\ns} \sum X_i X_i^\top$ concentrates for any distribution. For the second property, we use a general concentration inequality of~\cite{boucheron2000sharp}, which can be used to establish the concentration of the maximum bipartite matching size when one side is fixed and the other side is random. In our setting, both $X_1, \dots, X_{\ns}, Y_1, \dots, Y_{\ns}$ are random, but this is not a problem, as we can fix one side and resample the other, and then switch sides and repeat. Overall, we prove that the maximum matching size and operator norm of the covariance concentrate well for \emph{arbitrary} distributions, so we can create a replicable algorithm to remove all bad distributions.

Unfortunately, this method does not suffice to reject all distributions where the variance of $T = \|\sum_{i=1}^{\ns} X_i \|^2$ is more than would be expected from a Gaussian. As a result, our sample complexity increases by an additional additive factor of $\frac{\sqrt{d}}{\alpha \rho^2}$. Yet, we still improve over the previous bound of $\frac{\sqrt{d}}{\alpha^2 \rho^2}$~\cite{bun2023stability}. A natural open question is whether one can remove this final additive factor from the upper bound.

\paragraph{Lower bound:} We recall that one difficulty in the algorithm is that the replicable algorithm must be replicable regardless of whether the distribution is $\cN(\mu, I)$ for some $\mu$, or some totally arbitrary distribution. In the lower bound case, it suffices to show a lower bound against the weaker class of algorithms that are only replicable when given i.i.d. samples from $\cN(\mu, I)$.

For such ``weakly replicable'' algorithms, we are able to generalize the canonical lower bound properties such as the symmetric and permutation-robust properties to much stronger assumptions. By using the rotational symmetry of identity-covariance Gaussians, we show that the algorithm can WLOG assume the samples are randomly rotated, i.e., it should behave the same on $X_1, \dots, X_{\ns} \in \R^d$ as on $MX_1, \dots, MX_{\ns}$ for any orthogonal $M \in \R^{d \times d}$. Moreover, using the fact that the empirical mean is a sufficient statistic for identity-covariance Gaussians (see \Cref{subsec:sufficient_statistics} for details on sufficient statistics), we can assume that the algorithm only depends on the samples $X_1, \dots, X_{\ns}$ via the empirical mean. Both of these reductions follow a similar approach to the symmetric and permutation-robust reductions, but the latter only holds for weakly replicable algorithms, because the empirical mean is not a sufficient statistic for general distributions. By combining these reductions, we may assume that the algorithm depends only on the $\ell_2$ norm of the empirical mean.

The final step involves a chaining lower bound similar to \Cref{thm:LB_chaining}. Specifically, we define $\mu_0 = 0$, $\mu_t = \mu$ to have norm $\alpha$, and choose appropriate $\mu_1, \mu_2, \dots, \mu_{t-1} \in \R^d$. For appropriate choices of $\mu_i$, we show that the norm of the empirical mean of $\ns$ samples from $\NN(\mu_i, I)$ or from $\NN(\mu_{i+1}, I)$ are statistically indistinguishable, unless $\ns$ is sufficiently large. This will allow us to prove our desired lower bound.

\medskip

The details of both the algorithm and the lower bound are found in \cref{sec:gaussian}.

\subsection{Selection via Testing} \label{subsec:hypothesis-selection-overview}

Hypothesis selection, also known as density estimation, is a core primitive underlying many statistical estimation tasks (see, e.g., \cite{devroye2001combinatorial}).
In this problem, given a collection of known distributions $\mathcal{H} = \{H_1, \ldots, H_n\}$ and sample access to an unknown distribution $P$, the goal is to return an index $i \in [n]$ corresponding to a distribution $H_i$ which is an approximate nearest neighbor of $P$ in total variation distance.
We design a $\rho$-replicable algorithm for this problem which reduces replicable selection to a sequence of (adaptively generated) replicable coin testing problems.
Using our optimal bounds for coin testing, we design an algorithm which is $\rho$-replicable, succeeds with high probability in $n$, and takes $O\p{\frac{\log^5 n}{\eps^2 \rho^2}}$ samples in the worst-case and $O\p{\frac{\log^5 n}{\eps^2 \rho}}$ samples in expectation.
Lower bounds for replicable Gaussian mean estimation imply our bounds are tight up to a factor of $\log^3 n \log(1/\eps)$. In the non-replicable setting, the optimal bounds are $\Theta\p{\frac{\log n}{\eps^2}}$.

Our reduction makes use of the non-replicable ``min distance estimate''~\cite{devroye2001combinatorial} for hypothesis selection.
This technique assigns a score $W_i$ to each hypothesis $H_i$.
The score can be estimated up to $\pm \eps$ with high probability in $\frac{\log(n)}{\eps^2}$ samples, and the hypothesis with a score within $\pm \eps$ of the minimum score is an approximate nearest neighbor of $P$.
Direct application of the concepts we explore for replicable testing seem difficult to apply to this problem as a replicable algorithm needs to coordinate, across independent sample sets, which among $n$ items it chooses, many of which may be close to being an approximate minimizer of the score.

We solve this problem via a hierarchy of testing problems.
We split the set of $n$ hypotheses into two groups.
We run the min distance estimator from ~\cite{devroye2001combinatorial} and observe which group contains the output of the algorithm.
We view the outcome of which group wins as a draw from a $\Ber(p)$ distribution.
Using our optimal replicable coin testing algorithm, we repeat this process several times and choose one of the groups.
Via coin testing, we guarantee that if one group's true probability of winning was at least $3/4$, we choose that group.
Thereby, we ensure that the group we pick has an approximate nearest neighbor.
We recursively repeat this procedure $\log n$ times until there is a single hypothesis remaining and output that hypothesis.

As the overall procedure is repeated $\log n$ times, within each iteration, we need to use error parameter $\eps_0 = \eps/\log(n)$ and replicability parameter $\rho_0 = \rho / \log(n)$ in order to guarantee the desired correctness and replicability over the entire algorithm.
The final sampling bounds follow from calculations using our sampling bounds for coin testing.
The details of this result are in \cref{sec:selection}.

\section{Preliminaries}\label{sec:prelims}

\subsection{Definitions}

As applications of our structural results on replicable testers as well as our general expectation-gap estimators, we design new upper and lower bounds on two classic distribution testing problems: uniformity (identity) testing and closeness testing. These follow the standard guarantees except we additionally require replicability. We formally define replicable uniformity and replicable closeness testing. 

\begin{definition}[Replicable Uniformity Testing (slightly modified from \cite{liu2024replicableuniformity})]\label{def:replicable_unif} Consider $n \in \N$, $0 \leq \delta \leq \rho \leq 1$ and $\eps > 0$. A randomized algorithm $\A$, given sample access to an unknown discrete distributions $p$ on $[n]$, is said to solve $(n, \eps, \rho, \delta)$-replicable uniformity testing if it is $\rho$-replicable and satisfies the following:
\begin{enumerate}
    \item If $p = \Unif([n])$, $\A$ accepts with probability at least $1-\delta$,
    \item If $\|p - \Unif([n])\|_1 \geq \eps$, $\A$ rejects with probability at least $1-\delta$.
\end{enumerate}
\end{definition}

\begin{definition}[Replicable Closeness Testing]\label{def:replicable_closeness}
Consider $n \in \N$, $0 \leq \delta \leq \rho \leq 1$ and $\eps > 0$.
A randomized algorithm $\A$, given sample access to a pair of distributions $p, q$ on $[n]$, is said to solve $(n, \eps, \rho, \delta)$-replicable closeness testing if it satisfies the following:
\begin{enumerate}
    \item If $p = q$, $\A$ accepts with probability at least $1-\delta$,
    \item If $\|p-q\|_1 \ge \eps$, $\A$ rejects with probability at least $1-\delta$,
    \item If for all pairs of distributions $(p', q')$ over $[n]$, 
    \[\Prma[X,X',r]{\A(X , r) = \A(X', r)} \ge 1 - \rho, \]
    where $r$ is the internal randomness of $\A$ and $X, X'$ consist of i.i.d. samples from the product distribution $p' \times q'$ over $[n]^2$.
\end{enumerate}
\end{definition}

We also define a concept called \emph{weak replicability testing}, which we will use in our lower bound against Gaussian mean testing. The intuition is that replicability must hold against all distributions, but we may want an algorithm to be replicable if the samples are actually drawn from a Gaussian.

\begin{definition}[Weak Replicability] \label{def:weak_replicability}
    Let $\DD_\theta$ be a family of distributions, parameterized by $\theta \in \Theta$, over some domain $\Omega$. An algorithm $\A$ taking $\ns$ samples is \emph{weakly replicable} if, given any $\theta \in \Theta$, we have 
\[\Prma[X, X' \sim \DD_\theta^{\otimes \ns}, r]{\A(X, r) = \A(X', r)} \ge 1-\rho.\]
\end{definition}

In our use of weak replicability, the choice of distribution family $\DD_\theta$ will always be evident.

\subsection{Concentration Inequalities}

We will need the following standard concentration inequalities in the analysis of our algorithms.

\begin{theorem}[Chernoff, Hoeffding, and Bernstein Bounds]\label{thm:concentration}
Let $Y = \sum_{i = 1}^n Y_i$ be a sum of independent random variables.  We have
\begin{itemize}
    \item If each $Y_i$ is distributed as Bernoulli $p_i$, then letting $\mu = \Ema[]{Y}$,
    \begin{enumerate}
        \item $\Prma[]{Y \ge (1+\delta)\Ema[]{Y}} \le \p{\frac{e^\delta}{(1+\delta)^{1+\delta}}}^{\Ema[]{Y}}$ for all $\delta > 0$,
        \item $\Prma[]{Y \ge (1+\delta)\Ema[]{Y}} \le \exp\left( -\frac{\delta^2 \mu}{2+\delta} \right)$ for all $\delta > 0$,
        \item $\Prma[]{Y \le (1-\delta)\Ema[]{Y}} \le \exp\left( -\frac{\mu \delta^2}{2} \right)$ for all $\delta \in (0, 1)$.
    \end{enumerate}
    \item If each $Y_i \in [a_i, b_i]$, then for all $t > 0$ 
    \[ \Prma[]{Y - \Ema[]{Y} \ge t } \le \exp\left( - \frac{2 t^2}{\sum_i (b_i-a_i)^2} \right).  \]
    \item If $|Y_i - \Ema[]{Y_i}| \le M$ with probability $1$ for all $i$, then for any $t > 0$,
    \[ \Prma[]{|Y - \Ema[]{Y}| \ge t } \le 2 \exp \left( - \frac{t^2/2}{\Varma[]{Y} + tM/3} \right). \]

\end{itemize}

\end{theorem}

In our replicable Gaussian mean testing section, \cref{sec:gaussian}, we will make use of the following non-standard inequality for concentration of a submodular function.

\begin{theorem}[Theorem 1, \cite{boucheron2000sharp}, Restated]\label{thm:boucheron}
Let $(X_1, ..., X_n)$ be independent random variables taking values in some measurable set $\mathcal{X}$, and let $f : \mathcal{X}^n \to [0,\infty)$ be a function. Assume that there exists another function $g : \mathcal{X}^{n-1} \to \mathbb{R}$ such that for any $x_1,\ldots,x_n \in \mathcal{X}$, the following properties hold:

\[0 \leq f(x_1,\ldots,x_n) - g(x_1,\ldots,x_{i-1},x_{i+1},\ldots,x_n) \leq 1, \quad \text{for every } 1 \leq i \leq n\]

and

\[\sum_{i=1}^n [f(x_1,\ldots,x_n) - g(x_1,\ldots,x_{i-1},x_{i+1},\ldots,x_n)] \leq f(x_1,\ldots,x_n).\]

Denote $Z = f(X_1,\ldots,X_n)$. Then for every positive number $t$,
\[\Prma[]{Z \geq \Ema[]{Z} + t} \leq \exp\left[-0.1 \cdot \min\left(t, \frac{t^2}{\Ema[]{Z}}\right)\right].\]
and\footnote{While the second inequality is only stated for $t \le \Ema[]{Z}$, it trivially holds for $t > \Ema[]{Z}$ since $f$ is nonnegative.}
\[\Prma[]{Z \leq \Ema[]{Z} - t} \leq \exp\left[-0.1 \cdot \min\left(t, \frac{t^2}{\Ema[]{Z}}\right)\right].\]
\end{theorem}

The following inequality bounds the variance in terms of the location of the random variable.

\begin{theorem}[\cite{bhatia2000better}] Let $Y \in [a, b]$ be a random variable. We have 
\[ \Varma[]{Y} \le (b-\Ema[]{Y})(\Ema[]{Y}-a). \]
\end{theorem}

Finally, we note the following folklore bound.

\begin{proposition} \label{prop:tv-chi-square-shifted}
    For any $k \ge 1$, the total variation distance between a chi-square $\chi_k^2$ and a shifted chi-square $\chi_k^2 + t,$ if $0 \le t \le 0.001 \sqrt{k}$, is at most $0.1$.
\end{proposition}

\subsection{Sufficient Statistics} \label{subsec:sufficient_statistics}

We recall the definition of sufficient statistics.

\begin{definition}[Sufficient Statistic] \label{def:sufficient_statistic}
    Given a distribution $\mathcal{D}(\theta)$ parameterized by some $\theta \in \Theta$, and given samples $X_1, \dots, X_{\ns} \overset{i.i.d.}{\sim} \mathcal{D}(\theta)$, a function $T = T(X_{1}, \dots, X_{\ns})$ is a \emph{sufficient statistic} for $\theta$ if the conditional distribution of $X_{1}, \dots, X_{\ns}$ conditioned on $T$ and $\theta$ is independent of $\theta$.
\end{definition}

Importantly, we use the well-known result that the empirical mean is a sufficient statistic for an identity-covariance Gaussian (see, e.g.,~\cite[Example 6.2.4]{CasellaBerger}, the proof of which easily generalizes to the multivariate setting).

\begin{proposition} \label{prop:gaussian_sufficient_statistic}
    Let $\mathcal{N}(\mu, I)$ be parameterized only by $\mu \in \R^d$. Then, given samples $X_{1}, \dots, X_{\ns},$ the empirical mean $\bar{X} = \frac{X_{1}+\cdots+X_{\ns}}{\ns}$ is a sufficient statistic for $\mu$. In other words, if given $\ns$ i.i.d. samples from $\mathcal{N}(\mu, I)$, the conditional distribution of $X_{1}, \dots, X_{\ns}$ conditioned on $\frac{X_{1}+\cdots+X_{\ns}}{\ns}$ is independent of $\mu$. 
\end{proposition}

\subsection{Translating between Worst-Case and In-Expectation Sampling Bounds}

An simple but powerful observation made in \cite{hopkins2024replicability} is that Markov's inequality can be used to translate between worst-case and in-expectation sample complexity bounds.
We summarize this observation in the following proposition, including its proof for completeness.

\begin{proposition}\label{prop:markov-translate}
    Let $\A(X;r)$ be a $\rho$-replicable algorithm which takes at most $\ns$ samples in expectation (over $X$ and $r$).
    Then, there exists an algorithm $\mathcal{B}(X;r)$ which deterministically takes at most $\ns/\rho$ samples, and has the property that
    \begin{equation*}
        \Prma[X, r]{\A(X;r) = \mathcal{B}(X;r)} \geq 1-\rho.
    \end{equation*}
\end{proposition}

\begin{proof}
Let $\mathcal{B}(X;r)$ simulate $\A(X;r)$, terminating if more than $\ns/\rho$ samples are used and outputting $\perp$.
By Markov's inequality, the probability that $\mathcal{B}(X;r)$ terminates early is at most $\rho$.
Otherwise, the two algorithms are equivalent.
\end{proof}

This proposition is useful from \emph{in-expectation to worst case} for sampling upper bounds and \emph{worst-case to in-expectation} sampling lower bounds with a blowup of $1/\rho$.
The correctness of $\mathcal{B}(X;r)$ can be guaranteed beyond failure probability $\rho$ by post-processing: if $\mathcal{B}(X;r)$ returns $\perp$, run a non-replicable, high probability algorithm and outputs its answer.

\subsection{Other Notation}

For a matrix $M$, we use $\|M\|_{op}$ to denote its operator norm. 

\section{Canonical Properties of a Replicable Tester} \label{sec:canonical-tester}

In this section, we show that the existence of a replicable algorithm implies the existence of another with a well-defined structural form. These structural assumptions will enable us to derive lower bounds later. In particular, we identify the following structural properties:

\begin{definition}[Random threshold algorithm]\label{def:canonical}
     A random threshold algorithm computes a deterministic function of its input, \( f: \mathcal{X}^n \rightarrow [0,1] \), and compares the value of this function to a random variable \( r \) drawn uniformly from \([0,1]\). 
\end{definition}

\begin{definition}[Sample order invariant algorithm]\label{def:sample_invariant}
    An algorithm is sample order invariant if the output distribution of the algorithm is invariant to permutations of the order in which the samples are received. 
\end{definition}

For a large class of ``symmetric properties'', we show more structure.
\begin{definition}[Symmetric Property]\label{def:symmetric_property} Suppose $\PP$ is a property of discrete distributions over $[n]$ ($\PP$ is a collection of distributions).
We say a property $\PP$ is symmetric if, for every $p \in \PP$ and every permutation function $\pi:[n]\rightarrow[n]$, the distribution $p \circ \pi$ is also in $\PP$.
\end{definition}
For symmetric properties—such as uniformity testing or closeness testing—membership in the property (and even the distance to the property) does not change if we permute the labels of the elements. For such properties, the labels are, in a sense, irrelevant and do not carry any information. Here, we prove that if a replicable algorithm exists for testing a symmetric property, then there also exists a replicable algorithm whose output is invariant to the labels of the samples and which achieves the same performance.

Finally, we show that
the label-invariant algorithm  satisfies a stronger replicability assumption. Specifically, the outcome of the algorithm remains stable even if we change the underlying distribution to another one obtained by permuting the labels. 

\begin{definition}[Permutation Robust Replicability]\label{def:perm-rubust-replicability}
    We say an algorithm $\A$ satisfies $\rho$-permutation robust replicability iff 
    for any prior distribution $\DD$ over a given distribution and all of its permutation, we have:
    $$\Prma[{r\sim\Unif[0,1]},\ p, p^\pi \sim \DD, X\sim p^{\otimes \ns}, X'\sim (p^\pi)^{\otimes \ns} ]{\A(X;r) \not = \A(X';r)} \leq \rho\,,$$
    where $p^\pi \coloneqq p \circ \pi^{-1}$.
\end{definition}

This high stability allows us to apply the replicability constraint not only to two sample sets drawn from the same distribution, but to two sample sets drawn from any two distributions which are equivalent up to a permutation of the domain.

More formally, we have the following theorem:

\canonical*

The existence of the canonical random threshold algorithm is established in Lemma~\ref{lem:canonical}. In \cref{lem:order_invariant}, we demonstrate that the canonical algorithm can be modified to be invariant to the order of the samples. In \cref{lem:label_invariant}, we further modify the algorithm so that it exhibits identical behavior on every sample set that can be obtained by relabeling its elements according to some permutation. Finally, in Lemma~\ref{lem:perm_robust}, we prove that the algorithm robustly maintains its replicability even when the underlying distribution is altered to a permuted version.
These lemmas and their proofs are presented in Section~\ref{sec:canonical_lemmas}.

\subsection{Proof of Canonical Properties }\label{sec:canonical_lemmas}
\begin{lemma}[Canonical Random Threshold Algorithm] \label{lem:canonical}
Let $\A_0(X; r)$ be a $\rho$-replicable algorithm that solves a given problem using $\ns$ samples $X = (X_1, \dots, X_{\ns})$ and randomness $r$, and outputs a binary decision in $\{\text{\accept}, \text{\reject}\}$. Then, there exists another $\rho$-replicable algorithm $\A_1(X; r)$ that also solves the problem on $\ns$ samples with the same accuracy as $\A_0$, and it operates as follows:

\begin{center}
It computes a deterministic function $f: \XX^n \rightarrow [0,1]$ of the input sample set $X$. Then, it samples a seed $r \sim \Unif([0,1])$, and outputs \accept if $r \le f(X)$, and \reject otherwise.
\end{center}
\end{lemma}

\begin{proof}
    For any sample set $X = (X_1, X_2, \dots, X_{\ns})$ of $\ns$ samples, we simply define: $$f(X) := \Prma[r]{\A_0(X; r) = \text{\accept}}\,.$$
    Then, as in the description of the lemma statement, the algorithm $\A_1$ is defined to sample $r \sim \Unif[0, 1]$ and accept if and only if $r \le f(X)$.
    
    First, note that for any fixed $X$, given the structure of $\A_1$, we have: 
    \begin{equation}\label{eq:A_0_accept}
        \Prma[{r \sim \Unif([0, 1])}]{\A_1(X; r) = \text{\accept}}  = \Prma[r] {r \le f(X)}  = f(X) = \Prma[r]{\A_0(X; r) = \text{\accept}}
        \,.
    \end{equation}
    Thus, given that the tester has only two possible outcome, the probabilities of both outputting \accept and \reject match between $\A_1$ and $\A_0$. Therefore, $\A_1$ will still solve the problem with the same accuracy guarantees as $\A_0$.

    Next, we check replicability. Let $\cD$ be any distribution, and let $X = (X_1, \dots, X_{\ns})$ and $X'= (X_1', \dots, X_{\ns}')$ be two sample sets each containing $\ns$ i.i.d.\  samples drawn from that distribution. Since $\A_0$ is replicable, and by the law of total expectation, we have
\begin{align*}
    1-\rho 
    &\le \Prma[{r, X, X'}] {\A_0(X; r) = \A_0(X'; r)} \\
    &= \Ema[{X, X'}] {\Prma[r] {\A_0(X; r) = \A_0(X'; r)} }.
\end{align*}

    Now, for any fixed $X$ and $X'$, note that 
\begin{align*}
    &\Prma[r]{\A_0(X; r) \neq \A_0(X'; r)} \geq \tv{\A_0(X; r),\, \A_0(X'; r)} \tag{by coupling inequality}
    \\ & \quaaad \ge \abs{\Prma[r]{\A_0(X; r) = \text{\accept}} - \Prma[r]{\A_0(X'; r) = \text{\accept}}}
    \\
    & \quaaad = \abs{f(X) ~-~f(X')} = \Prma[{r \sim \Unif([0, 1])}]{r \in (f(X) , f(X')]}
    \\ & \quaaad
    = \Prma[{r \sim \Unif([0, 1])}]{\A_1(X; r) \neq \A_1(X'; r)}\,,
\end{align*}
    by the way we have defined $\A_1$. Thus, by taking the expectation over $X$ and $X'$, we have that 
\begin{align*}
    1-\rho &\le \Ema[X, X'] {\Prma[r]{\A_0(X; r) = \A_0(X'; r)} } \\
    &\le \Ema[X, X'] {\Prma[r] {\A_1(X; r) = \A_1(X'; r)} } \\
    &= \Prma[r, X, X'] {\A_1(X; r) = \A_1(X'; r)}.
\end{align*}
    Thus, $\A_1$ is $\rho$-replicable.
\end{proof}

\begin{remark}
    We note that the reduction in \Cref{lem:canonical} is potentially inefficient if we cannot compute $f$ efficiently. However, to prove \emph{statistical} lower bounds, the efficiency of computing $f$ is irrelevant.
\end{remark}

\begin{lemma}[Order Invariant Algorithm]\label{lem:order_invariant}
    Let $\A_0(X; r)$ be a $\rho$-replicable algorithm that solves a given problem using $\ns$ i.i.d.\ samples $X = (X_1, \dots, X_{\ns})$ from an underlying distribution and randomness $r$, and outputs a binary decision in $\{\text{\accept}, \text{\reject}\}$. Then, there exists another $\rho$-replicable algorithm $\A_2(X; r)$ that solves the same problem on $\ns$ samples with the same accuracy and is invariant to the order of the samples. That is, for every seed $r$, permutation function $\sigma:[\ns]\rightarrow[\ns]$, and sample set $X$, we have: 
    $$\A_2(X;r) = \A_2(X_\sigma;r)\,,$$
where $X_\sigma$ denotes $(X_{\sigma(1)}, X_{\sigma(2)}, \ldots, X_{\sigma(n)})$.
\end{lemma}

\begin{proof}
    Let $\A_1(X;r)$ be the algorithm defined in Lemma~\ref{lem:canonical} with the deterministic function $f:\XX^n \rightarrow [0,1]$.
    Consider the following deterministic function of the sample set $X$, $q:\XX^n\rightarrow[0,1]\,$:
    $$q(X) \coloneqq \frac{1}{\ns!}\sum_\sigma f(X_{\sigma})\,.$$
    Essentially, $q$ is the average of $f$ over all possible permutations of the samples. The algorithm $\A_2(X; r)$ operates similarly to $\A_1(X; r)$, except that it uses $q$ instead of $f$. For a random seed $r \sim \Unif([0, 1])$, $\A_2(X; r)$ outputs \accept if $r \le q(X)$, and \reject otherwise. Clearly, $\A_2$ is order invariant.

To prove the accuracy of $\A_2$, consider a fixed underlying distribution $p$. For any permutation function $\sigma$, since $X$ contains  i.i.d.\ samples from $p$, the permuted sample $X_\sigma$ has the same distribution as $X$; that is, $X_\sigma \overset{d}{=} X$.
Thus, we have:
\begin{align}
    \Prma[r,X\sim p^{\otimes \ns}]
    {\A_2(X;r) = 
    \text{\accept}} & = \Ema[X\sim p^{\otimes \ns}]
    {\Prma[r]{\A_2(X;r) = 
    \text{\accept}}} \nonumber
    \\& = 
    \Ema[X\sim p^{\otimes \ns}] {q(X)} = \frac{1}{\ns!}\sum_\sigma \Ema[X\sim p^{\otimes \ns}]{ f(X_{\sigma})} \nonumber
    \\ & =
    \frac{1}{\ns!}\sum_\sigma \Ema[X\sim p^{\otimes \ns}]{ f(X)} \tag{using $X_\sigma \overset{d}{=}X$} \nonumber
    = 
    \Ema[X\sim p^{\otimes \ns}]{ f(X)}
    \\ & = \Ema[X\sim p^{\otimes \ns}]
    {\Prma[{r\sim \Unif[0,1]}]{\A_1(X;r)= 
    \text{\accept}}} \nonumber
    \\ & =
    \Prma[r,X\sim p^{\otimes \ns}]{\A_1(X;r) = \text{\accept}} \nonumber
    \\ & = \Prma[r,X\sim p^{\otimes \ns}]{\A_0(X;r) = \text{\accept}}
    \,. 
        \tag{Using Lemma~\ref{lem:canonical}, Eq.~\eqref{eq:A_0_accept}}
    \\
    & \label{eq:A_2_A_0_accept}
\end{align}

Therefore, $\A_2$ has the same probabilities of outputting \accept and \reject as $\A_0$, and thus inherits the accuracy guarantees of $\A_0$.

Next, we show replicability of $\A_2$. Similar to Lemma~\ref{lem:canonical}, we have: 
\begin{equation} \label{eq:pr_rep_A_2_part1}
    \begin{split}
\Prma[r,X,X'\sim p^{\otimes \ns}]
    {\A_2(X;r) \not = \A_2(X';r)} & = \Ema[X,X']
    {\Prma[r]{\A_2(X;r) \not = \A_2(X';r)}} 
    \\& = 
    \Ema[X,X']
    {\abs{q(X) ~-~q(X')}}\,,
    \end{split}
\end{equation}
where in the last line, we use the structure of $\A_2$. Using the definition of $q$, we have:
\begin{align} 
    \Prma[r,X,X'\sim p^{\otimes \ns}]
    {\A_2(X;r) \not = \A_2(X';r)}
    & = 
    \Ema[X,X']
    {\abs{q(X) ~-~q(X')}} \nonumber
    \\ & \leq  
    \Ema[X,X']
    {\frac{1}{\ns!} \sum_\sigma \abs{f(X_\sigma) ~-~f(X'_\sigma)}} \tag{Via triangle inequality}
    \\&  \leq \frac{1}{\ns!} \sum_\sigma   
    \Ema[X,X']
    {\abs{f(X_\sigma) ~-~f(X'_\sigma)}}\,. \label{eq:pr_rep_A_2_part2}
\end{align}
Recall that the distribution of sample sets remains identical after permutation, so $X_\sigma \overset{d}{=} X$ and $X'_\sigma \overset{d}{=} X'$. Hence, we obtain:
\begin{align*}
\Prma[r,X,X'\sim p^{\otimes \ns}]
    {\A_2(X;r) \not = \A_2(X';r)}
    &   \leq 
    \frac{1}{\ns!} \sum_\sigma   
    \Ema[X,X']
    {\abs{f(X_\sigma) ~-~f(X'_\sigma)}} 
    \\ & =
    \frac{1}{\ns!} \sum_\sigma   
    \Ema[X,X']
    {\abs{f(X) ~-~f(X')}}
    \\ & = 
    \Ema[X,X']
    {\abs{f(X) ~-~f(X')}}
    \\ & = \Ema[X,X']{\Prma[{r\sim \Unif[0,1]}]{\A_1(X;r) \not = \A_1(X';r)}}
    \\& = \Prma[r,X,X'\sim p^{\otimes \ns}]
    {\A_1(X;r) \not = \A_1(X';r)} \leq \rho \tag{Using Lemma~\ref{lem:canonical}}
\end{align*}
Hence, the proof is complete. 
\end{proof}

\begin{lemma}[Label Invariant Algorithm]\label{lem:label_invariant}
    Let $\A_0(X; r)$ be a $\rho$-replicable algorithm for testing a symmetric property $\PP$ of discrete distributions over $[n]$, using $\ns$ i.i.d.\ samples $X = (X_1, \dots, X_{\ns})$ drawn from an underlying distribution $p$, and randomness $r$. The algorithm outputs a binary decision in $\{\text{\accept}, \text{\reject}\}$. The accuracy of $\A_0$ is determined by two parameters $\epsilon$ and $\delta$ in $(0,1)$, satisfying the following:

    \begin{itemize}
        \item If $p \in \PP$, then 
        $$\Prma[X\sim p^{\otimes \ns}, r]{\A_0(X;r) = \text{\accept}} \geq 1-\delta\,.$$
        \item If $p$ is $\epsilon$-far from $\PP$, then 
        $$\Prma[X\sim p^{\otimes \ns}, r]{\A_0(X;r) = \text{\reject}} \geq 1-\delta\,.$$
    \end{itemize}

    Then, there exists another $\rho$-replicable algorithm $\A_3(X; r)$ that solves the same problem using $\ns$ samples with the same accuracy, and is invariant to the labels of the samples. That is, for every seed $r$, permutation function $\pi:[\ns]\rightarrow[\ns]$, and sample set $X$, we have:
    $$\A_3(X;r) = \A_3(\pi(X);r)\,,$$
    where $\pi(X)$ denotes $\left(\pi(X_1), \pi(X_2), \ldots, \pi(X_{\ns})\right)$. 

    Moreover, $\A_3$ is invariant to the order of the samples, and operates in the canonical format of comparing a deterministic function to a random threshold, as defined in Lemma~\ref{lem:canonical}.
\end{lemma}

\begin{proof}
Let $\A_2(X;r)$ be the algorithm defined in \cref{lem:order_invariant} corresponding to $\A_0$, with an associated deterministic function $q:\XX^n \rightarrow [0,1]$. Define a new deterministic function $h:\XX^n \rightarrow [0,1]$ as:
$$h(X) \coloneqq \frac{1}{n!}\sum_\pi q(\pi(X))\,.$$
That is, $h$ computes the average of $q$ over all permutations of the sample labels.

The algorithm $\A_3(X; r)$ behaves similarly to $\A_2(X; r)$ and $\A_1(X;r)$: for $r \sim \Unif([0, 1])$, it outputs \accept if $r \le h(X)$, and \reject otherwise. Since $q$ and therefore $h$ are invariant to the sample order, $\A_3$ is order-invariant.

Next, we verify the accuracy guarantee of $\A_3$. Consider any permutation function $\pi:[n]\rightarrow[n]$. For any sample set $X$, the permuted sample $\pi(X)$ can be viewed as drawn i.i.d.\ from $p^\pi \coloneqq p \circ \pi^{-1}$. That is, for all $j$, the probability of element $\pi(j)$ under $p^\pi$ equals $p(j)$. We now follow the structure of the proof in \cref{lem:order_invariant}:

\begin{align*}
    \Prma[r,X\sim p^{\otimes \ns}]{\A_3(X;r) = \text{\accept}} 
    &= \Ema[X\sim p^{\otimes \ns}]{\Prma[r]{\A_3(X;r) = \text{\accept}}} \\
    &= \Ema[X\sim p^{\otimes \ns}]{h(X)} \\
    &= \frac{1}{n!}\sum_\pi \Ema[X\sim p^{\otimes \ns}]{ q(\pi(X))} \\
    &= \frac{1}{n!}\sum_\pi \Ema[X\sim (p^\pi)^{\otimes \ns}]{ q(X)} 
    \tag{since $\pi(X) \sim (p^\pi)^{\otimes \ns}$ when $X \sim p^{\otimes \ns}$} \\
    &= \frac{1}{n!}\sum_\pi \Ema[X\sim (p^\pi)^{\otimes \ns}]{\Prma[r]{\A_2(X;r)= \text{\accept}}} \\
    &= \frac{1}{n!}\sum_\pi \Prma[r, X\sim (p^\pi)^{\otimes \ns}]{\A_2(X;r)= \text{\accept}} \\
    &= \frac{1}{n!}\sum_\pi \Prma[r, X\sim (p^\pi)^{\otimes \ns}]{\A_0(X;r)= \text{\accept}} 
    \tag{by \cref{lem:order_invariant}, Eq.~\eqref{eq:A_2_A_0_accept}}
\end{align*}

The above shows that the probability $\A_3$ accepts under $p$ is the average acceptance probability of $\A_0$ under $p^\pi$ for all permutations $\pi$.

Now, for symmetric properties: if $p \in \PP$, then by definition of symmetry (Definition~\ref{def:symmetric_property}), each $p^\pi$ also belongs to $\PP$. Thus, $\A_0$ accepts each $p^\pi$ with probability at least $1-\delta$, and therefore $\A_3$ accepts with probability at least $1-\delta$.

Conversely, if $p$ is $\epsilon$-far from $\PP$, then so is every $p^\pi$, since total variation distance is invariant under permutation. In particular, for every $\pi$:
\begin{align*}
    \dist(p,\PP) &= \min_{d \in \PP} \dist(p,d) = \min_{d \in \PP} \dist(p^\pi, d^\pi) \\
    &\leq \min_{d \in \PP} \dist(p^\pi, d) = \dist(p^\pi,\PP)\,,
\end{align*}
where the inequality holds since $d^\pi \in \PP$ for symmetric $\PP$. By applying the same reasoning in the reverse direction using $\pi^{-1}$, we conclude:
$$\dist(p,\PP) = \dist(p^{\pi}, \PP)\,.$$

Hence, each $p^\pi$ is also $\epsilon$-far from $\PP$, so $\A_0$ rejects each with probability at least $1 - \delta$. Therefore, $\A_3$ accepts with probability less than $\delta$, as desired.

We now prove the replicability of $\A_3$. Using the same logic as in Equations~\eqref{eq:pr_rep_A_2_part1} and~\eqref{eq:pr_rep_A_2_part2}, we get:
\begin{equation*}
    \begin{split}
\Prma[r,X,X'\sim p^{\otimes \ns}]{\A_3(X;r) \not = \A_3(X';r)} 
&\leq \frac{1}{n!} \sum_\pi \Ema[X,X'\sim p^{\otimes \ns}]{\abs{q(\pi(X)) - q(\pi(X'))}}\,.
    \end{split}
\end{equation*}
Since $\pi(X)$, when $X \sim p^{\otimes \ns}$, is distributed as $(p^\pi)^{\otimes \ns}$, we have:
\begin{equation*}
    \begin{split}
\Prma[r,X,X'\sim p^{\otimes \ns}]{\A_3(X;r) \not = \A_3(X';r)} 
&\leq \frac{1}{n!} \sum_\pi \Ema[X,X'\sim (p^\pi)^{\otimes \ns}]{\abs{q(X) - q(X')}}\,.
    \end{split}
\end{equation*}
Using \cref{eq:pr_rep_A_2_part1}, we get:
$$\Ema[X,X'\sim (p^\pi)^{\otimes \ns}]{\abs{q(X) - q(X')}} = \Prma[r,X,X'\sim p^{\otimes \ns}]{\A_2(X;r) \not = \A_2(X';r)} \leq \rho\,.$$
The inequality holds because $\A_2$ is $\rho$-replicable for any $p^\pi$ (by \cref{lem:order_invariant}). Therefore:
\begin{equation*}
    \begin{split}
\Prma[r,X,X'\sim p^{\otimes \ns}]{\A_3(X;r) \not = \A_3(X';r)} 
&\leq \frac{1}{n!} \sum_\pi \Prma[r,X,X'\sim p^{\otimes \ns}]{\A_2(X;r) \not = \A_2(X';r)} \leq \rho\,.
    \end{split}
\end{equation*}
Hence, the proof is complete.\end{proof}

\begin{lemma}[Permutation-Robust Replicability]\label{lem:perm_robust}
    Algorithm $\A_3$, introduced in \cref{lem:label_invariant}, satisfies an important stability assumption—namely, $\rho$-permutation robust replicability. That is, for any prior distribution $\DD$ over a given distribution and all of its permutation, we have:
    $$\Prma[{r\sim\Unif[0,1]},\ p, p^\pi \sim \DD, X\sim p^{\otimes \ns}, X'\sim (p^\pi)^{\otimes \ns} ]{\A_3(X;r) \not = \A_3(X';r)} \leq \rho\,,$$
    where $p^\pi \coloneqq p \circ \pi^{-1}$ and $\pi$ is a permutation $\pi:[n]\rightarrow[n]$.
\end{lemma}

\begin{proof}
Fix a distribution $p$ over $[n]$. Consider any permutation function $\pi:[n]\rightarrow[n]$. A sample set $X' \sim (p^\pi)^{\otimes \ns}$ has the same distribution as the sample set  $\pi^{-1}(X'')$ where $X''$ from $p^{\otimes \ns}$. This identity in distribution allows us to write: 
\begin{align*}
& \Prma[{r\sim\Unif[0,1]},\ X\sim p^{\otimes \ns}, X'\sim (p^\pi)^{\otimes \ns} ]{\A_3(X;r) \not = \A_3(X';r)} \\& \quaad = \Prma[{r\sim\Unif[0,1]},\ X\sim p^{\otimes \ns}, X''\sim (p)^{\otimes \ns} ]{\A_3(X;r) \not = \A_3(\pi^{-1}(X'');r)}
\\& \quaad = \Prma[{r\sim\Unif[0,1]},\ X\sim p^{\otimes \ns}, X''\sim (p)^{\otimes \ns} ]{\A_3(X;r) \not = \A_3(X'';r)}
\\& \quaad \leq \rho
\end{align*}
In the second to last line above, we use the label-invariant property of $\A_3$ in~\cref{lem:label_invariant}. And, in the last line, we use the fact that $A_3$ is $\rho$-replicable. Taking an expectation over all possible choices of $p$ and $\pi$ gives us the desired statement. 
\end{proof}

\section{Replicable Lower Bounds via Chaining}\label{sec:chaining}

In this section, we introduce a general framework for proving lower bounds for replicable testing algorithms. We refer to \cref{sec:tech_chaining} for a high-level intuitive overview of the following theorem.

We start with the following general lemma, which we will apply to prove our chaining lower bound for discrete distributions.

\begin{lemma} \label{lem:chaining_lemma}
    Let $\rho \in (0, 0.001]$ and $1 \le t \le 1/(300 \rho)$. Let $Z_0, Z_1, \dots, Z_t$ be distributions in some probability space $\Omega$, such that for every $1 \le i \le t$, $\dtv(Z_{i-1}, Z_i) \le 0.5$. Let $\A$ be an algorithm that takes a value $x \in \Omega$ and randomness $r \sim \Unif[0, 1]$, computes a deterministic function $h(x) \in [0, 1]$ and accepts if and only if $r \le h(x)$. In other words, it satisfies the canonical property in \Cref{def:canonical}, though we think of the input as a single sample.

    Suppose that for every $0 \le i \le t$, $\Prma[r \sim \Unif{[0, 1]}, x, x' \sim Z_i]{(\A(x; r) \neq \A(x'; r))} \le \rho$. Then, either the probability that $\A$ accepts on $Z_0$ is less than $2/3$, or the probability that $\A$ rejects on $Z_t$ is less than $2/3$. 
\end{lemma}

Although \Cref{lem:chaining_lemma} is stated for a single sample, we can apply it to the setting of multiple samples by letting $\Omega$ be a product space, as we will see in the proof of \Cref{thm:LB_chaining}.

\begin{proof}
    Since $\A$ satisfies \Cref{def:canonical}, for any distribution $Z$, 
\begin{align*}
    \Prma[r \sim \Unif{[0, 1]}, x, x' \sim Z]{(\A(x; r) \neq \A(x'; r))} 
    &= \Ema[x \sim Z]{\Ema[x' \sim Z]{|h(x)-h(x')|}} \\
    &\ge \Ema[x \sim Z]{\left|h(x)-\Ema[x' \sim Z]{h(x')}\right|} \tag{by Jensen's inequality}.
\end{align*}
    So, if $\Prma[r \sim \Unif{[0, 1]}, x, x' \sim Z]{(\A(x; r) \neq \A(x'; r))} \le \rho$, then by Markov's inequality, with probability at least $0.9$ over $x \sim Z_i$, $|h(x) - \Ema[x' \sim Z_i]{h(x')}| \le 10 \rho$. If we define $q_i := \Ema[x' \sim Z_i]{h(x')}$, then $|h(x) - q_i| \le 10 \rho$ with at least $0.9$ probability.

    As a consequence, we claim that $|q_i-q_{i-1}| \le 20 \rho$ for all $1 \le i \le t$. This is because if $\dtv(Z_{i-1}, Z_i) \le 0.5$, then since $h$ is deterministic, $\dtv(h(Z_{i-1}), h(Z_i)) \le 0.5$. So, if $|h(x)-q_i| \le 10 \rho$ with probability at least $0.9$ for $x \sim Z_i$, then $|h(x)-q_i| \le 10 \rho$ with probability at least $0.4$ for $x \sim Z_{i-1}$. However, $|h(x)-q_{i-1}| \le 10 \rho$ with probability at least $0.9$ for $x \sim Z_{i-1}$, so with probability at least $0.3$, $|h(x)-q_{i-1}| \le 10 \rho$ and $|h(x)-q_i| \le 10 \rho$. By the Triangle inequality, $|q_i-q_{i-1}| \le 20 \rho$.

    Applying again the Triangle inequality for $i = 1, 2, \dots, t$, we have $|q_0-q_t| \le 20 \rho \cdot t \le \frac{1}{15}$, since $t \le \frac{1}{300 \rho}$. So, either $q_0 < \frac{2}{3}$ or $q_t > \frac{1}{3}$. Since $q_i = \Ema[x \sim Z_i]{h(x)} = \Prma[x \sim Z_i, r]{r \le h(x)}$ equals the probability of accepting a sample from $Z_i$, the claim is complete.
\end{proof}

\chaininglb*

\begin{proof}
Assume for the sake of contradiction that a $\rho$-replicable algorithm $\A_0$ exists. Let $\A_3$ be the improved version of $\A_0$, which has all the canonical properties including being a $\rho$-permutation robust replicable algorithm as guaranteed by \cref{thm:main_canonical}. The construction of $\A_3$ may be found in \cref{lem:label_invariant}. Recall that $\A_3$ computes a deterministic function $h(X)\in[0,1]$ and then compares it with a random threshold $r\sim\Unif[0,1]$, and accepts as long as $r \le h(X)$.

Note that the priors $\DD_0, \dots, \DD_t$ are distributions over $[n]^k$. Moreover, since $\A_3$ is label-invariant, 
\[\Prma[X \sim p_0, r]{\A_3(X; r) = \accept} = \Prma[X_1, \dots, X_k \sim p_0^{\pi}, r]{\A_3(x; r) = \accept} \ge 1-\delta\]
and 
\[\Prma[X_1, \dots, X_k \sim p_t, r]{\A_3(X; r) = \reject} = \Prma[x \sim p_t^{\pi}, r]{\A_3(X; r) =\reject} \ge 1-\delta\]
for any permutation $\pi$. Since $\DD_0$ is a prior over $p_0^{\pi}$ over different permutations, $\Prma[X \sim \DD_0, r]{\A_3(X; r) = \accept} \ge 1-\delta \ge \frac{2}{3}$ and $\Prma[X \sim \DD_t, r]{\A_3(X; r) = \reject} \ge 1-\delta \ge \frac{2}{3}$.

By setting $Z_i = \DD_i$ and $\Omega = [n]^k$ and applying \Cref{lem:chaining_lemma}, this is a contradiction.
\end{proof}

\subsection{Lower Bound Applications}\label{sec:chaining_applications}

\subsubsection{Coin Testing}\label{sec:coin_lb}


Our first application is a lower bound for coin testing (or coin estimation), a result already proven in~\cite{impagliazzo2022reproducibility}. Here, we include it as an application of our lower bound framework and demonstrate how our approach simplifies the proof. We set $p_0$ to be an unbiased coin and $p_t$ to be a coin with bias $\epsilon$. We then pack $t = \Theta(1/\rho)$ distributions between these two by setting the bias of the $i$-th coin to be $i\,\epsilon/t$. At this point, we follow the folklore result for showing the indistinguishability via Pinsker's inequality.

\begin{theorem} \label{thm:coin_lb_via_chaining} For any $\rho \leq 0.001$, $\epsilon \leq 0.25$, 
    any $\rho$-replicable algorithm that can distinguish an unbiased coin from one with bias of $1/2\pm\epsilon$ is required to use $\Omega(1/(\rho \epsilon)^2)$ samples. 
\end{theorem}
\begin{proof}
We apply \Cref{lem:chaining_lemma}, to the following chain of distributions. Let $t \coloneqq \floor{1/(300\rho)}$. 
\begin{align*}
    \text{For } i \in \{0, \ldots, t\}:\quaaad p_i(\text{head}) = \frac{1}{2} + \frac{i\,\epsilon}{t}\,, \text{ and }\quaad p_i(\text{tail}) = \frac{1}{2} - \frac{i\,\epsilon}{t}
    \,.
\end{align*}
Next, we show indistinguishability with $\ns = o(1/(\rho \epsilon)^2)$ samples.
If for every $i$, $\tv{p_i^{\otimes \ns} - p_{i-1}^{\otimes \ns}} \le 0.5$, then we may apply \Cref{lem:chaining_lemma} by setting $Z_i := p_i^{\otimes \ns}$.
So, we assume the contrary, and we have: 
\begin{align*}
    0.5 &\leq \tv{p_i^{\otimes \ns} - p_{i-1}^{\otimes \ns}} \leq \sqrt{\frac{1}{2}KL(p_i^{\otimes \ns} ~||~p_{i-1}^{\otimes \ns})}\tag{By Pinsker's inequality}
    \\ & \leq \sqrt{\frac{\ns}{2}KL(p_i~||~p_{i-1})} \tag{since samples are i.i.d.}\,.
\end{align*}
To achieve a contradiction, it suffices to show that $KL(p_i~||~p_{i-1})$ is $O(\epsilon^2/t^2)$.
\begin{align*}
    KL(p_i~||~p_{i-1}) & = \left(\frac{1}{2} - \frac{i\,\epsilon}{t}\right) \cdot \log\left[\frac{\frac{1}{2} - \frac{i\,\epsilon}{t}}{\frac{1}{2} - \frac{(i-1)\,\epsilon}{t}} 
    \right]
    +
    \left(\frac{1}{2} + \frac{i\,\epsilon}{t}\right) 
    \cdot \log \left[\frac{\frac{1}{2} + \frac{i\,\epsilon}{t}}{\frac{1}{2} + \frac{(i-1)\,\epsilon}{t}}
    \right] 
\\ 
    & = \frac{1}{2} \cdot
    \log\left[
    \frac{\frac{1}{2} - \frac{i\,\epsilon}{t}}{\frac{1}{2} - \frac{(i-1)\,\epsilon}{t}}
    \cdot
    \frac{\frac{1}{2} + \frac{i\,\epsilon}{t}}{\frac{1}{2} + \frac{(i-1)\,\epsilon}{t}}
    \right] 
    +
    \frac{i\,\epsilon}{t} \cdot 
    \log\left[
    \frac{\frac{1}{2} + \frac{i\,\epsilon}{t}}{\frac{1}{2} + \frac{(i-1)\,\epsilon}{t}}
    \cdot
    \frac{\frac{1}{2} - \frac{(i-1)\,\epsilon}{t}}{\frac{1}{2} - \frac{i\,\epsilon}{t}}
    \right] 
\\ &
    = \frac{1}{2} \cdot \log \left[
    \frac{\frac{1}{4} - \left(\frac{i\,\epsilon}{t}\right)^2}{\frac{1}{4} - \left(\frac{(i-1)\,\epsilon}{t}\right)^2}
    \right]
    + \frac{i\,\epsilon}{t} \cdot 
    \log \left[\frac{
    \frac{1}{4} + \frac{\epsilon}{2\,t} - \frac{i\,(i-1)\epsilon^2}{t^2}
    }{
    \frac{1}{4} - \frac{\epsilon}{2\,t} - \frac{i\,(i-1)\epsilon^2}{t^2}
    }
    \right]
    \\& = 
    \frac{1}{2} \cdot \log \left[1 - 
    \frac{
    \left(\frac{i\,\epsilon}{t}\right)^2
    -
    \left(\frac{(i-1)\,\epsilon}{t}\right)^2 }
    {\frac{1}{4} - \left(\frac{(i-1)\,\epsilon}{t}\right)^2}
    \right]
    + \frac{i\,\epsilon}{t} \cdot 
    \log \left[ 1 + \frac{
    \frac{\epsilon}{t}
    }{
    \frac{1}{4} - \frac{\epsilon}{2\,t} - \frac{i\,(i-1)\epsilon^2}{t^2}
    }
    \right]
    \,.
\end{align*} 
To remove the $\log$ terms, we upper bound $1+x$ with $e^x$ which holds for all $x \in \mathbb{R}$. We get:
\begin{align*}
    KL(p_i~||~p_{i-1}) & \leq
    \frac{-1}{2} \cdot  
    \frac{\frac{\epsilon}{t} \cdot \frac{(2\,i-1)\,\epsilon}{t}}
    {\frac{1}{4} - \left(\frac{(i-1)\,\epsilon}{t}\right)^2}
    + \frac{i\,\epsilon}{t} \cdot 
    \frac{
    \frac{\epsilon}{t}
    }{
    \frac{1}{4} - \frac{\epsilon}{2\,t} - \frac{i\,(i-1)\epsilon^2}{t^2}
    }
    \\ & =
    \frac{-(i-0.5)\,\epsilon^2}{t^2} \cdot \frac{1}{\frac{1}{2} + \frac{(i-1)\,\epsilon}{t}} \cdot \frac{1}{\frac{1}{2} - \frac{(i-1)\,\epsilon}{t}} 
    + 
    \frac{i\,\epsilon^2}{t^2} \cdot 
    \frac{1}{\frac{1}{2} + \frac{(i-1)\,\epsilon}{t}}
    \cdot 
    \frac{1}{\frac{1}{2} - \frac{i\,\epsilon}{t}}
    \\& = 
    \frac{(i-0.5)\,\epsilon^2}{t^2} \cdot 
    \frac{1}{\frac{1}{2} + \frac{(i-1)\,\epsilon}{t}} 
    \cdot \left(\frac{1}{\frac{1}{2} - \frac{i\,\epsilon}{t}}
    - 
    \frac{1}{\frac{1}{2} - \frac{(i-1)\,\epsilon}{t}} 
    \right)
    + 
    \frac{0.5 \, \epsilon^2}{t^2} \cdot 
    \frac{1}{\frac{1}{2} + \frac{(i-1)\,\epsilon}{t}}
    \cdot 
    \frac{1}{\frac{1}{2} - \frac{i\,\epsilon}{t}}
    \\ & = \frac{(i-0.5)\epsilon^3}{t^3} \cdot 
    \left(
    \frac{1}{\frac{1}{2} + \frac{(i-1)\,\epsilon}{t}} 
    \cdot 
    \frac{1}{\frac{1}{2} - \frac{i\,\epsilon}{t}}
    \cdot 
    \frac{1}{\frac{1}{2} - \frac{(i-1)\,\epsilon}{t}} 
    \right)
    + \cdot
    \frac{0.5\,\epsilon^2}{t^2} \cdot 
    \frac{1}{\frac{1}{2} + \frac{(i-1)\,\epsilon}{t}}
    \cdot 
    \frac{1}{\frac{1}{2} - \frac{i\,\epsilon}{t}}
    \\ & = O \left(\frac{i \,\epsilon^3}{t^3} + \frac{\epsilon^2}{t^2}
    \right)
    =  O\left(\frac{\epsilon^2}{t^2}\right)
    \,.
\end{align*}
In the last line, we used the fact that $i\leq t$, and $i\,\epsilon/t \leq 0.25$.
\end{proof}

\subsubsection{Uniformity and Identity Testing}\label{sec:unif_lb}
Our second application is a lower bound for uniformity testing, which also implies a lower bound for identity testing.  In~\cite{liu2024replicableuniformity}, a lower bound for uniformity was presented that is restricted to label-invariant algorithms, which the authors refer to as ``symmetric algorithms.''  Their proof, similar to our framework, provides $t = O(1/\rho)$ classes of distributions for which consecutive pairs are hard to distinguish.  Besides the indistinguishability of consecutive pairs, their proof required an extensive argument for the ``Lipschitz continuity of acceptance probability'' of the algorithm. 

We tighten their indistinguishability lower bound by a logarithmic factor and remove the restriction to label-invariant algorithms, making the result hold for all replicable algorithms.  Specifically, our result on the canonical tester implies that if \emph{any} $\rho$-replicable algorithm exists for a symmetric property (such as uniformity), a label-invariant tester must also exist.  This fact allows us to remove the assumption on the algorithm type, thereby simplifying the proof extensively.

\uniflb*

\begin{proof} 
The second term in the maximum is necessary for testing an unbiased coin, as shown in \Cref{thm:coin_lb_via_chaining} and~\cite{impagliazzo2022reproducibility}. For the rest of this proof, we focus on the first term in the lower bound which is relevant only when $n \gg 1/\rho^2$.

We aim to apply our lower bound machinery introduced in \Cref{thm:LB_chaining}. We introduce $\DD_0, \ldots, \DD_t$. Let $t = \floor{1/(300\rho)}$. Similar to ~\cite{liu2024replicableuniformity}, we define $p_i$ for each $i \in \{0,1,\ldots, t\}$ as follows: for every $j \in [n]$, the probability of $j$ is given by
\begin{align*}
    p_i(j) \coloneqq \left\{\begin{array}{ll}
         \frac{1+(i\,\epsilon/t)}{n} & \quaaad \text{if $j$ is even\,,} \vspace{2mm}\\ 
         \frac{1-(i\,\epsilon/t)}{n} & \quaaad \text{otherwise,} 
    \end{array}
    \right.
\end{align*}
Each $\DD_i$ is a uniform distribution over all possible permutations of $p_i$.
Clearly, $p_0$ is the uniform distribution and must be accepted, while $p_t$ is $\epsilon$-far from uniform.

Ignoring logarithmic factors, Lemma~4.3 in~\cite{liu2024replicableuniformity} establishes the indistinguishability of consecutive pairs, showing that distinguishing them requires at least $\Tilde{\Omega}(\sqrt{n} \epsilon^{-2} \rho^{-1})$ samples. Plugging this result directly into our lower bound machinery from \Cref{thm:LB_chaining} immediately implies the desired lower bound.

We now tighten their indistinguishability result by applying the ``wishful thinking'' lemma from~\cite{Valiant11}. First, let us define the necessary tools. 
\begin{definition}
The $m$-based moments $M(a)$ of a distribution $p$ are 
\begin{equation*}
    M(m)=m^{a} \sum_{i=1}^n p(i)^a.
\end{equation*}
\end{definition}

\begin{theorem}[\cite{Valiant11}]\label{thm:lb-moments-single}
Suppose we are given two integers $m$ and $n$, and 
two distributions $p^+$ and $p^-$ such that their probabilities are at most $1/(500\,m)$. Also, assume the $m$-based moments of $p^+$ and $p^-$, denoted by $M^+$ and $M^-$, satisfy
\begin{equation}\label{eq:moment_bnd}
    \sum_{a \geq 2} \frac{
        \abs{M^+(a) - M^-(a)}
    }
    {
         \floor{a/2}!\cdot \sqrt{1 + \max\bc{M^+(a), M^-(a)}}
    }
    < \frac{1}{24}\,.
\end{equation}
If a tester exists for a symmetric property $\PP$ that outputs \accept for $p^+$ with at least $2/3$ probability and \reject for $p^-$ with at least $2/3$ probability, then it must use more than $k$ samples. 
\end{theorem}

Let's compute the $m$-based moment of each $p_i$:
\begin{align*}
    M_i(a) \coloneqq m^a \cdot   \left(\frac{n}{2} \cdot \frac{\left(1 + \frac{i\epsilon}{t}\right)^{a}}{n^a}  + \frac{n}{2} \cdot \frac{\left(1 - \frac{i\epsilon}{t}\right)^{a}}{n^a} \right)
\end{align*}

Consider the scenario where
$$\frac{1}{\epsilon^6 \rho^2} \leq n 
\quad
\Longleftrightarrow
\quad 
\frac{1}{\epsilon^3 \rho} \leq \sqrt{n}
\quad
\Longleftrightarrow
\quad 
\frac{\sqrt{n}}{\epsilon^2 \rho } \leq n \epsilon
\,.$$ Let $m$ be the following value: 
$$m = c\cdot\frac{\sqrt{n}}{\epsilon^2\,\rho}\,,$$ for a sufficiently small constant $c$.
Clearly, $m \ll n$ and $\frac{m}{n} \leq \epsilon$.

Now, assume there exists an algorithm that can distinguish $\DD_{i}$ and $\DD_{i-1}$ for some $i \in [t]$ using $m$ samples. If such an algorithm exists, \Cref{thm:lb-moments-single} implies that the sum described in \Cref{eq:moment_bnd} must be lower bounded by a constant:
\begin{align*}
\frac{1}{24} & \leq \sum_{a=2}^\infty \frac{M_i(a) - M_{i-1}(a)}{\floor{a/2}!\cdot \sqrt{1 + \max\left(M_i(a),\ M_{i-1}(a)\right)}}
\\ & \leq \sum_{a=2}^\infty \frac{\frac{n}{2} \cdot \left(\frac{m}{n}\right)^a \cdot \left(\left(1 + \frac{i\epsilon}{t}\right)^{a} + \left(1 - \frac{i\epsilon}{t}\right)^{a} - \left(1 + \frac{(i-1)\epsilon}{t}\right)^{a} - \left(1 - \frac{(i-1)\epsilon}{t}\right)^{a}\right)}{\sqrt{\frac{n}{2} \cdot \left(\frac{m}{n}\right)^a }}    
\end{align*}
In the inequality above, we used that $\max\left(M_i(a),\ M_{i-1}(a)\right)$ is at least $(n/2) \cdot (m/n)^a$. For $a = 2$, the above term is: 
\begin{align*}
    \sqrt{\frac{n}{2}} &\cdot \left(\sqrt{\frac{m}{n}}\right)^2 \left(2 + 2\,\left(\frac{i\epsilon}{t}\right)^2 - 2 - 2\,\left(\frac{(i-1)\,\epsilon}{t}\right)^2\right) = \Theta\left(\frac{m}{\sqrt{n}} \cdot \left(\left(\frac{i\epsilon}{t}\right)^2 - \left(\frac{(i-1)\,\epsilon}{t}\right)^2 \right)\right) \\ & = \Theta \left(\frac{m\,i\,\epsilon^2}{ \sqrt{n}\,t^2}\right) = \Theta \left(\frac{m\,\epsilon^2}{ \sqrt{n}\,t}\right) = \Theta \left(\frac{m\,\epsilon^2 \rho}{ \sqrt{n}}\right)
\end{align*}
For the rest of the terms, we have: 
\begin{align*}
\frac{1}{24} & \leq \Theta \left(\frac{m\,\epsilon^2 \rho}{ \sqrt{n}}\right) + \sqrt{\frac{n}{2}} \ \sum_{a=3}^\infty \cdot \left(\sqrt{\frac{m}{n}}\right)^{a} \cdot \left(
\left(1 + \frac{i\epsilon}{t}\right)^{a} 
+ \left(1 - \frac{i\epsilon}{t}\right)^{a} 
- \left(1 + \frac{(i-1)\epsilon}{t}\right)^{a}
- \left(1 - \frac{(i-1)\epsilon}{t}\right)^{a}
\right)
\end{align*}
The terms in the summation form four geometric series for which we have $\sum_{a\geq3} x^a = x^3/(1-x)$ for any $|x| < 1$. Summing the the geometric series yields:
\begin{equation}\label{eq:large_terms}
= \sqrt{\frac{n}{2}} \cdot \left(\frac{m}{n}\right)^{3/2}\cdot \left( \frac{\left(1 + \frac{i\epsilon}{t}\right)^3}{1 - \sqrt{\frac{m}{n}}\left(1 + \frac{i\epsilon}{t}\right)} 
+ \frac{\left(1 - \frac{i\epsilon}{t}\right)^3}{1 - \sqrt{\frac{m}{n}}\left(1 - \frac{i\epsilon}{t}\right)} 
- \frac{\left(1 + \frac{(i-1)\,\epsilon}{t}\right)^3}{1 - \sqrt{\frac{m}{n}}\left(1 + \frac{(i-1)\,\epsilon}{t}\right)} 
- \frac{\left(1 - \frac{(i-1)\,\epsilon}{t}\right)^3}{1 - \sqrt{\frac{m}{n}}\left(1 - \frac{(i-1)\,\epsilon}{t}\right)}\right)
\end{equation}
Now, consider a function $f$ parameterized by $c \in (0,1/2]$:
$$f_c(x) = \frac{(1-x)^3}{1- c(1+x)} - \frac{(1-x)^3}{1- c(1-x)}\,$$
Clearly, $f_c(0)$ is zero. The derivative is:
$$\frac{df}{dx}(x) = \frac{2c(1 - x)^2\left((1-c)^2 (1-4x) + c^2 x^2 + 2 c^2 x^3\right)}{\left(c(x - 1) + 1\right)^2\,(cx + c - 1)^2},
$$
which is positive for any $x \in [0,0.25]$. Hence, $f_c$ is increasing in this interval. That is, for any $x<x'$ in $[0,0.25]$ $f_c(x) \leq f_c(x')$. Now, we set $x = (i-1)\,\epsilon/t$ and $x' = i\epsilon/t$. Hence, we have
\begin{align*}
    0 & \leq f_{\sqrt{m/n}}\left(\frac{i\epsilon}{t}\right) - f_{\sqrt{m/n}}\left(\frac{(i-1)\,\epsilon}{t}\right) 
    \\ & = 
    \frac{\left(1 - \frac{i\epsilon}{t}\right)^3}{1 - \sqrt{\frac{m}{n}}\left(1 + \frac{i\epsilon}{t}\right)}  
    - 
    \frac{\left(1 - \frac{i\epsilon}{t}\right)^3}{1 - \sqrt{\frac{m}{n}}\left(1 - \frac{i\epsilon}{t}\right)} 
    -
    \frac{\left(1 - \frac{(i-1)\epsilon}{t}\right)^3}{1 - \sqrt{\frac{m}{n}}\left(1 - \frac{(i-1)\epsilon}{t}\right)} 
    + 
    \frac{\left(1 - \frac{(i-1)\epsilon}{t}\right)^3}{1 - \sqrt{\frac{m}{n}}\left(1 - \frac{(i-1)\epsilon}{t}\right)} 
\end{align*}
Using this inequality, we can bound the expression in \Cref{eq:large_terms}:
\begin{align*}
& \leq \frac{m^{3/2}}{\sqrt{2}\, n} \cdot \left( \frac{\left(1 + \frac{i\epsilon}{t}\right)^3 + \left(1 - \frac{i\epsilon}{t}\right)^3}{1 - \sqrt{\frac{m}{n}}\left(1 + \frac{i\epsilon}{t}\right)} 
- \frac{\left(1 + \frac{(i-1)\,\epsilon}{t}\right)^3 + \left(1 - \frac{(i-1)\,\epsilon}{t}\right)^3}{1 - \sqrt{\frac{m}{n}}\left(1 + \frac{(i-1)\,\epsilon}{t}\right)} 
\right)
\\&= \frac{m^{3/2}}{\sqrt{2}\, n} \cdot \left( \frac{2 + 6 \left(\frac{i\epsilon}{t}\right)^2}{1 - \sqrt{\frac{m}{n}}\left(1 + \frac{i\epsilon}{t}\right)} 
- \frac{2 + 6\,\left(\frac{(i-1)\,\epsilon}{t}\right)^2}{1 - \sqrt{\frac{m}{n}}\left(1 + \frac{(i-1)\,\epsilon}{t}\right)} 
\right)
\end{align*}

After algebraic simplification, the expression is:
\begin{align*}
& = \frac{m^{3/2}}{\sqrt{2}\, n}\left[ 
\frac{6\,\left(1-\sqrt{\frac{m}{n}}\right)\left(\left(\frac{i\epsilon}{t}\right)^{2}-\left(\frac{(i-1)\,\epsilon}{t}\right)^{2}\right)
- \sqrt{\frac{m}{n}} \cdot 
\left(
\frac{(i-1)\epsilon}{t} \cdot \left(2 + 6\,\left(\frac{i\epsilon}{t}\right)^{2}\right) -\frac{i\epsilon}{t} \cdot\left(2+6\,\left(\frac{(i-1)\,\epsilon}{t}\right)^{2}\right)
\right)}
{\left(1-\sqrt{\frac{m}{n}}\left(1+\frac{i\epsilon}{t}\right)\right)\left(1-\sqrt{\frac{m}{n}}\left(1+\frac{(i-1)\epsilon}{t}\right)\right)} 
\right]
\\& = \frac{m^{3/2}}{\sqrt{2}\, n}\left[ 
\frac{6\,\left(1-\sqrt{\frac{m}{n}}\right)\left(\left(\frac{i\epsilon}{t}\right)^{2}-\left(\frac{(i-1)\,\epsilon}{t}\right)^{2}\right)
+ \sqrt{\frac{m}{n}} \cdot 
\left(
2-6\,\left(\frac{(i-1)\epsilon}{t}\right) \left(\frac{i\epsilon}{t}\right)\right)
\cdot \left(
\left(\frac{i\epsilon}{t}\right) -\left(\frac{(i-1)\epsilon}{t}\right)
\right)
}
{\left(1-\sqrt{\frac{m}{n}}\left(1+\frac{i\epsilon}{t}\right)\right)\left(1-\sqrt{\frac{m}{n}}\left(1+\frac{(i-1)\epsilon}{t}\right)\right)} 
\right]
\\ & = \Theta\left(\frac{m^{3/2} \epsilon^2}{n\,t} + \frac{m^{2}\,\epsilon}{n^{3/2}\,t}\right) = \Theta\left(\frac{m^{3/2} \epsilon^2}{n\,t} + \frac{m^{2}\,\epsilon}{n^{3/2}\,t}\right)
\end{align*}
Now, adding the term for $a = 2$ leaves us with:
\begin{align*}
    \frac{1}{24} \leq \Theta\left(\sqrt{n} \cdot \left(\left(\sqrt{\frac{m}{n}}\right)^2\,\epsilon^2 \,\rho 
    + 
    \left(\sqrt{\frac{m}{n}}\right)^3\,\epsilon^2 \,\rho 
    + 
    \left(\sqrt{\frac{m}{n}}\right)^4\,\epsilon \,\rho\right)\right)
\end{align*}
Given our assumption, we show that  $m/n \leq \epsilon < 1$. Clearly the second term is dominated by the first term. The third term is also dominated by the first one. Putting this together implies: 

$$\frac{1}{24} \leq \Theta\left(\sqrt{n} \cdot \left(\sqrt{\frac{m}{n}}\right)^2\,\epsilon^2 \,\rho \right) \quad \Longleftrightarrow\quad m = \Omega\left(\frac{\sqrt{n}}{\epsilon^2 \rho}\right)\,.$$
Hence, the proof is complete. 
\end{proof}

\subsubsection{Closeness Testing}

To prove our lower bound, we use the machinery of wishful thinking lemma introduce in~\cite{Valiant11}, and its application in to proving lower bounds for closeness testing in~\cite{chan2014optimal}. We start by defining the $(m,m)$-based moments: 

\begin{definition}
The $(m, m)$ moments $M(r,s)$ of a distribution pair $(p,q)$ are 
\begin{equation*}
    M(r,s)=m^{r+s} \sum_{i=1}^n p_i^r q_i^s.
\end{equation*}
\end{definition}

\begin{theorem}[\cite{Valiant11}]\label{thm:lb-moments}
If distributions $p_1^+, p_2^+, p_1^-, p_2^-$ have probabilities at most $1/1000m$ and their $(m, m)$-based moments $M^+$ and $M^-$ satisfy
\begin{equation*}
    \sum_{r+s \geq 2} \frac{
        \abs{M^+(r,s) - M^-(r,s)}
    }
    {
        \floor{r/2}!\floor{s/2}! \sqrt{1 + \max\bc{M^+(r,s), M^-(r,s)}}
    }
    < \frac{1}{360},
\end{equation*}
then the distribution pair $(p_1^+, p_2^+)$ cannot be distinguished with probability $13/24$ from $(p_1^-, p_2^-)$ by a tester for a symmetric property that takes $\Poi(m)$ samples from each distribution.
\end{theorem}

\closelb*

\begin{proof}
The second and third terms in the lower bound are implied by uniformity testing \Cref{thm:LB_unif}, and coin testing~\cite{impagliazzo2022reproducibility}, respectively.

It remains to consider the first term. Note that this term is only relevant when
\begin{equation}\label{eq:lb-n}
    \frac{n^{2/3}}{\eps^{4/3} \rho^{2/3}} \gg \frac{\sqrt{n}}{\eps^2 \rho}
    \iff n^{1/6} \gg \frac{1}{\eps^{2/3} \rho^{1/3}}
    \iff n \gg \frac{1}{\eps^4 \rho^{2}}.
\end{equation}

We will construct a chain of $t = \floor{1/(300\rho)}$ distributions that are pairwise hard to distinguish and invoke~\Cref{thm:LB_chaining}. The construction of our lower bound is similar to \cite{chan2014optimal}.
Let $b = \eps^{4/3} \rho^{2/3} /n^{2/3}$ and let $a = 4/n$.
Note that $b \geq a$ by \cref{eq:lb-n}:
\begin{equation*}
    b = \frac{\eps^{4/3} \rho^{2/3}}{n^{2/3}}
    \geq \frac{4 \eps^{4/3} \rho^{2/3}}{n^{2/3} \p{n \eps^{4} \rho^{2}}^{1/3}}
    = \frac{4}{n} = a.
\end{equation*}
Observe that both $1/a$ and $1/b$ are much larger than $t =\Theta(1/\rho)$. Hence, without loss of generality assume $(1-\eps)/b$ and $1/a$ are both integer and divisible by $t$. This will not affect our lower bound by more than a constant factor.

Let $A$, $B$, and $D$ be three disjoint sets of size $(1-\eps)/b$, $1/a$, and $1/a$ respectively that are subsets of the domain $[n]$. Partition $B$ and $D$ into $t$ sets of equal sizes, namely $B = \{B_i\}_{i=1}^t$, and $D= \{D_i\}_{i=1}^t$.
Now, we define a series of sets $C_0, C_1, \ldots, C_t$. We define $C_i$ to be a set of $1/a$, disjoint from $A$ but overlapping with $B$ and $D$:
$$C_i \coloneqq  \left(\bigcup_{j={i+1}}^t B_j \right) \cup \left(\bigcup_{j=1}^i D_j \right) \,.$$ 
Define pairs of difficult distributions:
\begin{equation*}
    p = b \mathbf{1}_A + \eps a \mathbf{1}_B
    \qquad
    q_i =  b \mathbf{1}_{A} + \eps a \mathbf{1}_{C_i}
\end{equation*}

Note that the pair, $q_0 = p$, and $q_t$ is the standard hard example from \cite{chan2014optimal}.
Note that
\begin{equation*}
    \|p - q_i\|_{TV}
    = \frac{1}{2} \sum_{j \in B \triangle C} \eps a
    = \frac{i\,\abs{B}}{t} \cdot \eps\,a = \frac{i}{t} \cdot \eps
\end{equation*}

Let $m = \frac{c_1 n^{2/3}}{\eps^{4/3} \rho^{2/3}}$ for some constant $c_1$. We wish to show that this is an insufficient number of samples for distinguishing two consecutive pairs. Note that the maximum probability mass of $p$ or any $q_i$ is $b < 1/1000m$ as long as $c_1$ is small enough (this is a condition of \cref{thm:lb-moments}). 

Next, we aim to show that for any $i \in [t]$ many samples are required to distinguish between the pair $(p_1^+, p_2^+) = (p, q_{i})$ and the pair $(p_1^-, p_2^-) = (p, q_{i+1})$.

Consider the $(m,m)$ moments of both pairs of distributions using $u = r+s$, and set $\alpha = i/t$ and $\Delta = 1/t$:
\begin{align*}
    M^+(r,s)
    &= m^u \p{\frac{1-\eps}{b}} b^u
    + m^u \p{\frac{1-\alpha}{a}} \eps^u a^u \\
    &= m^u \p{1-\eps} b^{u-1}
    + m^u \p{1-\alpha} \eps^u a^{u-1} \\
    &= m^u\p{(1-\eps)\frac{\eps^{4u/3 - 4/3} \rho^{2u/3 - 2/3}}{n^{2u/3 - 2/3}} + (1-\alpha)\frac{4^{u-1}\eps^u}{n^{u-1}}}
\end{align*}
and
\begin{equation*}
    M^-(r,s)
    = m^u\p{(1-\eps)\frac{\eps^{4u/3 - 4/3}\rho^{2u/3 - 2/3}}{n^{2u/3 - 2/3}} + (1-\alpha - \Delta)\frac{4^{u-1}\eps^u}{n^{u-1}}}.
\end{equation*}

Recall that $\Delta = \Theta(\rho)$. We will focus on the following key term in \cref{thm:lb-moments}: 
\begin{align*}
    \frac{\abs{M^+(r,s) - M^-(r,s)}}{\sqrt{1 + \max\bc{M^+(r,s), M^-(r,s)}}}
    &= \frac{m^u \Delta 4^{u-1} \eps^u/n^{u-1}}{\sqrt{1 + m^u(1-\eps) \eps^{4u/3 - 4/3} \rho^{2u/3 - 2/3}/n^{2u/3 - 2/3} + m^u(1-\alpha) 4^{u-1} \eps^u / n^{u-1} }} \\
    &\leq \frac{m^u \Delta 4^{u-1} \eps^u/n^{u-1}}{\sqrt{m^u(1-\eps) \eps^{4u/3 - 4/3} \rho^{2u/3 - 2/3}/n^{2u/3 - 2/3}}} \\
    &\leq O\p{\frac{m^{u/2} \rho 4^{u-1} \eps^{u/3 + 2/3}}{n^{2u/3-2/3} \rho^{u/3 - 1/3}}} \\
    &\leq O\p{\frac{m^{u/2} 4^{u-1} \eps^{u/3 + 2/3}}{n^{2u/3-2/3} \rho^{u/3 - 4/3}}}.
\end{align*}

Note the following inequality which is a consequence of \cref{eq:lb-n} when $u \geq 2$:
\begin{equation*}
    n \gg \frac{1}{\eps^4 \rho^{2}}
    \implies \frac{1}{n} \leq \frac{\eps \rho^{2}}{4}
    \implies \frac{1}{n^{u/3 - 2/3}} \leq \frac{\eps^{u/3-2/3} \rho^{2u/3 - 4/3}}{4^{u/3-2/3}}.
\end{equation*}

Using this inequality,
\begin{align*}
    \frac{\abs{M^+(r,s) - M^-(r,s)}}{\sqrt{1 + \max\bc{M^+(r,s), M^-(r,s)}}}
    &\leq O\p{\frac{m^{u/2} 4^{u-1} \eps^{u/3 + 2/3}}{n^{u/3} \rho^{u/3 - 4/3}} \p{\frac{\eps^{u/3-2/3} \rho^{2u/3-4/3}}{4^{u/3-2/3}}}} \\
    &= O\p{ \frac{m^{u/2} 4^{2u/3-1/3} \eps^{2u/3} \rho^{u/3}}{n^{u/3}}} \\
    &= O\p{c_1^{u/2} 4^{2u/3-1/3}} =  O\p{ \p{4^{2/3} \cdot \sqrt{c_1}}^u}
\end{align*}
For a sufficiently small $c_1$, the above function is exponentially decreasing in $u$. Then, one can argue the sum of all moments converges to a value less than $1/360$ as desired due to the convergence of the geometric series. 
The result follows directly from the argument in \cite{chan2014optimal} (see Proposition 9).
\end{proof}

\section{Expectation-Gap Replicable Testers}\label{sec:expectation_gap_tester}

We generalize and quantitatively improve bounds for replicable expectation-gap estimators.

\subsection{General Expectation-Gap Estimator}

\begin{definition}[Expectation-Gap Statistic]\label{def:expectation-gap}
Consider a hypothesis testing problem with a null hypothesis $H_0$ (e.g., the distribution belongs to a property) and alternate hypothesis $H_1$ (the distribution is far from the property).
Let $\ns$ be the number of samples taken from the distribution.
Then, an expectation-gap statistic is defined by  a test statistic $Z(\ns)$ of the samples with the following properties given by real-valued functions $\tau_0(\ns), \tau_1(\ns), \sigma(\ns)$ defined as follows:
\begin{itemize}
    \item Null threshold (upper bound): $\Ema{Z(\ns) | H_0} \leq \tau_0(\ns)$.
    \item Alternative threshold (lower bound): $\Ema{Z(\ns) | H_1} \geq \tau_1(\ns)$. We require that for some $\nsmin \in \N$, for all $\ns \geq \nsmin$, $\tau_1(\ns) \geq \tau_0(\ns)$.
    \item Variance upper bound: $\sqrt{\Varma{Z(\ns)}} \leq \sigma(\ns)\p{1 + \max\bc{0, \frac{\Ema[]{Z(\ns)} - \tau_1(\ns)}{\Delta(\ns)}, \frac{\tau_0(\ns) - \Ema[]{Z(\ns)}}{\Delta(\ns)}}}$.
    If $\Ema[]{Z(\ns)} \in [\tau_0(\ns), \tau_1(\ns)]$, the condition is simply $\sqrt{\Varma{Z(\ns)}} \leq \sigma(\ns)$.\footnote{Importantly, we do not simply use a uniform bound on the variance. Numerous applications of expectation-gap style analyses allow for the variance to increase when the expectation of the statistic is far from the interval between the null and alternate thresholds.}
\end{itemize}
The following quantities summarize the important properties of the test statistic $Z$:
\begin{itemize}
    \item Threshold gap: $\Delta(\ns) = \tau_1(\ns) - \tau_0(\ns)$. Note that this is a lower bound on the true gap $\Ema{Z(\ns)|H_1} - \Ema{Z(\ns)|H_0}$.
    \item Noise-to-signal ratio: $f(\ns) = \frac{\sigma(\ns)}{\Delta(\ns)}$.
    \item Sampling breakpoints $\ns_t$ defined for any $t \in (0,1]$:
    \begin{equation}\label{eq:sampling-breakpoint}
        \ns_t = \min \{\ns \in \N: \ns \geq \nsmin \text{ and } f(\ns) \leq t/2\}.
    \end{equation}
\end{itemize}
\end{definition}

A core primitive will be bounding deviations of the expectation-gap statistic via Chebyshev's inequality.
\begin{lemma}\label{lem:chebyshev-expectation-gap}
Let $Z(\ns), \tau_0(\ns), \tau_1(\ns), \sigma(\ns)$ be the parameters of an expectation-gap statistic.
Consider any $\ns \geq \nsmin$ and $\alpha \in [0,1]$.
If $\Ema[]{Z(\ns)} \in [\tau_0(\ns),\tau_1(\ns)]$, then
\begin{equation*}
    \Prma{\abs{Z(\ns) - \Ema[]{Z(\ns)}} \geq \alpha \Delta(\ns)} \leq \frac{f(\ns)^2}{\alpha^2}.
\end{equation*}
If $\Ema[]{Z(\ns)} > \tau_1(\ns)$, then
\begin{equation*}
    \Prma{Z(\ns) \leq \tau_1(\ns) - \alpha \Delta(\ns)} \leq \frac{f(\ns)^2}{\alpha^2}.
\end{equation*}
If $\Ema[]{Z(\ns)} < \tau_0(\ns)$, then
\begin{equation*}
    \Prma{Z(\ns) \geq \tau_0(\ns) + \alpha \Delta(\ns)} \leq \frac{f(\ns)^2}{\alpha^2}.
\end{equation*}
\end{lemma}

\begin{proof}
First, consider the case where $\Ema{Z(\ns)} \in [\tau_0(\ns), \tau_1(\ns)]$.
By Chebyshev's inequality:
\begin{align*}
    \Prma{\abs{Z(\ns) - \Ema[]{Z(\ns)}} \geq \alpha \Delta(\ns)}
    &\leq \frac{\Varma{Z(\ns)}}{\alpha^2 \Delta(\ns)^2} \\
    &= \frac{f(\ns)^2}{\alpha^2}.
\end{align*}
Now, consider the case where $\Ema{Z(\ns)} > \tau_1(\ns)$. The case where the expectation is less than $\tau_0(\ns)$ is symmetric.
\begin{align*}
    \Prma{Z(\ns) \leq \tau_1(\ns) - \alpha \Delta(\ns)}
    &\leq \Prma{\abs{Z(\ns) - \Ema{Z(\ns)}} \geq \alpha \Delta(\ns) + \Ema{Z(\ns)} - \tau_1(\ns)} \\
    &\leq \frac{\Varma{Z(\ns)}}{\Delta(\ns)^2\p{\alpha + \frac{\Ema{Z(\ns)} - \tau_1(\ns)}{\Delta(\ns)}}^2} \\
    &\leq \frac{\sigma(\ns)^2 \p{1 + \frac{\Ema{Z(\ns)} - \tau_1(\ns)}{\Delta(\ns)}}^2}{\Delta(\ns)^2\p{\alpha + \frac{\Ema{Z(\ns)} - \tau_1(\ns)}{\Delta(\ns)}}^2} \\
    &\leq \frac{f(\ns)^2}{\alpha^2}.
\end{align*}
\end{proof}

When $\Ema{Z(\ns)} \in [\tau_0(\ns), \tau_1(\ns)]$, at the sampling breakpoints,
\begin{equation}\label{eq:sampling-breakpoint-chebyshev}
    \Prma[]{\abs{Z(\ns_t) - \Ema{Z(\ns_t)}} \geq t \Delta(\ns_t)}
    \leq \frac{1}{4}.
\end{equation}
For example, taking $\ns_{0.5}$ samples and thresholding at $\tau_0(\ns_{0.5}) + \Delta(\ns_{0.5})/2$ would yield a standard non-replicable tester with constant success probability.

\begin{example}
In the (unbiased) coin testing problem, samples are drawn from a $\Ber(p)$ distribution. The null hypothesis is that $p = 1/2$ and the alternative hypothesis is that $p \geq 1/2 + \eps$ for a parameter $\eps > 0$.
The test statistic $Z(\ns)$ is the number of sampled elements which are heads.
\begin{itemize}
    \item A valid null and alternative threshold are $\tau_0(\ns) = \ns/2$ and $\tau_1(\ns) = \ns/2 + \eps \ns$, respectively, with $\nsmin = 0$.
    \item A variance upper bound is given by $\sigma(\ns) = \sqrt{\ns}/2$.
    \item The threshold gap and noise-to-signal ratio are $\Delta(\ns) = \eps \ns$ and $f(\ns) = \frac{1}{2\eps \sqrt{\ns}}$, respectively.
    \item The sampling breakpoints are therefore $\ns_t = \ceil{\frac{1}{t^2\eps^2}}$. 
\end{itemize}
\end{example}

\begin{enumeratebox}[alg:general-estimator]{General Expectation-Gap Estimator}
\begin{enumerate}
    \item Pick $t \in [\rho, 1/16]$.
    \item Sample $r \sim \Unif\p{\br{\frac{1}{4}, \frac{3}{4}}}$.
    \item Repeat $L=O\p{\frac{t^2}{\rho^2}}$ times: take $\ns_t$ samples and compute the statistic $\hat{Z}(\ns_t)_\ell$ for $\ell \in [L]$.
    \item Compute the median estimate $\hat{V} = \text{median}\p{\hat{Z}(\ns_t)_1, \ldots, \hat{Z}(\ns_t)_L}$.
    \item Output \accept if $\hat{V} < \tau_0(\ns_t) + \Delta(\ns_t)/8$ and \reject if $\hat{V} > \tau_1(\ns_t) - \Delta(\ns_t)/8$. Otherwise, continue.
    \item Compute the mean estimate $\hat{W} = \frac{1}{L} \sum_{\ell=1}^L \hat{Z}(\ns_t)_\ell$.
    \item Output \accept if $\hat{W} \leq \tau_0(\ns_t) + r\Delta(\ns_t)$ and \reject otherwise.
\end{enumerate}
\end{enumeratebox}

\begin{theorem}\label{thm:general-estimator}
Given parameters $0 \leq \delta \leq \rho \leq 1$,
 and $t \in [\rho, 1/16]$, as well as a constant $C$ and given a hypothesis testing problem $(H_0, H_1)$ and statistic with $Z(\ns)$ with $\tau_0(\ns), \tau_1(\ns), \sigma(\ns)$, there exists an Algorithm(\cref{alg:general-estimator}) with the following properties: 

\begin{itemize}
    \item It takes $\ns = O\p{\frac{\ns_t t^2}{\rho^2}}$ samples.
    \item It succeeds with probability at least $1-t^{C t^2/\rho^2}$.
    \item The algorithm is $\rho$-replicable.
\end{itemize}
\end{theorem}

\begin{remark}
For any application, $t$ should be optimized given the sampling breakpoints of the given estimator and desired sample complexity/failure probability tradeoff. In the two extremes of the setting of $t$, the failure probability can be as small as $\exp\p{-\Omega(1/\rho^2)}$ when $t=O(1)$ or as large as $\rho^C$ when $t=\rho$. In either case, the algorithm succeeds with high probability in $1/\rho$.
\end{remark}

\begin{proof}[Proof of \cref{thm:general-estimator}]
The sample complexity of the algorithm is immediate.
Let $X$ be a random variable representing the set of samples collected by the algorithm. Let $C'$ be such that $L \leq C't^2/\rho^2$.

\paragraph{Correctness}
Correctness of the algorithm is guaranteed by the median estimate.
Let $z = \Ema[]{Z(\ns_t)}$ be the expectation of the statistic.
Recall that in the null hypothesis, $z \leq \tau_0(\ns_t)$, and in the alternate hypothesis, $z \geq \tau_1(\ns_t)$.
The algorithm is correct if it outputs \accept and \reject in these two cases, respectively.

We will show that the median estimate $\hat{V}$ is contained within an interval of length $\Delta(\ns_t)/4$ around $z$ with high probability.
Recall from \cref{lem:chebyshev-expectation-gap}, we get Chebyshev-style bounds within the interval $[\tau_0(\ns_t), \tau_1(\ns_t)]$ regardless of the location of $\Ema{Z(\ns)}$. We will only be concerned with deviations within this interval and assume without loss of generality that $\Ema{Z(\ns_t)} \in [\tau_0(\ns_t), \tau_1(\ns_t)]$.
By \cref{lem:chebyshev-expectation-gap} and the definition of $\ns_t$ (\cref{eq:sampling-breakpoint}):
\begin{equation*}
    \Prma[]{\abs{Z(\ns_t) - z} \geq \Delta(\ns_t)/8}
    \leq 64 f(\ns_t)^2
    \leq 16 t^2.
\end{equation*}
As we choose $t \leq 1/16$, this probability is at most $1/16$.

By a standard median analysis via a Chernoff bound, the probability that $\hat{V}$ deviates from $z$ decays exponentially:
\begin{align*}
    \Prma{\abs{\hat{V} - z} \geq \Delta(\ns_t)/8}
    &\leq \Prma{\sum_{\ell=1}^L \mathbf{1}\br{\abs{\hat{Z}(\ns_t)_\ell - z} \geq \Delta(\ns_t)/8} \geq L/2} \\
    &\leq \Prma[A \sim \Bin(L, 16t^2)]{\abs{A - \Ema{A}} \geq L/4} \\
    &= \Prma[A \sim \Bin(L, 16t^2)]{\abs{A - \Ema{A}} \geq \frac{1}{64t^2} \Ema{A}} \\
    &\leq \p{\frac{e^{1/64t^2}}{(1+1/64t^2)^{1+1/64t^2}}}^{16t^2 L} \\
    &\leq \p{\frac{e}{1 + 1/64t^2}}^{\frac{C' t^2}{4\rho^2}}.
\end{align*}

For any constant $C$ as given in the theorem statement, for a large enough constant $C'$, the probability that the median $\hat{V}$ deviates from $z$ by more than $\Delta(\ns_t)/8$ can be upper bounded by $t^{C t^2/\rho^2}$.
By thresholding above at $\tau_0(\ns_t) + \Delta(\ns_t)/8$ and below at $\tau_1(\ns_t) - \Delta(\ns_t)/8$, we ensure correctness under $H_0$ and $H_1$ with probability $1 - t^{C t^2/\rho^2}$.

\paragraph{Replicability}
Replicability must be satisfied even if neither $H_0$ nor $H_1$ holds.
In this case, the standard hypothesis testing definition does not specify a correct answer between \accept and \reject.
In order to prove replicability, we will chose a correct answer as follows.
We say that the algorithm should \accept if $z \leq \tau_0 + r\Delta(\ns_t)$ and \reject otherwise.
We will show that, with probability $1-\rho$ over the randomness of $r$ and the sampling process, the algorithm will return the ``correct'' answer with probability $1-\rho$.
By union bound over resampling with the same draw of $r$, the algorithm is $2\rho$-replicable. We will proceed by cases over the location of $z$.

First, consider the case that the algorithm terminates on the median comparison of $\hat{V}$.
By the correctness analysis, $\hat{V}$ is concentrated within $\Delta(\ns_t)/8$ of $z$ with high probability in $1/\rho$.
Therefore, the median comparison will only \accept if $z < \tau_0(\ns_t) + \Delta(\ns_t)/4$ and will only \reject if $z > \tau_1(\ns_t) - \Delta(\ns_t)/4$.
Since we choose $r \in \br{\frac{1}{4}, \frac{3}{4}}$, this step will only terminate with the correct answer.

If the median comparison does not return decisively, the output of the algorithm is determined by the mean comparison of $\hat{W}$.
Note that $\hat{W}$ is the mean of $L$ independent estimators with expectation $z$ and standard deviation at most $\sigma(\ns_t) \leq f(\ns_t) \Delta(\ns_t) \leq t \Delta(\ns_t)/2$ (the last step uses the definition of $\ns_t$ in \cref{eq:sampling-breakpoint}).
It follows that the mean of $\hat{W}$ is $z$ and its variance is $t^2\Delta(\ns_t)^2/4L$.
By Chebyshev's inequality,
\begin{equation} \label{eq:hatw-chebyshev}
    \Prma[X]{\abs{\hat{W} - z} \geq \alpha \Delta(\ns_t)}
    \leq \frac{t^2\Delta(\ns_t)^2}{4L \alpha^2\Delta(\ns_t)^2}
    = \frac{\rho^2}{4 C' \alpha^2}.
\end{equation}

The algorithm outputs the incorrect answer if $\hat{W}$ is on the wrong side (i.e., not the same side as $z$) of the random threshold defined by $r$. The probability that this occurs is upper bounded by the probability that $\hat{W}$ deviates from its expectation by more than the distance between $\tau_0(\ns_t) + r\Delta(\ns_t)$ and $z$.
Let $D(r) = \abs{\tau_0(\ns_t) + r\Delta(\ns_t) - z}$ be a random variable (over the randomness of $r$) for this distance.
As $r$ is sampled uniformly in $\br{\frac{1}{4}, \frac{3}{4}}$ independently of $z$, the probability that $D(r) \leq \beta\Delta(\ns_t)$ is upper bounded by the probability that $r$ lands in an interval of length $2\beta$.
This probability is at most $4\beta$.
Then, the probability of a replicability failure is upper bounded by:
\begin{align*}
    \Prma[X, r]{\abs{\hat{W} - z} \geq D(r)}
    &\leq \sum_{k=1}^{\infty} \Prma[X, r]{D(r) \in \br{2^{-(k+1)}, 2^{-k}} \Delta(\ns_t) \text{ and } \abs{\hat{W} - z} \geq D(r)} \\
    &\leq \sum_{k=1}^{\infty} \Prma[X, r]{D(r) \in \br{2^{-(k+1)}, 2^{-k}} \Delta(\ns_t) \text{ and } \abs{\hat{W} - z} \geq 2^{-(k+1)}\Delta(\ns_t)} \\
    &= \sum_{k=1}^{\infty} \Prma[r]{D(r) \in \br{2^{-(k+1)}, 2^{-k}} \Delta(\ns_t)} \Prma[X]{\abs{\hat{W} - z} \geq 2^{-(k+1)}\Delta(\ns_t)} \\
    &\leq \sum_{k=1}^{\infty} 2^{-k+2} \Prma[X]{\abs{\hat{W} - z} \geq 2^{-(k+1)}\Delta(\ns_t)} \\
    &\leq \sum_{k=1}^\infty 2^{-k+2} \min\bc{\frac{\rho^2}{4C'2^{-2(k+1)}}, 1} \tag{by \cref{eq:hatw-chebyshev}} \\
    &= \sum_{k=1}^\infty \min\bc{\frac{2^k \rho^2}{4C'}, 2^{-k+2}}. 
\end{align*}

Consider the case when the second term in the minimization dominates:
\begin{equation*}
    \frac{2^k \rho^2}{4C'} \geq 2^{-k+2}
    \iff 2^{2k} \geq \frac{16C'}{\rho^2}
    \iff k \geq \frac{1}{2} \lg \frac{16C'}{\rho^2}.
\end{equation*}

Then, the probability of a replicability failure is at most
\begin{align*}
    \Prma[X, r]{\abs{\hat{W} - z} \geq D(r)}
    &\leq \sum_{k=1}^{\ceil{\frac{1}{2} \lg \frac{16C'}{\rho^2}} - 1} \frac{2^k \rho^2}{4C'} + \sum_{k=\ceil{\frac{1}{2} \lg \frac{16C'}{\rho^2}}}^\infty 2^{-k+2} \\
    &\leq \sqrt{\frac{16C'}{\rho^2}}\p{\frac{\rho^2}{4C'}} + 8\sqrt{\frac{\rho^2}{16C'}} \\
    &= \frac{3\rho}{\sqrt{C'}}.
\end{align*}

Choosing a large enough constant $C'$ suffices for $\rho$-replicability.
\end{proof}

\subsection{Size-Invariant Expectation-Gap Estimator}

In this subsection, we define a special type of expectation-gap statistics (generalizing coin testers and collision-based testers) where, under some normalization, the expectation of the statistic as well as the null and alternate thresholds are constant with respect to the number of samples.
For such statistics, we design a general estimator which also gives improved bounds on the number of samples taken \emph{in expectation}.

\begin{definition}[Size-Invariant Expectation-Gap Statistics]\label{def:size-invariant}
An expectation-gap statistic defined by $Z(\ns),\allowbreak \tau_0(\ns), \tau_1(\ns), \sigma(\ns)$ (see \cref{def:expectation-gap}) is ``size-invariant'' if there exist fixed values $z, \tau_0, \tau_1$ such that for all $\ns \in \N$, $\Ema{Z(\ns)} = z$, $\tau_0(\ns) = \tau_0$, and $\tau_1(\ns) = \tau_1$.
In words, the location of the test statistic as well as the expectation thresholds do not vary with the number of samples.
We parameterize such a statistic by $Z(\ns), \tau_0, \tau_1, \sigma(\ns)$ and define $\Delta$, $f(\ns)$, and sampling breakpoints $\ns_t$ analogously to \cref{def:expectation-gap}.
\end{definition}

\begin{example}
For the coin testing problem, a size-invariant expectation gap statistic is given by $Z(\ns)$ being the fraction of heads in the sample, $\tau_0=1/2$, $\tau_1=1/2+\eps$, and $\sigma(\ns)=\frac{1}{2\sqrt{\ns}}$.
\end{example}

\begin{theorem}\label{thm:size-invariant-estimator}
Given parameters $0 \leq \delta \leq \rho \leq 1$, \cref{alg:size-invariant-estimator} for a hypothesis testing problem $(H_0, H_1)$ with a size-invariance expectation gap statistic defined by $Z(\ns), \tau_0, \tau_1, \sigma(\ns)$ and a sequence of breakpoints given by $t_1, \ldots, t_K$ has the following properties:
\begin{itemize}
    \item The algorithm succeeds with probability at least $1-\delta$.
    \item The algorithm is $O(\rho)$-replicable.
    \item The algorithm takes
    \begin{equation*}
        O\p{\ns_{1/8} \log(1/\delta)
        + \sum_{k=1}^{\ceil{\lg(1/\rho)}} 2^{k} (K-k+1) t_k^2 \ns_{t_k}}
    \end{equation*}
    samples in expectation, and
    \begin{equation*}
        O\p{\ns_{1/8} \log(1/\delta)
        + \sum_{k=1}^{\ceil{\lg(1/\rho)}} 2^{2k} (K-k+1) t_k^2 \ns_{t_k}}
    \end{equation*}
    samples in the worst-case.
\end{itemize}
\end{theorem}

\begin{enumeratebox}[alg:size-invariant-estimator]{Size-Invariant Expectation-Gap Estimator}
\begin{enumerate}
    \item Repeat $L=\ceil{8 \ln(1/\delta)}$ times: take $\ns_{1/8}$ samples and compute the statistic $\hat{Z}(\ns_{1/8})$. Compute the median of these estimates. Output \accept if the median is less than $\tau_0 + \Delta/8$ and \reject if it is greater than $\tau_1 - \Delta/8$. Otherwise, continue.
    \item Pick $r \sim \Unif\p{\br{\frac{1}{4}, \frac{3}{4}}}$.
    \item Let $K=\ceil{\lg(1/\rho)}$. For $k \in [K]$:
    \begin{enumerate}[label=(\alph*)]
        \item Choose $t_k \in [2^{-k}, 1/2]$.
        \item Repeat $J_k=\ceil{t_k^2 2^{2k}}$ times: take $\ns_{t_k}$ samples and compute the test statistic $\hat{Z}(\ns_{t_k})$. Call the mean of these test statistics $\hat{W}_k$.
        \item Repeat the preceding step $L_k = 16(K-k+1)$ times and take the median-of-means estimate of the $\hat{W}_k$'s, call this quantity $\hat{V}_k$.
        \item Output \accept if $\hat{V}_k \leq \tau_0 + \p{r - 2^{-k}}\Delta$ and \reject if $\hat{V}_k \geq \tau_0 + \p{r + 2^{-k}}\Delta$. Otherwise, continue.
    \end{enumerate}
    \item Output \accept.
\end{enumerate}
\end{enumeratebox}

The flexibility of choosing $t_1, \ldots t_K$ allows the estimator to leverage an analysis of estimator's variance beyond the black-box sample mean bounds (see new results on uniformity testing in \cref{sec:unif-testing-upper-bound}).
We present the following corollary which gives a simpler expression when no such analysis exists (or when the basic analysis is tight as in coin testing in \cref{sec:coin-testing-upper-bound}).

\begin{corollary}\label{cor:scale-invariant-simple}
Given parameters $0 \leq \delta \leq \rho \leq 1$, there exists an algorithm for a hypothesis testing problem $(H_0, H_1)$ with a size-invariant expectation gap statistic defined by $Z(\ns), \tau_0, \tau_1, \sigma(\ns)$ which is $\rho$-replicable, fails with probability at most $\min\{\delta, \exp(-1/\rho)\}$, and takes samples $O\p{\frac{\ns_{1/8}}{\rho} + \ns_{1/8}\log(1/\delta)}$ in expectation and $O\p{\frac{\ns_{1/8}}{\rho^2} + \ns_{1/8}\log(1/\delta)}$ in the worst-case.
\end{corollary}

\begin{proof}
The corollary follows from applying \cref{thm:size-invariant-estimator} with $t_1, \ldots, t_K$ all set to $1/8$ and $\delta = \exp\p{1/\rho}$.
As $t_k$ and $\ns_{t_k}$ are fixed, the summations in the sample complexity (both expected and worst-case) are dominated by the final term with $k=\ceil{\log(1/\rho)}$.
\end{proof}

\begin{proof}[Proof of \cref{thm:size-invariant-estimator}]
Let $X$ be a random variable representing the samples collected by the algorithm.
Recall that $z = \Ema{Z(\ns)}$ is a constant due to the size-invariant property of the statistic.

\paragraph{Correctness}
The correctness of the algorithm is achieved in the first step. Note that correctness only needs to hold in the case that the true distribution belongs to the null hypothesis $H_0$ or the alternate hypothesis $H_1$. Assume without loss of generality that we are in the former case.

Recall from \cref{lem:chebyshev-expectation-gap}, we get Chebyshev-style bounds within the interval $[\tau_0, \tau_1]$ regardless of the location of $z$. We will only be concerned with deviations within this interval and assume without loss of generality that $z \in [\tau_0, \tau_1]$. By \cref{lem:chebyshev-expectation-gap} and \cref{eq:sampling-breakpoint-chebyshev}, a single estimate $\hat{Z}(\ns_{1/8})$ will deviate from its expectation $\E[Z(\ns_{1/8} | H_0] \leq \tau_0$ by more than $\Delta/8$ with probability at most $1/4$. By the standard median-of-means analysis, the probability that the median of $L$ estimates of $\hat{Z}(\ns_{1/8})$ will deviate from the expectation by more than $\Delta/8$ is at most the probability that the sum of $L$ $\Ber(1/4)$ i.i.d.\ random variables exceeds $L/2$.
By Hoeffding's bound, this occurs with probability at most $\exp\p{-\frac{2(L/4)^2}{L}} = \exp\p{-L/8}$.
By choosing $L = 8\ln(1/\delta)$ and outputting \accept if the median-of-means estimate is less than $\tau_0(\ns_{1/8}) + \Delta(\ns_{1/8})/8$, the algorithm only fails (does not output \accept) with probability at most $\delta$.

\paragraph{Sample Complexity} The worst-case sample complexity comes from a simple summation of the number of samples used if the algorithm does not terminate early.

For the analysis of the expected sample size, we will need to analyze the quality of the median-of-means estimates $\hat{V}_k$.
First, consider a single mean estimate $\hat{W}_k$.
As it is the mean of independent random variables with standard deviation $\sigma(\ns_t) \leq t\Delta/2$, the expectation of $\hat{W}_k$ is $z$ and its variance is $t^2\Delta^2/4J_k$. Chebyshev's inequality implies that
\begin{equation*}
    \Prma[X]{\abs{\hat{W}_k - z} \geq 2^{-k} \Delta}
    \leq \frac{t^2 \Delta^2 2^{2k}}{4J_k \Delta^2}
    \leq \frac{t^2 2^{2k}}{4(t_k^2 2^{2k})}
    = 1/4.
\end{equation*}
By the standard median-of-means Hoeffding bound for the median of $L_k$ such estimates,
\begin{equation}\label{eq:median-of-means-hoeffding}
    \Prma[X]{\abs{\hat{V}_k - z} \geq 2^{-k} \Delta}
    \leq \exp\p{-L_k/8}.
\end{equation}

Consider the random variable $D(r) = \abs{z - (\tau_0 + r\Delta)}$ which is the distance between the random threshold and the expectation of the statistic.
Let $A_k$ be a binary random variable for the event that the algorithm has not terminated at the end of step $k$.
The algorithm will have terminated at the end of step $k$ if the estimate $\hat{V}_{k'}$ is far from the random threshold $r$ for any $k' \leq k$.
The quantity of how far $\hat{V}_{k'}$ must be from $r$ in order to terminate decreases with $k'$.
The probability of the event $A_k$ is upper bounded by:
\begin{align*}
    \Prma[X,r]{A_k}
    &\leq \Prma[X,r]{\abs{\hat{V}_k - (\tau_0 + r\Delta)} < 2^{-k}\Delta)} \\
    &\leq \sum_{k'=1}^k \p{\Prma[r]{D(r) \in \br{2^{-k'}, 2^{-k'+1}}\Delta} \prod_{\ell=k'}^k \Prma[X]{\abs{\hat{V}_\ell - z]} > 2^{-k'}\Delta)}} + \Prma[r]{D(r) > 2^{-k+1}}.
\end{align*}
The inequality follows from conditioning on the geometric interval (defined by $k'$) which contains $D(r)$.
Conditioned on this interval, the algorithm only does not terminate by step $k$ only if every estimate of $\hat{V}_\ell$ for $\ell \leq k$ deviates from its expectation by more than the lower endpoint of this interval, which is $2^{-k'}\Delta$.

Recall from the proof of \cref{thm:general-estimator} that the probability that $D(r) \leq \beta \Delta$ is upper bounded by $4\beta$ (simply from the probability of the uniform random variable $r$ landing in a $2\beta$ sized interval).
Combined with \cref{eq:median-of-means-hoeffding},
\begin{align}
    \Prma[X,r]{A_k}
    &\leq \sum_{k'=1}^k \p{2^{-k'+3} \prod_{\ell=k'}^k \exp(-L_k/8)} + 2^{-k+3} \nonumber \\
    &\leq \sum_{k'=1}^k \p{2^{-k'+3} \exp\p{- \sum_{\ell=k'}^k L_k/8}} + 2^{-k+3} \nonumber \\
    &\leq \sum_{k'=1}^k \p{2^{-k'+3} 2^{-2(K-k'+1)}} + 2^{-k+3} \nonumber \\
    &= \sum_{k'=1}^k 2^{-2K+k'+1} + 2^{-k+3} \nonumber \\
    &\leq 2^{-2K+k+1} + 2^{-k+3} \nonumber \\
    &\leq 2^{-k+4}. \label{eq:a_k}
\end{align}

The expected number of samples is equal to the samples taken in each step times the probability that the algorithm reaches that step.
For simplicity, define $A_0 = 1$ deterministically.
Ignoring the $8\ns_{1/8}\ln(1/\delta)$ samples taken at the beginning of the algorithm, the expected number of samples is upper bounded by
\begin{align*}
    \sum_{k=1}^K \Prma[X, r]{A_{k-1}} J_k L_k \ns_{t_k}
    &\leq \sum_{k=1}^K 2^{-k+4} \ceil{t_k^2 2^{2k}} \ceil{16(K-k+1)} \ns_{t_k}\\
    &= O\p{\sum_{k=1}^{\ceil{\lg(1/\rho)}} 2^k (K-k+1) t_k^2 \ns_{t_k}}.
\end{align*}

The final bound follows by including the $8\ns_{1/8}\ln(1/\delta)$ samples taken deterministically at the beginning of the algorithm.

\paragraph{Replicability} As in the proof of \cref{thm:general-estimator}, we will define a notion of replicable correctness (which depends on the randomness $r$) which is defined even if the underlying distribution satisfies neither the null nor alternate hypothesis.
As long as the algorithm is correct in this notion with probability $1-\rho$, it will overall be $2\rho$-replicable.
As in the prior proof, we will say that the algorithm should output \accept if $z \leq \tau_0 + r\Delta$ and \reject otherwise.

The algorithm only outputs the wrong answer in the first step if the median estimate is smaller than $\tau_0 + \Delta/8$ or larger than $\tau_1 + \Delta/8$.
To show correctness (in the standard, non-replicable) sense, we showed that the median deviates from $z$ by more than $\Delta/8$ only with probability $\delta \leq \rho$.
As $r \in \bc{\frac{1}{4}, \frac{3}{4}}$, with probability $1-\rho$, if the algorithm terminates in the first step, it is replicably correct as it must be the case that $z < \tau_0 + \Delta/4$ or $z > \tau_0 + 3\Delta/4$.

Now, consider the case that the algorithm does not terminate in the first step.
As in the proof of expected sample complexity, we will break down the failure probability of the algorithm depending on the event $A_k$ that the algorithm does has not terminated by the end of step $k$.
Let $B_k$ be the event that the algorithm outputs the replicably incorrect answer at step $k$.
At step $k < K$, the algorithm only terminates if $\hat{V}_k$ is $2^{-k}\Delta$ far from the random threshold $\tau_0 + r\Delta$.
As the definition of replicable correctness is based on the location of $z$ relative to the threshold, the algorithm only terminates incorrectly if $\hat{V}_k$ is $2^{-k}\Delta$ far from $z$.
By \cref{eq:median-of-means-hoeffding}, this occurs with probability at most $\exp(-L_k/8)$.
Note that this deviation event is independent of $A_{k-1}$, the event of the algorithm having not terminated up to this point.
Therefore,
\begin{align*}
    \Prma[X, r]{B_k}
    &\leq A_{k-1} \exp(-L_k/8) \\
    &\leq 2^{-k+5} e^{-2(K-k+1)} \tag{by \cref{eq:a_k}} \\
    &\leq 2^{-2K+k+4}.
\end{align*}
By union bound, the failure probability across of the steps is upper bounded by
\begin{equation*}
    \sum_{k=1}^K \Prma[X, r]{B_k}
    \leq 2^{-2K+4} \sum_{k=1}^k 2^k
    \leq 2^{-K+5}
    = O(\rho).
\end{equation*}

Overall, the algorithm is $O(\rho)$-replicable, as required.
\end{proof}

\begin{remark}[Importance of Size-Invariance]
The size-invariant property is key to the replicability of our algorithm.
Replicable correctness is defined by comparing $\Ema{Z(\ns)}$ to the random threshold defined by $r$.
If the location of $\Ema{Z(\ns)}$ changes with the number of samples across different levels of the algorithm, then two independent runs of the algorithm may terminate on different levels with a different notion of replicable correctness.
This would violate replicability.
\end{remark}

\subsection{Upper Bound Applications}
\subsubsection{Coin Testing}\label{sec:coin-testing-upper-bound}

We apply our expectation-gap framework to the fundamental replicable coin testing problem which is ubiquitous as a subroutine in replicable algorithms.
Our algorithm improves the sample complexity both in expectation and in the worst-case over existing bounds and \emph{matches lower bounds in all parameter regimes up to constant factors}.
The proof of this result is a direct application of \cref{thm:size-invariant-estimator}.

\begin{definition}[Replicable (Biased) Coin Testing (as defined in \cite{hopkins2024replicability})]\label{def:coin_testing}
Consider $0 \leq p_0 < q_0 \leq 1$ and let $\eps = q_0 - p_0$. For a $p \in [0,1]$, an algorithm receives samples from $\Ber(p)$. The null hypothesis is that $p = p_0$ and the alterate hypothesis is that $p \geq q_0$.
\end{definition}

\paragraph{Prior work: }
Recall that the prior state-of-the-art sample complexity was given in \cite{hopkins2024replicability}.
\begin{theorem}[Theorem 3.5 of \cite{hopkins2024replicability}]\label{thm:prior_coin}
For any $0 \leq \delta \leq \rho \leq 1$, there exists an algorithm for replicable coin testing which is $\rho$-replicable, fails with probability at most $\delta$, and uses samples $O\p{\frac{q_0 \log(1/\delta)}{\eps^2 \rho}}$ in expectation and $O\p{\frac{q_0 \log(1/\delta)}{\eps^2 \rho^2}}$ in the worst-case.
\end{theorem}

Lower bounds (with and without replicability) appear in \cite{hopkins2024replicability} and \cite{lee2021uncertaintycoins}.
\begin{theorem}[Theorem 3.7 of \cite{hopkins2024replicability}]\label{thm:coin-lb-1}
Consider any $p_0 < q_0 < \frac{1}{2}$ and $\delta \leq \rho \leq 1/16$. Any replicable coin testing algorithm with these parameters must use $\Omega\p{\frac{q_0}{\eps^2\rho}}$ samples in expectation and $\Omega\p{\frac{q_0}{\eps^2\rho^2}}$ in the worst-case.
\end{theorem}

\begin{theorem}[Direct Corollary of Theorem 1.3 of \cite{lee2021uncertaintycoins}]\label{thm:coin-lb-2}
Consider any $p_0 \in [0, 1/2)$ and $q_0 > p_0$ where $q_0 - p_0 \leq 1 - 2p_0$. Any (non-replicable) algorithm which solves the coin testing problem must use $\Omega\p{\frac{p_0 \log(1/\delta)}{\eps^2}}$ samples in expectation.
\end{theorem}

\paragraph{Our result:}
Our result improves upon the prior best upper bound and matches the lower bounds up to constant factors for the sample complexity in the worst-case and in expectation.
\footnote{Note that the non-replicable lower bound has $p_0$ in the numerator rather than $q_0 = p_0 + \eps$ which introduces an extra $\frac{\log(1/\delta)}{\eps}$ additive term in our upper bound. This is necessary even if $p_0 = 0$ as with $\ns = o\p{\frac{\log(1/\delta)}{\eps}}$ from $\Ber(\eps)$, no heads appear in the sample with probability greater than $\delta$ and thus it is impossible to distinguish from sampling from $\Ber(0)$.} 
\begin{theorem}\label{thm:optimal-coin-testing}
For any $0 \leq \delta \leq \rho \leq 1$, \cref{alg:size-invariant-estimator} solves replicable coin testing.
It is $\rho$-replicable, fails with probability at most $\min\{\delta, \exp(-1/\rho)\}$, and uses samples $O\p{\frac{q_0}{\eps^2 \rho} + \frac{q_0 \log(1/\delta)}{\eps^2}}$ in expectation and $O\p{\frac{q_0}{\eps^2 \rho^2} + \frac{q_0 \log(1/\delta)}{\eps^2}}$ in the worst-case.
\end{theorem}

\begin{proof}
Consider the size-invariant expectation-gap statistic statistic $Z(\ns)$ which is the fraction of samples which are heads.
As $\Ema[]{Z(\ns)} = p$, $\tau_0 = p_0$ and $\tau_1 = q_0$ are valid null and alternate thresholds, respectively.
Therefore, $\Delta = \eps \leq 1$.
The variance of the estimator is $\Varma[]{Z(\ns)} = p(1-p)/\ns \leq p/\ns$. Note that if $p > q_0$,
\begin{equation*}
    \sqrt{\Varma[]{Z(\ns)}}
    \leq \sqrt{\frac{p}{\ns}}
    = \sqrt{\frac{q_0}{\ns}\p{1 + \frac{p - q_0}{q_0}}}
    \leq \sqrt{\frac{q_0}{\ns}\p{1 + \frac{p - q_0}{\eps}}}
    \leq \sqrt{\frac{q_0}{\ns}}\p{1 + \frac{p - q_0}{\eps}}.
\end{equation*}
Therefore, $\sigma(\ns) = \sqrt{q_0/\ns}$ is a valid variance upper bound.
The resulting noise-to-signal ratio is $\frac{\sqrt{q_0}}{\eps\sqrt{\ns}}$, and the constant sampling breakpoint is achieved at $\ns_{1/8}=O(q_0/\eps^2)$.

Applying \cref{cor:scale-invariant-simple} with this statistic yields the result.
\end{proof}

\subsubsection{Uniformity Testing}\label{sec:unif-testing-upper-bound}

\paragraph{Prior work:}
The work of \cite{liu2024replicableuniformity} introduces the problem of replicable uniformity testing (see \cref{def:replicable_unif}) and gave the following upper bound.

\begin{theorem}[Theorem 1.3 of \cite{liu2024replicableuniformity}]
Consider $n \in \N$, $0 \leq \rho \leq 1$ and $\eps > 0$.
There exists an algorithm which solves $(n, \eps, \rho, \rho)$-replicable uniformity testing and takes
$$O\p{\frac{\sqrt{n}\log(1/\rho)\sqrt{\log(n/\rho)}}{\eps^2\rho} + \frac{\log(1/\rho)}{\eps^2\rho^2}}$$
samples in the worst-case.
\end{theorem}

\paragraph{Our result:}
We give the first in-expectation and with high probability sampling bounds for replicable uniformity testing as a direct application of our framework.
These bounds significantly improve the $\rho$ dependence from prior work and show that \textbf{replicability can be achieved for free} if $\rho \gg \frac{\eps}{\log(1/\delta)}$ and $n$ is sufficiently large.
The worst-case sample complexity of our algorithm improves upon log factors over the prior work in some regimes though includes a $\frac{\sqrt{n}}{\eps \rho^2}$ term which does not appear in the preceeding theorem.
\begin{theorem}\label{thm:uniformity-testing}
Consider $n \in \N$, $0 \leq \delta \leq \rho \leq 1$ and $\eps > 0$.
\cref{alg:size-invariant-estimator} solves $(n, \eps, \rho, \delta)$-replicable uniformity testing, taking
$$
O\p{\frac{\sqrt{n}\log(1/\delta)}{\eps^2} + \frac{\sqrt{n}}{\eps \rho} + \frac{1}{\eps^2 \rho}}
$$
samples in expectation and
$$
O\p{\frac{\sqrt{n}\log(1/\delta)}{\eps^2}
+ \frac{\sqrt{n}}{\eps^2 \rho} +\frac{\sqrt{n}}{\eps \rho^2} + \frac{1}{\eps^2 \rho^2}}
$$
samples in the worst-case.
\end{theorem}

Key to our result will be the well-known collision tester for uniformity testing where $Z(\ns)$ is the number of collisions among the sampled elements divided by $\binom{\ns}{2}$. We will make use of the tight analysis of this statistic given by \cite{diakonikolas2019collisiontesters}.

\begin{lemma}[From the proof of Lemma 7 in \cite{diakonikolas2019collisiontesters}]
Consider a distribution $p$ over $[n]$ where $\|p\|_2^2 = (1+\alpha)/n$ with $\alpha > 0$.
Then, there exists a universal constant $C$ such that:
\begin{equation*}\label{lem:uniformity-variance}
    \Varma[]{Z(\ns)}
    \leq C\p{\frac{1}{\ns^2} \p{\frac{1+\alpha}{n}} + \frac{1}{\ns} \p{\frac{\alpha}{n^2} + \frac{\alpha^{3/2}}{n^{3/2}}}}.
\end{equation*}
\end{lemma}

\begin{proof}[Proof of \cref{thm:uniformity-testing}]
Consider the size-invariant expectation-gap statistic $Z(\ns)$ which is the number of collisions among the sampled elements divided by $\binom{\ns}{2}$.
Standard analysis (e.g. see \cite{canonne2020survey}) shows $z = \Ema[]{Z(\ns)} = \|p\|_2^2$.
Furthermore, if $p=\Unif([n])$, $\|p\|_2^2 = 1/n$, and if $\|p, \Unif[n]\|_1 \geq \eps$, $\|p\|_2^2 \geq (1 + \eps^2)/n$.
Therefore, setting $\tau_0 = 1/n$ and $\tau_1 = (1+ \eps^2)/n$ are valid choices for the null and alternate thresholds, respectively.
Then, $\Delta = \eps^2/n$.

Let $\alpha$ be such that $z = (1 + \alpha)/n$.
To choose a valid variance upper bound $\sigma(\ns)$, it must satisfy the condition that
\begin{equation*}
    \sqrt{\Varma[]{Z(\ns)}}
    \leq \sigma(\ns)\p{1 + \max\bc{0,\frac{z - \tau_1}{\Delta}}}
    = \sigma(\ns)\p{1 + \max\bc{0, \frac{(\alpha - \eps^2)/n}{\eps^2/n}}}
    = \sigma(\ns)\max\bc{1, \frac{\alpha}{\eps^2}}.
\end{equation*}

Using the variance analysis from \cref{lem:uniformity-variance} which is monotonically increasing in $\alpha$, it suffices to choose $\sigma(\ns)$ such that, for any $\alpha \geq \eps^2$,
\begin{align*}
    &\frac{\sigma(\ns) \alpha}{\eps^2}
    \geq \sqrt{C\p{\frac{1}{\ns^2} \p{\frac{1+\alpha}{n}} + \frac{1}{\ns} \p{\frac{\alpha}{n^2} + \frac{\alpha^{3/2}}{n^{3/2}}}}} \\
    \Longleftarrow \qquad & \sigma(\ns)
    \geq \frac{C_1 \eps^2}{\alpha m \sqrt{n}} + \frac{C_2}{\sqrt{\ns}}\p{\frac{\eps^2}{\alpha^{1/2} n} + \frac{\eps^2}{\alpha^{1/4} n^{3/4}}} \\
    \Longleftarrow \qquad & \sigma(\ns)
    \geq \frac{C_1}{\ns \sqrt{n}} + \frac{C_2}{\sqrt{\ns}}\p{\frac{\eps}{n} + \frac{\eps^{3/2}}{n^{3/4}}} \tag{as $\alpha \geq \eps^2$}. 
\end{align*}
We will choose this final expression as our variance upper bound $\sigma(\ns)$.

The resulting noise-to-signal ratio is
\begin{equation*}
    f(\ns) = \frac{\sigma(\ns)}{\Delta} 
    = \frac{C_1 \sqrt{n}}{\ns \eps^2} + \frac{C_2}{\sqrt{\ns}}\p{\frac{1}{\eps} + \frac{n^{1/4}}{\eps^{1/2}}}. 
\end{equation*}
To bound the sampling breakpoints $\ns_t$, we will proceed by cases on each of the two terms of $f(\ns)$.
For the first term, it must be the case that
\begin{equation*}
    \frac{C_1\sqrt{n}}{\ns_t\eps^2} \leq t/2
    \iff
    \ns_t \geq \frac{2C_1\sqrt{n}}{t\eps^2}.
\end{equation*}
For the second term, it must be that
\begin{equation*}
    \frac{C_2}{\sqrt{\ns}}\p{\frac{1}{\eps} + \frac{n^{1/4}}{\eps^{1/2}}} \leq t/2
    \iff \ns_t \geq \frac{4C_2^2}{t^2}\p{\frac{1}{\eps^2} + \frac{\sqrt{n}}{\eps}}.
\end{equation*}
Overall, the sampling breakpoint is
\begin{equation*}
   \ns_t = O\p{\frac{\sqrt{n}}{\eps^2t} + \frac{\sqrt{n}}{\eps t^2} + \frac{1}{\eps^2 t^2}}. 
\end{equation*}

Now that we have described the size-invariant expectation-gap statistic, it remains to choose breakpoints for our algorithm.
Consider running \cref{alg:size-invariant-estimator} with $t_k = 2^{-k}$.
The algorithm is $\rho$-replicable and is correct with probability $1-\delta$.
The expected number of samples taken are
\begin{align*}
    O&\p{\ns_{1/8} \log(1/\delta)
    + \sum_{k=1}^{\ceil{\lg(1/\rho)}} 2^{k} (K-k+1) t_k^2 \ns_{t_k}}
    = O\p{\frac{\sqrt{n}\log(1/\delta)}{\eps^2}
    + \sum_{k=1}^{\ceil{\lg(1/\rho)}} 2^{-k} (K-k+1) \ns_{t_k}} \\
    &= O\p{\frac{\sqrt{n}\log(1/\delta)}{\eps^2}
    + \sum_{k=1}^{\ceil{\lg(1/\rho)}} 2^{-k} (K-k+1) \p{\frac{\sqrt{n}2^k}{\eps^2} +\frac{\sqrt{n}2^{2k}}{\eps} + \frac{2^{2k}}{\eps^2}}} \\
    &= O\p{\frac{\sqrt{n}\log(1/\delta)}{\eps^2}
    + \sum_{k=1}^{\ceil{\lg(1/\rho)}} (K-k+1) \p{\frac{\sqrt{n}}{\eps^2} +\frac{\sqrt{n}2^{k}}{\eps} + \frac{2^{k}}{\eps^2}}} \\
    &= O\p{\frac{\sqrt{n}\log(1/\delta)}{\eps^2}
    + \frac{\sqrt{n} \log(1/\rho)}{\eps^2} +\frac{\sqrt{n}}{\eps \rho} + \frac{1}{\eps^2 \rho}}.
\end{align*}
Note that the second term in the summation is dominated by the first term as $\delta <= \rho$.
The worst-case number of samples is
\begin{align*}
    O&\p{\ns_{1/8} \log(1/\delta)
    + \sum_{k=1}^{\ceil{\lg(1/\rho)}} 2^{2k} (K-k+1) t_k^2 \ns_{t_k}}
    = O\p{\frac{\sqrt{n}\log(1/\delta)}{\eps^2}
    + \sum_{k=1}^{\ceil{\lg(1/\rho)}} (K-k+1) \ns_{t_k}} \\
    &= O\p{\frac{\sqrt{n}\log(1/\delta)}{\eps^2}
    + \sum_{k=1}^{\ceil{\lg(1/\rho)}} (K-k+1) \p{\frac{\sqrt{n}2^k}{\eps^2} +\frac{\sqrt{n}2^{2k}}{\eps} + \frac{2^{2k}}{\eps^2}}} \\
    &= O\p{\frac{\sqrt{n}\log(1/\delta)}{\eps^2}
    + \frac{\sqrt{n}}{\eps^2 \rho} +\frac{\sqrt{n}}{\eps \rho^2} + \frac{1}{\eps^2 \rho^2}}.
\end{align*}
\end{proof}

\subsubsection{Closeness Testing}\label{sec:closeness-testing-upper-bound}

We use the general estimator of \cref{thm:general-estimator} to get the first non-trivial bounds for replicable closeness testing. These bounds match our lower bound in \cref{thm:LB_close} up to constant or logarithmic factors depending on the parameter regime.

The statistic we will utilize is the $\chi^2$ style statistic used in prior work on optimal non-replicable closeness testing~\cite{chan2014optimal}.
Unfortunately, this statistic (and any normalization of it) is not size-invariant, so we will utilize the general expectation-gap estimator in \cref{alg:general-estimator} without in-expectation sampling bounds.

\begin{lemma}[\cite{chan2014optimal}]\label{lem:chan14} Given a set $T$ of Pois$(\ns)$ samples from the product distribution $p \times q$ over $[n]^2$, let $X_i, Y_i$ denote the number of occurrences of the $i$th domain elements in the samples from $p$ and $q$, respectively. Define 
\begin{equation}\label{eq:closeness_stat}
    Z = \sum_{i = 1}^n \frac{ (X_i - Y_i)^2 - X_i - Y_i }{X_i + Y_i}.
\end{equation}
We have
\begin{enumerate}
    \item $\Ema[]{Z} = \ns \sum_i \frac{(p_i-q_i)^2}{p_i + q_i} \left( 1 - \frac{1-e^{-\ns(p_i+q_i)}}{\ns(p_i + q_i)}\right)$,
    \item $\Ema[]{Z} \ge \frac{\ns^2}{4n + 2\ns} \|p-q\|_1^2$,
    \item If $p = q$ then $\Ema[]{Z} = 0$,
    \item $\Varma[]{Z} \le 2 \min(n, \ns) + 5m \sum_i  \frac{(p_i-q_i)^2}{p_i + q_i} \le 10(\min(n,\ns) + \ns)$,
    \item If $\ns \ge n,$ $\Varma[]{Z} \le 10(n + \Ema[]{Z})$.
\end{enumerate}
\end{lemma}

\begin{remark}\label{rem:poisson}
    To simplify the computations involved, much of the the analysis of distribution testing algorithms in the literature applies the standard ``Poissonlization" trick \cite{canonne2020survey, Canonne_book_22}. In particular, this means we draw a random number of samples from a Poisson distribution rather than a fixed number. This simplifies the calculation as the number of occurrences of each element become mutually independent. Furthermore, it is without loss of generality using the fact that Poisson distributions are highly concentrated (in all cases, the failure probability of not receiving a sample that is within a constant factor of $\Poi(\ns)$ for our choices of $\ns$ can be made to be an arbitrary large polynomial in $\rho$). The same trick was also applied in the replicable uniformity testing paper of \cite{liu2024replicableuniformity}.
\end{remark}

\begin{theorem}\label{thm:closeness-upperbound}
Consider $n \in \N$, $0 \leq \rho \leq 1$ and $\eps > 0$.
Let $C$ be a constant with $\delta = \rho^C$.
\cref{alg:general-estimator} solves $(n, \eps, \rho, \delta)$-replicable closeness testing and with worst-case sample complexity
\begin{equation*}
    O\p{\frac{n^{2/3}}{\eps^{4/3}\rho^{2/3}} + \frac{\sqrt{n}}{\eps^2 \rho} + \frac{1}{\eps^2 \rho^2}}.
\end{equation*}
\end{theorem}

\begin{proof}
Consider the statistic $Z(\ns)$ defined in \cref{eq:closeness_stat}.
Using \cref{lem:chan14}, we will choose the rest of the parameters of the expectation-gap statistic.
Let $\tau_0(\ns) = 0$ and $\tau_1(\ns) = \frac{\ns^2 \eps^2}{4n + 2\ns}$ be the null and alternate thresholds.
Note that $\Delta(\ns) = \tau_1(\ns)$.
It remains to choose a variance bound $\sigma(\ns)$.
We will split into two cases depending on whether $\ns < n$ (note that this is a property of the input parameters to the algorithm).

\paragraph{Case 1:} $\ns < n$.
In this case, $\Varma[]{Z(\ns)} \leq 20 \ns$ by \cref{lem:chan14}, so it suffices to choose $\sigma(\ns) = \sqrt{20\ns}$.
The sampling breakpoints $\ns_t$ must satisfy $f(\ns_t) \leq t/2$.
Expanding this condition:
\begin{align*}
    f(\ns_t) \leq t/2
    &\iff \frac{\sqrt{20\ns_t}}{\frac{\ns_t^2 \eps^2}{4n + 2\ns_t}} \leq t/2 \\
    &\iff \frac{\sqrt{20}\p{4n + 2\ns_t}}{\ns_t^{3/2} \eps^2} \leq t/2 \\
    &\Longleftarrow \frac{30n}{\ns_t^{3/2} \eps^2} \leq t/2 \tag{as $\ns < n$}\\
    &\iff \ns_t^{3/2} \geq \frac{60n}{\eps^2t} \\
    &\iff \ns_t \geq \frac{16 n^{2/3}}{\eps^{4/3}t^{2/3}}.
\end{align*}

\paragraph{Case 2: $\ns \geq n$}
Let $\|p-q\|_1 = \alpha$.
In this case, $\Varma[]{Z(\ns)} \leq 10\p{n + \Ema[]{Z}}$ via \cref{lem:chan14}.
We must choose $\sigma(\ns)$ such that
\begin{align*}
    &\sqrt{\Varma[]{Z(\ns)}} \leq \sigma(\ns) \p{1 + \min\bc{0, \frac{\Ema[]{Z(\ns)} - \tau_1(\ns)}{\Delta(\ns)}}} \\
    \iff \qquad & \sqrt{\Varma[]{Z(\ns)}} \leq \sigma(\ns) \min\bc{1, \frac{\Ema[]{Z(\ns)}}{\Delta(\ns)}} \\
    \Longleftarrow \qquad & \sigma(\ns) \geq \frac{\sqrt{10\p{n + \Ema[]{Z(\ns)}}}}{\min\bc{1, \frac{\Ema[]{Z(\ns)}}{\Delta(\ns)}}}
\end{align*}
Note that the right hand side is maximized when $\Ema[]{Z(\ns)} \geq \Delta(\ns)$.
Recall from \cref{lem:chan14} that $\Ema[]{Z(\ns)} \geq \frac{\ns^2 \alpha^2}{4n + 2\ns}$.
Therefore, it suffices to choose $\sigma(\ns)$ with
\begin{align*}
    & \sigma(\ns)^2 \geq \frac{10\p{n + \Ema[]{Z(\ns)}}}{\Ema[]{Z(\ns)}^2/\Delta(\ns)^2}
    = \frac{10n \Delta(\ns)^2}{\Ema[]{Z(\ns)}^2} + \frac{10\Delta(\ns)^2}{\Ema[]{Z(\ns)}} \\
    \Longleftarrow \qquad & \sigma(\ns)^2 \geq 10n + 10\Delta(\ns)
    = 10 n + \frac{10\ns^2 \eps^2}{4n + 2\ns} \\
    \Longleftarrow \qquad & \sigma(\ns)^2 \geq 10n + 5\ns \eps^2.
\end{align*}
Therefore, a valid choice is $\sigma(\ns) = \sqrt{10n + 5\ns\eps^2}$.
We will proceed by cases depending on which of these two terms dominates when bounding the sampling breakpoint $\ns_t$.

Assume first that $\sigma(\ns_t) \leq \sqrt{20n}$.
\begin{align*}
    f(\ns_t) \leq t/2
    &\iff \frac{\sqrt{20n}}{\frac{\ns_t^2 \eps^2}{4n + 2\ns_t}} \leq t/2 \\
    &\iff \frac{\sqrt{20n}(4n + 2\ns_t)}{\ns_t^2 \eps^2} \leq t/2 \\
    &\Longleftarrow \frac{27 \sqrt{n}}{\ns_t \eps^2} \leq t/2 \tag{as $\ns \geq n$} \\
    &\iff \ns_t \geq \frac{54 \sqrt{n}}{\eps^2 t}.
\end{align*}

Now, assume that $\sigma(\ns_t) \leq \sqrt{10\ns_t \eps^2}$.
\begin{align*}
    f(\ns_t) \leq t/2
    &\iff \frac{\sqrt{10\ns_t \eps^2}}{\frac{\ns_t^2 \eps^2}{4n + 2\ns_t}} \leq t/2 \\
    &\iff \frac{\sqrt{10}(4n + 2\ns_t)}{\ns_t^{3/2} \eps} \leq t/2 \\
    &\Longleftarrow \frac{19}{\sqrt{\ns_t} \eps} \leq t/2 \tag{as $\ns \geq n$} \\
    &\iff \sqrt{\ns_t} \geq \frac{38}{\eps t} \\
    &\iff \ns_t \geq \frac{1444}{\eps^2 t^2}.
\end{align*}

Completing the case analysis, across all parameter settings, the breakpoint $\ns_t$ will be bounded by:
\begin{equation*}
    \ns_t = O\p{\frac{n^{2/3}}{\eps^{4/3}t^{2/3}} + \frac{\sqrt{n}}{\eps^2 t} + \frac{1}{\eps^2 t^2}}.
\end{equation*}
Applying \cref{thm:general-estimator} with $t = \rho$ completes the proof.
\end{proof}

\begin{remark}
We are not aware of a size-invariant statistic for closeness testing which gets optimal bounds in the non-replicable setting.
Therefore, we do not get improved sampling bounds in expectation for this problem: this is an interesting open question.
\end{remark}

\section{Gaussian Mean Testing}\label{sec:gaussian}

In this section, we extend our study to continuous distributions by proving upper and lower bounds for Gaussian mean testing. In both cases, we make use of our frameworks for lower and upper bounds developed in the prior sections, but more work is needed to optimize these tools for the Gaussian setting.

We recall our upper bound for replicable Gaussian mean testing.

\repgaussian*

\paragraph{Our Algorithm.}
We first present some intuition. At a high level, our algorithm works as follows. Since the distribution $\mathcal{D}$ can be arbitrary in $\R^n$, we must be careful in filtering out pathological distributions. For example, imagine a distribution that samples from a standard Gaussian with probability $1-\rho$, but with $\rho$-probability picks points far away in a manner such that simple estimators such as the sample mean are tricked into thinking the mean is very large. This is evidently not replicable, since our answer really hinges on an event with $\rho$-probability. 

Thus, we first filter out ``bad'' distributions that do not behave like a standard identity-covariance Gaussian distribution, such as distributions which large ``bias" in a certain direction. We capture this by filtering out distributions that have high probability of pairs of samples having large inner product (much larger than we expect for ``well-behaved" Gaussians). We then use a variation of the standard Gaussian identity test, which computes the norm of the sum of data points, and accepts if the norm is below some threshold. For replicability, we modify this test to accept with some probability that depends on the norm, in a manner similar to the canonical replicable tester in \Cref{sec:canonical-tester}.

Fix a parameter $L \ge 1$, that will be decided later. We start by making the following definition, which will be required to describe our algorithm in more detail. It helps us deal with the case of distributions that can sample points arbitrarily far away.  

\begin{definition}
    Given a distribution $\cD$ over $\R^d$, we define $\cD_{\proj}$ represent the projected distribution of $\cD$ onto $\cB$, the ball of radius $L \cdot \sqrt{d}$ around the origin. Formally, we sample $X_i \sim \cD$ and output $Y_i = \frac{X_i}{\max\left(1, \|X_i\|/(L \sqrt{d})\right)}$, meaning a point $X_i$ in $\cB$ is left as is and a point $X_i$ outside $\cB$ is projected to lie on the boundary of $\cB$.
\end{definition}

Based on this definition, we can assume that the distribution $\cD$ is contained in the ball of radius $L \sqrt{d}$ with probability $1$, by replacing $\cD$ with $\cD_{\proj}$ if necessary. Formally, given $n$ samples $X_1, \dots, X_n$, we perform the algorithm on $\{Y_i\}$ where $Y_i = \frac{X_i}{\max\left(\|X_i\|/(L \sqrt{d})\right)}$. Note that this preserves independence of the samples, so a replicable algorithm on $\{Y_i\}$ is still replicable on $\{X_i\}$. However, note that our desired goal is now slightly different. The null hypothesis is now $\cN(0, I)_{\proj}$, i.e., the projection of the standard Gaussian onto the ball of radius $L \sqrt{d}$, and the alternative hypothesis is $\cN(\mu, I)_{\proj}$ for any $\|\mu\| \ge \alpha$. Even with this simplification, we still need to deal with the case that the distribution we sample from can be biased along certain directions. 

Our algorithm as follows. For the sake of clarity, we break down our algorithm into three key primitives (denoted as Steps A, B, C below) and abstract away the contents of these steps to their own subsections (\cref{sec:gaussian_A}, \cref{sec:gaussian_B}, \cref{sec:gaussian_C} respectively). 

\begin{enumeratebox}[alg:gaussian-mean]{Gaussian Mean Tester}
\begin{enumerate}
    \item Let $\ns = \tilde{O}\left(\frac{\sqrt{d}}{\alpha^2 \rho} + \frac{\sqrt{d}}{\alpha \rho^2} + \frac{1}{\alpha^2 \rho^2}\right)$ be sufficiently large, and set thresholds $T_1, T_2$ and $S$ according to \cref{sec:gaussian_C}. We sample up to $5 \ns$ data points.
    \item \textbf{Step A:} Sample $\ns$ data points from $\cD$, and run the replicable algorithm from \cref{sec:gaussian_A} that rejects with high probability if $\|\Ema[X_i \sim \mathcal{D}]{X_iX_i^\top}\|_{op} \ge 5T_1$.
    \item \textbf{Step B:} Sample $2 \cdot \ns$ fresh data points. Given $2 \cdot \ns$ samples $X_1, \dots, X_{\ns}, Y_1, \dots, Y_{\ns} \sim \cD$, consider creating a bipartite graph between $X_1, \dots, X_{\ns}$ and $Y_1, \dots, Y_{\ns}$ that connects $X_i$ to $Y_j$ if and only if $|\langle X_i, Y_j \rangle| \ge S$. Let $f(\cD)$ be the maximum matching size in this bipartite graph.
    Then, run the replicable algorithm from \cref{sec:gaussian_B} that rejects with high probability if $f(\cD)$ exceeds $5T_2$.
    \item \textbf{Step C:} Assuming we have not rejected in Steps $A$ and $B$, draw $2 \ns$ fresh data points $X_1, \dots, X_{\ns}, Y_1, \dots, Y_{\ns}$\footnote{In a slight abuse of notation, we again call these points $X_1, \dots, X_{\ns}$.}, and compute the value $\langle \sum_{i=1}^{\ns} X_i,  \sum_{i=1}^{\ns} Y_i \rangle$. Use the replicable Expectation-Gap  \cref{alg:general-estimator} (detailed in \cref{sec:gaussian_C}) to reject if this value is too large, and accept if this value is small enough.
\end{enumerate}
\end{enumeratebox}

\subsection{Threshold Algorithm}
We note the following simple thresholding algorithm which will be used for Steps A and B. Its guarantees are very similar to that of our Expectation-Gap Estimator framework of \cref{def:expectation-gap} and \cref{thm:general-estimator} in \cref{sec:expectation_gap_tester}. However, we find it easier to work with the slightly modified guarantees that deal with high probability events rather than directly dealing with variances.

The threshold algorithm is as follows: Given a threshold parameter $T$, and a dataset $X = \{X_1, \dots, X_n\}$, suppose that $h: \XX^n \to \R_{\ge 0}$ is a positive-valued statistic (deterministic in $X$). We consider the following algorithm, that we call $\A_{h, T}$: compute $\gamma = \frac{3T - h(X)}{T}$, sample $r \sim \Unif([0, 1])$, and accept if and only if $\gamma \le r$. Note that if $h$ can be efficiently computed, then the algorithm $\A_{h, T}$ can be as well.

We have the following analysis of this basic threshold algorithm. Since this is a slight modification of the proof of \cref{thm:general-estimator}, its proof is presented in \cref{sec:appendix_gaussian}.

\begin{proposition} \label{prop:basic-thresholding}
    Fix a function $h: \XX^n \to \R$ and parameters $T \ge 0$ and $0 < \delta \le \rho \le 1$. Suppose that for any distribution $\cD$ over $\XX$ and i.i.d. samples $X = \{X_1, \dots, X_{\ns}\} \sim \cD$, there exists a value $q = q(\cD)$ (which may implicitly depend on $h$ and $T$) such that with probability $1-\delta$ over the randomness of $X$, $|h(X) - q| \le \rho \cdot \max(|q|, T)$. Then, the following claims hold:
\begin{itemize}
    \item $\A_{h, T}$ is $12 \rho$-replicable.
    \item If $q(\cD) \ge 5T$, then with probability at least $1-\delta$, the algorithm rejects.
    \item If $q(\cD) \le T$, then with probability at least $1-\delta$, the algorithm accepts.
\end{itemize}
\end{proposition}

For Step C, we directly rely on \cref{thm:general-estimator}.

\subsection{Step A}\label{sec:gaussian_A}

First, we note the following basic consequence of the Matrix Chernoff bound~\cite{tropp2015introduction}.

\begin{lemma} \label{lem:matrix-chernoff-application-1}
    Fix any parameters $L \ge 1$, $\delta \le 1$, and let $\cD$ be a distribution over $\R^d$ such that each sample $X_i \sim \cD$ is bounded in $\ell_2$ norm by $L \sqrt{d}$ with probability $1$. Then, with probability at least $1-\delta$, the operator norm of the empirical covariance, $\left\|\frac{1}{\ns} \sum_{k=1}^{\ns} X_i X_i^\top\right\|_{op},$ is in the range $\left[\|\Sigma\|_{op} - H, \|\Sigma\|_{op} + H\right]$, where 
\[H = O\left(\max\left(\frac{d}{\ns} \cdot L^2 \cdot \log \frac{d}{\delta}, \sqrt{\frac{d}{\ns} \cdot \|\Sigma\|_{op} \cdot L^2 \cdot \log \frac{d}{\delta}}\right) \right).\]
\end{lemma}

\begin{proof}
    Note that $\|X_i X_i^\top\|_{op} \le L^2 \cdot d$ for all $X_i$. Therefore, by Matrix Chernoff, for any $\eps > 0$, 
\begin{equation} \label{eq:matrix-chernoff-application-v1}
    \Prma[]{\left\|\frac{1}{\ns} \sum_{k=1}^{\ns} X_i X_i^\top\right\|_{op} \ge (1+\eps) \cdot \|\Sigma\|_{op}} \le d \cdot e^{-O(\min(\eps, \eps^2) \cdot \|\Sigma\|_{op} \cdot \ns/(L^2 d))}.
\end{equation}
    Next, let $v$ be a unit vector such that $v^\top \Sigma v = \|\Sigma\|_{op}$. If we consider $v^\top \left(\frac{1}{\ns} \sum_{k=1}^{\ns} X_i X_i^\top\right) v = \frac{1}{\ns} \sum_{k=1}^{\ns} \langle v, X_i \rangle^2,$ note that each $\langle v, X_i \rangle^2$ is an independent random variable with mean $\|\Sigma\|_{op}$ and is bounded by $L^2 d$ since we assume $\|X_i\| \le L \sqrt{d}$ with probability $1$. So, by a standard Chernoff bound, 
\begin{equation} \label{eq:chernoff-application-v2}
    \Prma[]{\left\|\frac{1}{\ns} \sum_{k=1}^{\ns} X_i X_i^\top\right\|_{op} \le (1-\eps) \cdot \|\Sigma\|_{op}} \le \Prma[]{\frac{1}{\ns} \langle v, X_i \rangle^2 \le (1-\eps) \cdot \|\Sigma\|_{op}} \le e^{-O(\min(\eps, \eps^2) \cdot \|\Sigma\|_{op} \cdot \ns/(L^2 d))}.
\end{equation}

    Hence, if we set $\eps$ to be a sufficiently large multiple of $\max\left(\frac{d}{\ns \cdot \|\Sigma\|_{op}} \cdot L^2 \cdot \log \frac{d}{\delta}, \sqrt{\frac{d}{\ns \cdot \|\Sigma\|_{op}} \cdot L^2 \cdot \log \frac{d}{\delta}}\right),$ both \eqref{eq:matrix-chernoff-application-v1} and \eqref{eq:chernoff-application-v2} are at most $\delta/2$. By writing $H = \eps \cdot \|\Sigma\|_{op},$ the lemma is complete.
\end{proof}

\begin{lemma} \label{lem:matrix-chernoff-application-2}
    Let $\delta \le \rho \le 0.01$. There exists a $O(\rho)$-replicable algorithm $\A_1$ with the following properties. Fix a parameter $L \ge 1$, and let $T_1$ be any parameter such that $T_1 \ge O\left(\frac{d}{\ns \cdot \rho^2} \cdot L^2 \cdot \log \frac{d}{\delta}\right)$. Then, for any distribution $\cD$ over $\R^d$ contained in the ball of radius $L \sqrt{d}$ around the origin:
\begin{itemize}
    \item if $\|\Ema[X_i \sim \cD]{X_iX_i^\top}\|_{op} \le T_1,$ the algorithm, given $\ns$ samples from $\cD$, accepts with probability at least $1-\delta$.
    \item if $\|\Ema[X_i \sim \cD]{X_iX_i^\top}\|_{op} \ge 5T_1,$ the algorithm, given $\ns$ samples from $\cD$, rejects with probability at least $1-\delta$.
\end{itemize}
\end{lemma}

\begin{proof}
    Given data points $X_1, \dots, X_{\ns} \sim \cD$, let $t$ be the statistic $\|\frac{1}{\ns} \sum X_i X_i^\top\|_{op}$. We show that for any distribution $\cD$, $t = \|\frac{1}{\ns} \sum X_i X_i^\top\|_{op}$ lies in the interval $[q - \rho \cdot \max(q, T_1), q + \rho \cdot \max(q, T_1)]$, where $q = \|\Ema[]{X_i X_i^\top}\|_{op}$, with probability $1-\delta$ over the randomness of $X_i \sim \cD$.
    To see why, by Lemma~\ref{lem:matrix-chernoff-application-1}, we know that $t$ lies in the interval $[q-H, q+H]$ with $1-\delta$ probability, where $H \le O\left(\max\left(\frac{d}{\ns} \cdot L^2 \cdot \log \frac{d}{\rho}, \sqrt{\frac{d}{\ns} \cdot q \cdot L^2 \cdot \log \frac{d}{\rho}}\right) \right)$. So, it suffices to verify that $H \le \rho \cdot \max(q, T_1)$. If $q \le T_1$, then it suffices to verify that $\rho \cdot T_1 \ge O\left(\max\left(\frac{d}{\ns} \cdot L^2 \cdot \log \frac{d}{\rho}, \sqrt{\frac{d}{\ns} \cdot T_1 \cdot L^2 \cdot \log \frac{d}{\rho}}\right) \right)$. This holds as long as $T_1 \ge O\left(\frac{d}{\ns \cdot \rho^2} \cdot L^2 \cdot \log \frac{d}{\rho}\right)$. If $q \ge T_1$, when we need to verify that $\rho \cdot q \ge O\left(\max\left(\frac{d}{\ns} \cdot L^2 \cdot \log \frac{d}{\rho}, \sqrt{\frac{d}{\ns} \cdot q \cdot L^2 \cdot \log \frac{d}{\rho}}\right) \right)$. This holds as long as $q \ge O\left(\frac{d}{\ns \cdot \rho^2} \cdot L^2 \cdot \log \frac{d}{\rho}\right)$, which is true if $T_1 \ge O\left(\frac{d}{\ns \cdot \rho^2} \cdot L^2 \cdot \log \frac{d}{\rho}\right)$ since we assumed $q \ge T_1$.

    Therefore, by Proposition~\ref{prop:basic-thresholding}, the algorithm that computes $\A_{h, T_1}$ where $h(X) = \|\frac{1}{\ns}\sum X_i X_i^\top\|_{op}$ is $O(\rho)$-replicable. Moreover, since $q = \|\Ema[X_i \sim \cD]{X_iX_i^\top} \|_{op}$, the accuracy guarantees of the lemma hold as well.
\end{proof}

\subsection{Step B}\label{sec:gaussian_B}

In this section, we consider the following bipartite graph on data points.

\begin{definition} \label{def:matching-graph}
    For any data points $X_1, \dots, X_{\ns}, Y_1, \dots, Y_{\ns}$ and a threshold parameter $S \ge 0$, define the bipartite graph $G_S(X, Y)$ on $X = \{X_1, \dots, X_{\ns}\},$ $Y = \{Y_1, \dots, Y_{\ns}\}$ that connects $X_i, Y_j$ if $|\langle X_i, Y_j \rangle| \ge S$. Define $M_S(X, Y)$ to be the maximum matching size between $X, Y$ in this bipartite graph $G_S(X, Y)$.
\end{definition}

Our first step in this subsection is to prove the following lemma.

\begin{lemma} \label{lem:matching-concentration}
    Fix $X = \{X_1, \dots, X_{\ns}\}$ and $S$, and let $Y = \{Y_1, \dots, Y_{\ns}\}$ be drawn i.i.d. from any distribution $\cD$. Then, the random variable $M_S(X, Y)$, as a function of $Y$, satisfies the concentration inequality
\[\Prma[]{\left|M_S(X, Y) - \Ema[Y]{M_S(X, Y)}\right| \ge t }\le 2 \cdot \exp\left(-0.1 \cdot \min\left(t, \frac{t^2}{\Ema[Y]{M_S(X, Y)}}\right)\right).\]
\end{lemma}

To prove this lemma, we use \cref{thm:boucheron} from \cite{boucheron2000sharp} (see \cref{sec:prelims}).

\begin{proof}[Proof of \Cref{lem:matching-concentration}]
    We apply \Cref{thm:boucheron} as follows. Let $f(Y_1, \dots, Y_{\ns})$ be the maximum matching size $M_S(X, Y)$, and let $g(Y_1, \dots, Y_{i-1}, Y_{i+1}, \dots, Y_{\ns})$ be the maximum matching size between $X_1, \dots, X_{\ns}$ and $Y_1, \dots, Y_{i-1}, Y_{i+1}, \dots, Y_{\ns}$. Adding a data point to $Y$ will never decrease the matching size and will increase it by at most $1$. Moreover, if there is a maximum matching between $X$ and $Y$ of some size $k$, using $Y_{i_1}, \dots, Y_{i_k}$, then removing $x_i$ for $i \not\in \{i_1, \dots, i_k\}$ will maintain the maximum matching size at $k$. So, $f(x_1,\ldots,x_{\ns}) - g(x_1,\ldots,x_{i-1},x_{i+1},\ldots,x_{\ns})$ is positive for at most $k$ choices of $i$, and is at most $1$ in that setting. Thus, $f$ satisfies the required properties.

    Hence, for any fixed $X_1, \dots, X_{\ns}, S$, by \Cref{thm:boucheron}, 
\[\Prma[]{\left|M_S(X, Y) - \Ema[Y]{M_S(X, Y)}\right| \ge t }\le 2 \cdot \exp\left(-0.1 \cdot \min\left(t, \frac{t^2}{\Ema[Y]{M_S(X, Y)}}\right)\right).\]
    as desired.
\end{proof}

\begin{lemma} \label{lem:step-2}
    For any distribution $\cD$ and any fixed threshold $S \ge 0$, there exists a value $\mu_1 = \mu_1(\cD, S, \delta)$ such that for samples $X_1, \dots, X_{\ns}, Y_1, \dots, Y_{\ns} \sim \cD$, the matching size $M_S(X, Y)$ satisfies $\Prma[X, Y]{|M_S(X, Y)-\mu_1|} \le \rho \cdot \max\left(\mu_1, O\left(\frac{1}{\rho^2} \log \frac{1}{\delta}\right)\right)$ with probability at least $1-\delta$.
\end{lemma}

\begin{proof}
    Suppose we sample $X = \{X_1, \dots, X_{\ns}\},$ $X' = \{X_1', \dots, X_{\ns}'\},$ $Y = \{Y_1, \dots, Y_{\ns}\}$, and $Y' = \{Y_1', \dots, Y_{\ns}'\}$, all i.i.d. from $\cD$. Define $\mu_1 = \Ema[Y]{M_S(X, Y)}$, where $X$ is fixed. By Lemma~\ref{lem:matching-concentration}, with probability at least $1-\delta/2$, both $M_S(X, Y)$ and $M_S(X, Y')$ are within $O(\log \frac{1}{\delta} + \sqrt{\mu_1 \cdot \log \frac{1}{\delta}})$ of $\mu_1$. Next, define $\mu_2 = \Ema[X]{M_S(X, Y')}$. Again, applying Lemma~\ref{lem:matching-concentration}, with probability at least $1-\delta/2$, both $M_S(X, Y')$ and $M_S(X', Y')$ are within $O(\log \frac{1}{\delta} + \sqrt{\mu_2 \cdot \log \frac{1}{\delta}})$ of $\mu_2$. So, by Triangle inequality, with probability at least $1-\delta,$ over the randomness of $X, X', Y, Y'$ both $M_S(X, Y)$ and $M_S(X', Y')$ are within $O(\log \frac{1}{\delta} + \sqrt{\mu_1 \cdot \log \frac{1}{\delta}})$ of $\mu_1 := \Ema[Y]{M_S(X, Y)}$.

    Therefore, there exists a choice of $\mu_1$ such that with probability at least $1-\delta$ over $X', Y'$, $|M_S(X', Y') - \mu_1| \le O\left(\log \frac{1}{\delta} + \sqrt{\mu_1 \cdot \log \frac{1}{\delta}}\right) \le \rho \cdot \max\left(\mu_1, O\left(\frac{1}{\rho^2} \log \frac{1}{\delta}\right)\right)$.
\end{proof}

Hence, we can apply \Cref{prop:basic-thresholding} again, to obtain the following corollary.

\begin{corollary} \label{cor:step-2}
    Let $\delta \le \rho \le 0.01$, and let $S \ge 0$ by any threshold. For $\mu_1(\cD, S, \delta)$ as in \Cref{lem:step-2}, and for some threshold $T_2 = O\left(\frac{1}{\rho^2} \log \frac{1}{\delta}\right)$, there exists an $O(\rho)$-replicable algorithm $\A_2$ that accepts with probability $1-\delta$ whenever $\mu_1(\cD, S, \delta) \le T_2$ and rejects with probability $1-\delta$ whenever $\mu_1(\cD, S, \delta) \ge 5T_2$.
\end{corollary}

\subsection{Step C}\label{sec:gaussian_C}

\paragraph{Setting of parameters.} We now set parameters to properly initialize our \cref{alg:gaussian-mean}. Let $K$ be a sufficiently large polylogarithmic multiple of $\ns, d, \frac{1}{\alpha}, \frac{1}{\rho}$.
Fix parameters $L = K$, $S = K \cdot \sqrt{d}$, $T_1 = \left(1 + \frac{d}{\ns \cdot \rho^2}\right) \cdot L^2 \cdot K$, and $T_2 = \frac{1}{\rho^2} \cdot K$. Now, we define $\cD$ contained in the ball of radius $L \sqrt{d}$ to be a \emph{good} distribution, if $\|\Ema[X_i \sim \cD]{X_iX_i^\top}\|_{op} \le 5 T_1$, and $\mu_1(\cD, S, \frac{\rho}{L^4 d^2}) \le 5T_2$, where $\mu_1$ is defined in \Cref{lem:step-2}.

Now suppose we are given samples $X_1, \dots, X_{\ns}, Y_1, \dots, Y_{\ns} \sim \cD$, where $\cD$ is good, and we compute the statistic $\langle X_1 + \cdots + X_{\ns}, Y_1 + \cdots + Y_{\ns} \rangle$. In expectation, this statistic equals $\ns^2 \cdot \|\mu\|^2,$ where $\mu = \Ema[X_i \sim \cD]{X_i}$. The variance of this statistic is
\begin{align*}
    &O\left(\sum_{i, j} \Varma[]{\langle X_i, Y_j \rangle} + \ns^3 \cdot \Cov_{X_1, Y_1, Y'_1 \sim \cD}(\langle X_1, Y_1 \rangle, \langle X_1, Y'_1 \rangle)\right)\\
    &\le O\left(\sum_{i, j} \Ema[]{\langle X_i, Y_j \rangle^2} + \ns^3 \cdot \Ema[X_1, Y_1, Y'_1 \sim \cD]{\langle X_1, Y_1 \rangle \langle X_1, Y'_1 \rangle}\right) \\
    &= O\left(\Ema[]{ \sum_{i, j} \langle X_i, Y_j \rangle^2} + \ns^3 \cdot \Ema[X_1 \sim \cD]{\langle X_1, \mu \rangle^2}\right).
\end{align*}

To bound $\Ema[]{\langle X_1, \mu \rangle^2}$, note that we can write this as $\Ema[]{ \mu^\top X_1X_1^\top \mu} = \mu^\top \cdot \Ema[]{X_1X_1^\top} \cdot \mu \le \|\mu\|^2 \cdot \|\Ema[]{X_1X_1^\top}\|_{op} \le 5 T_1 \cdot \|\mu\|^2.$

To bound $\Ema[]{\sum_{i, j} \langle X_i, Y_j \rangle^2},$ we again consider the matching size $M_S(X, Y)$ from the previous subsection. If the matching has size $m$, there are subsets $A, B \subset [n]$ of size $m$, such that for all $i \not\in A, j \not\in B,$ $|\langle X_i, Y_j \rangle| \le S$.
We write
\[\sum_{i, j} \langle X_i, Y_j \rangle^2 = \sum_{i \not\in A, j \not\in B} \langle X_i, Y_j \rangle^2 + \sum_{i \in A} X_i^\top \cdot \left(\sum_j Y_j Y_j^\top\right) \cdot X_i + \sum_{j \in A} Y_j^\top \cdot \left(\sum_j X_i X_i^\top\right) \cdot Y_j - \sum_{i \in A, j \in B} \langle X_i, Y_j \rangle^2.\]
The first term is at most $\ns^2 \cdot S^2,$ since $|\langle X_i, Y_j \rangle| \le S$ if $i \not\in A, j \not\in B$. The second term is at most $m \cdot \|X_i\|^2 \cdot \|\sum_j Y_j Y_j^\top\|_{op}$. By \Cref{lem:matrix-chernoff-application-1}, with probability at least $1 - \frac{1}{L^4 d^2},$ $\|\sum_j Y_j Y_j^\top\|_{op} = \ns \cdot \|\frac{1}{\ns} \sum_j Y_j Y_j^\top\|_{op} \le O(\ns \cdot T_1)$. The third term can be bounded similarly, and the fourth term is at most $0$. Overall, with probability at least $1 - O(\frac{1}{L^4 d^2})$, $\sum_{i, j} \langle X_i, Y_j \rangle^2 \le O(\ns^2 \cdot S^2 + T_2 \cdot L^2 d \cdot \ns \cdot T_1)$, and otherwise, it is still bounded by $\ns^2 \cdot (L \sqrt{d})^4 = \ns^2 \cdot L^4 d^2$, since there are $\ns$ choices for each of $i, j$ and $\|X_i\|, \|Y_j\| \le L \sqrt{d}$. So, 
\[\Ema[]{\sum_{i, j} \langle X_i, Y_j \rangle^2} \le O(\ns^2 \cdot S^2 + T_2 \cdot L^2 d \cdot \ns \cdot T_1).\]

Overall, the mean of the statistic is $\ns^2 \cdot \|\mu\|^2$, and the variance is bounded by $O(\ns^3 \cdot T_1 \cdot \|\mu\|^2 + \ns^2 \cdot S^2 + L^2 \cdot \ns d \cdot T_1 \cdot T_2)$.

Our approach, based on this calculation, is the following. First, we show that the projected null hypothesis (i.e., $\cN(0, I)_{\proj}$) is good. Then, we show that for good distributions, we can create an algorithm that is replicable on good distributions, accepts $\cN(0, I)_{\proj}$, and rejects $\cN(\mu, I)_{\proj}$ whenever $\cN(\mu, I)_{\proj}$ is good and $\|\mu\| \ge \alpha$. Finally, by combining with the previous subsections, we can extend both the replicability and accuracy guarantees beyond good distributions.

\begin{lemma} \label{lem:projected-null-is-good}
    We have $\|\Ema[X_i \sim \cN(0, I)_{\proj}]{X_iX_i^\top}\|_{op} \le 1$.
    Also, if $X = \{X_1, \dots, X_{\ns}\}, Y = \{Y_1, \dots, Y_{\ns}\} \sim \cN(0, I)_{\proj}$, then with probability at least $1 - \frac{\rho}{L^4 d^2}$, the maximum matching $M_S(X, Y)$, as defined in \Cref{def:matching-graph} has size $0$. Hence, $\cN(0, I)_{\proj}$ is good.
\end{lemma}

\begin{proof}
    By symmetry, $\cN(0, I)_{\proj}$ has mean $0$ and covariance that is a scalar multiple of identity. Also, if $x \sim \cN(0, I)$ and $\hat{x}$ is the projection, $\|x\| \ge \|\hat{x}\|$, so $\Ema[]{\|\hat{x}\|^2} \le \Ema[]{\|x\|^2} \le d$. So, the trace of the covariance matrix is $d$, which means the operator norm of $\Ema[X_i \sim \cN(0, I)_{\proj}]{X_iX_i^\top}$ is at most $1$.

    Given $X_i, Y_j \sim \cN(0, I)$, the probability that $|\langle X_i, Y_j \rangle| \ge K \sqrt{d}$ is at most $e^{-\Omega(K)}$. So, assuming $K \ge \poly\log(\ns, d, 1/\rho),$ this probability is at most $\frac{\rho}{L^4 d^2 \ns^2}$. By taking a union bound over $\ns^2$ pairs $(X_i, Y_j)$, we have that the corresponding graph $G_S(X, Y)$ is in fact empty with at least $\frac{\rho}{L^4 d^2}$ probability. 
\end{proof}

We will need the following auxiliary proposition, which characterizes the norm of a spherical Gaussian after the projection.

\begin{proposition} \label{prop:projection-preserves-norms}
    For a vector $\mu$ and any $\alpha \le 1$, consider the mean of the distribution $\cN(\mu, I)_{\proj},$ i.e., where we sample $X_i \sim \cN(\mu, I)$ and project on to the ball of radius $L \sqrt{d}$. If $\mu$ is the origin, then the mean of the distribution $\cN(\mu, I)_{\proj}$ is also the origin, and if $\|\mu\| \ge \alpha$, then the mean of the distribution $\cN(\mu, I)_{\proj}$ has norm at least $\alpha/2$.
\end{proposition}

\begin{proof}
    The claim when $\mu$ is the origin is trivial by symmetry.

    First, assume $\alpha \le \|\mu\| \le L \sqrt{d}/2$. Consider sampling $X_1 \sim \cN(\mu, I)$ and $Y_1$ as the projection of $X_1$. Note that $\|X_1-Y_1\| = \|X_1-Y_1\| \cdot \I[X_1 \neq Y_1]$, since if $X_1 = Y_1$ then $\|X_1-Y_1\| = 0$. So, $ \Ema[]{\|X_1-Y_1\|}  = \Ema[]{\|X_1-Y_1\|^2 \cdot \I[X_1 \neq Y_1]} \le \sqrt{ \Ema[]{\|X_1-Y_1\|^2} \cdot \Prma[]{X_1 \neq Y_1}}$, by Cauchy-Schwarz. The probability that $X_1 \neq Y_1$ equals the probability that $\|X\| \ge L \sqrt{d}$, which for $\|\mu\| \le L\sqrt{d}/2$ and $L$ at least a sufficiently large constant, is at most $e^{-L}$. Moreover, $ \Ema[]{\|X_1-Y_1\|^2} \le 2 \cdot (\Ema[]{\|X_1\|^2} + \Ema[]{\|Y_1\|^2})$, and we know $\Ema[]{\|X_1\|^2} = \|\mu\|^2 + d$ and $\Ema[]{\|Y_1\|^2} \le L^2 d$ since $Y$ is always contained in the ball of radius $L \sqrt{d}$. Overall, this means $ \Ema[]{\|X_1-Y_1\|^2} \le 2 \cdot (L^2 d + (L \sqrt{d}/2)^2 + d) \le 4L^2 d$, which means $\sqrt{\Ema[]{\|X_1-Y_1\|^2} \cdot \Prma[]{X \neq Y}} \le 2 e^{-L/2} \cdot L \sqrt{d}.$ Assuming $L$ is a sufficiently large polylogarithmic multiple of $1/\alpha$ and $d$, this is at most $\alpha/2$. Hence, $\Ema[]{\|X_1-Y_1\|} \le \alpha/2$, and since $\Ema[]{X} = \mu$ which has norm at least $\alpha$, by the Triangle inequality $\|\Ema[]{Y_1}\| \ge \alpha/2$.

    Alternatively, suppose $\|\mu\| \ge L \sqrt{d}/2$. In that case, let $\hat{\mu}$ be the projection of $\mu$ onto the ball of radius $L \sqrt{d}$. Since the projection never dilates distances, for any point $x$ with projection $\hat{x},$ $\|\hat{\mu}-\hat{x}\| \le \|\mu-x\|$. So, for $x \sim \cN(\mu, I)$, $\|\hat{\mu}-\hat{x}\| \le \|\mu-x\| \le \sqrt{\|\mu-x\|^2} = \sqrt{d},$ which means that by Triangle inequality, $\|\Ema[]{\hat{x}}\| \ge \|\hat{\mu}\| - \sqrt{d}$. Since $\|\hat{\mu}\| = \min(L \sqrt{d}, \|\mu\|) \ge L \sqrt{d}/2$, we have $\|\Ema[]{\hat{x}}\| \ge \sqrt{d} \ge \alpha$.
\end{proof}

We are now ready to show that there is a replicable algorithm, at least for good distributions, that can distinguish between $\cN(0, I)$ and $\cN(\mu, I)$ with $\|\mu\| \ge \alpha$.

\begin{lemma} \label{lem:step-3}
    Suppose $\ns \ge \tilde{O}\left(\frac{\sqrt{d}}{\alpha^2 \rho} + \frac{\sqrt{d}}{\alpha \rho^2} + \frac{1}{\alpha^2 \rho^2}\right)$, and let $K, L, S, T_1, T_2$ be as in the beginning of this subsection.
    There exists an algorithm $\A_3$ with the following properties.
\begin{itemize}
    \item For any good distribution $\cD$, if given samples $X_1, \dots, X_{\ns}, Y_1, \dots, Y_{\ns} \sim \cD$ and $X_1', \dots, X_{\ns}', Y_1', \dots, Y_{\ns}' \sim \cD$, we have \newline $\Prma[]{\A_3(X_1, \dots, X_{\ns}, Y_1, \dots, Y_{\ns}; r) = \A_3(X_1', \dots, X_{\ns}', Y_1, \dots, Y_{\ns}'; r)} \ge 1-\rho$.
    \item The algorithm accepts $2 \ns$ samples from $\cN(0, I)_{\proj}$ with probability at least $0.99$.
    \item For any $\mu$ with $\|\mu\| \ge \alpha$, if $\cN(\mu, I)_{\proj}$ is good, the algorithm rejects $2 \ns$ samples from $\cN(\mu, I)_{\proj}$ with probability at least $0.99$.
\end{itemize}
\end{lemma}

\begin{proof}
    Consider sampling $X_1, \dots, X_{\ns}$, $Y_1, \dots, Y_{\ns}$ from $\cD$, and compute the statistic $Z = \langle \sum X_i, \sum Y_j \rangle$. For any distribution $\cD$ satisfying the assumptions in the lemma statement, the expectation of $Z$ is $\ns^2 \cdot \|\mu\|^2$ and the variance is $O(\ns^3 \cdot T_1 \cdot \|\mu\|^2 + \ns^2 \cdot S^2 + L^2 \cdot nd \cdot T_1 T_2)$. We upper bound the first term in the variance as
    \[\ns^3 \cdot T_1 \cdot \|\mu\|^2 \le  O\left( \ns^3 \cdot T_1 \cdot  \left( \frac{\|\mu\|^4}{\alpha^2} + \alpha^2 \right)\right),\]
    and instead use the variance upper bound of $O(\ns^3 \cdot T_1 \cdot \|\mu\|^4/\alpha^2 + \ns^3 \cdot T_1 \cdot \alpha^2 + \ns^2 \cdot S^2 + L^2 \cdot \ns d \cdot T_1 T_2)$.

    Our goal is to use \cref{thm:general-estimator} to prove the lemma. Towards that end, let $\tau_0(\ns) = 0, \tau_1(\ns) = \ns^2 \alpha^2/4$ be the null and alternate hypothesis thresholds. Note that $\Delta(\ns) = \tau_1(\ns)$ and 
    \[\p{1 + \max\bc{0, \frac{\Ema[]{Z(\ns)} - \tau_1(\ns)}{\Delta(\ns)}, \frac{\tau_0(\ns) - \Ema[]{Z(\ns)}}{\Delta(\ns)}}} = O\left( 1 + \frac{\|\mu\|^2}{\alpha^2} \right). \]
    It remains to pick an appropriate function $\sigma(\ns)$, which we do based on which terms in the variance of $Z$ dominate.

    \paragraph{Case 1:} $\Varma[]{Z(\ns)} \le O(\ns^3 \cdot T_1 \cdot \|\mu\|^4/\alpha^2)$. In this case, it suffices to choose $\sigma(\ns) = \Omega(\alpha \ns^{1.5}\sqrt{T_1})$, since this gives
    \[ \frac{\| \mu\|^2}{\alpha^2} \cdot  \alpha  \ns^{1.5} \sqrt{T_1} \ge \Omega\left(\sqrt{\Varma[]{Z(\ns)}}\right).\]The sampling breakpoints $\ns_t$ must satisfy $f(\ns_t) \le t/2$. Expanding this condition, and recalling that  $\Delta(\ns) = \Theta(\ns^2 \alpha^2)$, it suffices to pick $t$ such that 
    \[ t \ge \Omega \left( \sqrt{\frac{T_1}{\alpha^2 \ns_t}} \right), \]
    or in other words, recalling our setting of $T_1$,
    \[ \ns_t \ge \tilde{\Omega}\left( \frac{1}{\alpha^2 t^2} + \frac{\sqrt{d}}{\rho \alpha t} \right). \]

  \paragraph{Case 2:}$\Varma[]{Z(\ns)} \le O(\ns^3 \cdot T_1 \cdot \alpha^2)$. It suffices to pick $\sigma(\ns) = \Omega(\ns^{1.5} \alpha \sqrt{T_1})$ and following a similar reasoning as above, we arrive at the same lower bound of $\ns_t$ as Case 1.
 \paragraph{Case 3:} $\Varma[]{Z(\ns)} \le O(\ns^2 \cdot S^2)$. It suffices to pick $\sigma(\ns) = \Omega(\ns S)$ and following the same reasoning as in Case 1, it suffices to pick $t$ such that 
 \[ \ns_t \ge \tilde{\Omega}\left( \frac{\sqrt{d}}{t \alpha^2} \right). \]
   \paragraph{Case 4:}  $\Varma[]{Z(\ns)} \le O(L^2 \cdot \ns d \cdot T_1 T_2)$. It suffices to pick $\sigma(\ns) = \Omega(L \sqrt{\ns d T_1 T_2})$ and again the same reasoning implies that it suffices to pick $t$ such that 
   \[ \ns_t^3 \ge  \tilde{\Omega}\left( \frac{dT_1T_2}{t^2 \alpha^4} \right). \]
   We now simplify the above expression. Recalling our values of $T_1$ and $T_2$, we have $dT_1T_2 \ge \tilde{\Omega}\left( \frac{d}{\rho^2} + \frac{d^2}{\ns_t \rho^4}\right)$, and so it suffices to pick $\ns_t$ and $t$ such that 
   \[ \ns_t \ge \tilde{\Omega}\left(  \frac{d^{1/3}}{\rho^{2/3}t^{2/3} \alpha^{4/3}} + \frac{\sqrt{d}}{\rho \alpha \sqrt{t}} \right). \]

   Completing the case analysis, across all parameter settings, the breakpoint $\ns_t$ will be bounded by 
   \[ \ns_t =  \tilde{O}\left( \frac{1}{\alpha^2 t^2} + \frac{\sqrt{d}}{\rho \alpha  t}  +  \frac{\sqrt{d}}{t \alpha^2} +   \frac{d^{1/3}}{\rho^{2/3}t^{2/3} \alpha^{4/3}} + \frac{\sqrt{d}}{\rho \alpha \sqrt{t}} \right).\]

   Applying \cref{thm:general-estimator} with $t = \rho$ and using the fact that 
   \[ \frac{1}{\rho \alpha^2} + \frac{1}{\rho^2 \alpha} \ge \Omega\left( \frac{1}{\rho^{4/3}\alpha^{4/3}} \right), \]
   gives us a sample complexity of 
   \[  \tilde{O}\left( \frac{1}{\alpha^2 \rho^2} + \frac{\sqrt{d}}{\alpha \rho^2}  +  \frac{\sqrt{d}}{ \rho \alpha^2} \right),\]
   meaning that as long as $\ns \ge \tilde{O}\left(\frac{\sqrt{d}}{\alpha^2 \rho} + \frac{\sqrt{d}}{\alpha \rho^2} + \frac{1}{\alpha^2 \rho^2}\right)$, we have an $O(\rho)$-replicable algorithm (at least, a replicable algorithm on distributions $\cD$ satisfying the assumption), that accepts w.h.p. on any good distribution with mean $0$ and rejects w.h.p. on any good distribution with mean at least $\alpha/2$ in absolute value.

    By \Cref{lem:projected-null-is-good} and \Cref{prop:projection-preserves-norms}, $\cN(0, I)_{\proj}$ is good and has mean $0$, which means that the algorithm accepts w.h.p. Also, by \Cref{prop:projection-preserves-norms}, if $\|\mu\| \ge \alpha$, then $\cN(0, I)_{\proj}$ has mean at least $\alpha/2$. So, either $\cN(0, I)_{\proj}$ is not good or the algorithm rejects w.h.p. This completes the lemma.
\end{proof}

\paragraph{Putting things together.} To summarize, our overall algorithm is to set the parameters $L, S, T_1, T_2$ as in the beginning of \cref{sec:gaussian_C}, set $\ns = \tilde{O}\left(\frac{\sqrt{d}}{\alpha^2 \rho} + \frac{\sqrt{d}}{\alpha \rho^2} + \frac{1}{\alpha^2 \rho^2}\right)$, and then run $\A_1$ from \Cref{lem:matrix-chernoff-application-2} with $\ns$ samples, $\A_2$ from \Cref{cor:step-2} with $\ns$ fresh samples, and $\A_3$ from \Cref{lem:step-3} with $2\ns$ fresh samples. If the distribution $\cD$ is bad, then either $\A_1$ or $\A_2$ will reject with $1-\rho$ probability, which means we also have $O(\rho)$-replicability. If $\cD$ is good, all of $\A_1, \A_2, \A_3$ are $O(\rho)$-replicable, so if we run them with fresh samples and fresh randomness, we still have $O(\rho)$-replicability. Finally, if given samples from $\cN(0, I)_{\proj}$, we pass all three steps with high probability, and if given samples from $\cN(\mu, I)_{\proj}$, we fail at least one of the three steps with high probability. Hence, this completes the proof of \Cref{thm:gaussian-testing-main}.

\subsection{Lower Bound}

We recall our lower bound on replicable Gaussian mean testing.

\repgaussianlb*

We start by proving a generalization of \Cref{lem:label_invariant} that allows us to convert replicable algorithms into (canonical) replicable algorithms that only depend on the \emph{sufficient statistic} of the data, assuming the samples are drawn from a parameterized distribution $\cD_\theta: \theta \in \Theta$. However, this will come at the cost of the new algorithm being only \emph{weakly} replicable (recall the definition of weakly replicable in \Cref{def:weak_replicability}).
In our application, the parameterized distributions are $\cN(\mu, I)$ for any $\mu \in \R^d$ (i.e., $\theta = \mu$ and $\Theta = \R^d$). 


\begin{lemma}[Sufficient Statistic Invariant Algorithm]\label{lem:suff_stat_invariant}
    Given a parameterized distribution $\cD_\theta$ and regions $\Theta_{\accept}, \Theta_{\reject}$, suppose $\A_0(X; r)$ is a $\rho$-replicable algorithm that distinguishes between distributions $\cD_\theta: \theta \in \Theta_{\accept}$ and $\cD_\theta: \theta \in \Theta_{\reject}$. Then, there exists a weakly $\rho$-replicable algorithm $\A_2(S(X); r)$ that solves the same problem on $\ns$ samples with the same accuracy and only depends on the sufficient statistic $S(X)$. Moreover, $\A_2$ has the canonical property, meaning $r \sim \Unif[0, 1]$ and there exists a deterministic function $q(S(X)) \in [0, 1]$ where $\A_2(S(X); r) = 1$ if $r \le q(S(X))$ and $0$ otherwise.
\end{lemma}

\begin{proof}
    Let $\A_1(X;r)$ be the algorithm defined in Lemma~\ref{lem:canonical} with the deterministic function $f:\XX^{\ns} \rightarrow [0,1]$.
    Choose an arbitrary $\theta$, and consider the following deterministic function of the sample set $X$, $q:S(\XX^{\ns}) \rightarrow[0,1]\,$:
\begin{equation} \label{eq:q_S}
    q(S(X)) \coloneqq \Ema[Y \sim \cD_{\theta}^{\otimes \ns}|S(Y) = S(X)]{f(Y)}\,.
\end{equation}
    (Recall that $Y \sim \cD_{\theta}^{\otimes \ns}|S(Y) = S(X)$ means we generate $Y = (Y_1, \dots, Y_{\ns}) \overset{i.i.d.}{\sim} \cD_{\theta}$ conditional on $S(Y) = S(X)$.)
    Note that by definition of sufficient statistic, the choice of ${\theta}$ does not affect the conditional distribution $Y \sim \cD_\theta^{\otimes \ns}|S(Y) = S(X)$.
    Moreover, the right-hand side of \eqref{eq:q_S} only depends on $X$ through $S(X)$, so $q$ can be defined.
    The algorithm $\A_2(X; r)$ operates similarly to $\A_1(X; r)$, except that it uses $q$ instead of $f$. For a random seed $r \sim \Unif([0, 1])$, $\A_2(X; r)$ outputs \accept if $r \le q(S(X))$, and \reject otherwise.

    To check the accuracy of $\A_2$, first consider any $\cD_\theta$, where $\theta \in \Theta_{\accept}$. Then, 
\begin{align*}
    \Prma[r,X\sim \cD_{\theta}^{\otimes \ns}] {\A_2(X;r) = \text{\accept}}
    & = \Ema[X\sim \cD_{\theta}^{\otimes \ns}] {q(S(X))} \\
    & = \Ema[X\sim \cD_{\theta}^{\otimes \ns}] {\Ema[Y\sim \cD_{\theta}^{\otimes \ns}|S(Y) = S(X)] {f(Y)}} \\
    &= \Ema[Y\sim \cD_{\theta}^{\otimes \ns}]{f(Y)} \\
    &= \Prma[r,Y\sim \cD_{\theta}^{\otimes \ns}] {\A_1(Y;r) = \text{\accept}} \\
    & = \Prma[r,Y\sim \cD_{\theta}^{\otimes \ns}]{\A_0(Y;r) = \text{\accept}}
    \,. 
        \tag{Using Lemma~\ref{lem:canonical}, Eq.~\eqref{eq:A_0_accept}}
\end{align*}
    The third line holds since if we sample $X \sim \cD_{\theta}^{\otimes \ns}$ and $Y \sim \cD_{\theta}^{\otimes \ns} | S(Y) = S(X),$ it is equivalent to sample $S(Y_1, \dots, Y_{\ns})$ where $Y_1, \dots, Y_{\ns} \sim \cD_\theta,$ and then $Y_1, \dots, Y_{\ns}$ have the right conditional distibution given $S(Y_1, \dots, Y_{\ns})$. So, $Y_1, \dots, Y_{\ns}$ in fact have the same marginal distribution as $\cD_\theta^{\otimes \ns}$. The same argument holds for $\theta \in \Theta_{\text{reject}}$, so $\A_2$ has the same probabilities of outputting \accept and \reject as $\A_0$, and thus inherits the accuracy guarantees of $\A_0$.

Next, we show weak replicability of $\A_2$. For any distribution $\cD_{\theta}$, similar to Lemma~\ref{lem:canonical}, we have: 
\begin{equation} \label{eq:pr_rep_A_2_part1}
    \begin{split}
\Prma[r,X,X'\sim \cD_{\theta}^{\otimes \ns}]
    {\A_2(X;r) \not = \A_2(X';r)} & = \Ema[X,X']
    {\Prma[r]{\A_2(X;r) \not = \A_2(X';r)}} 
    \\& = 
    \Ema[X,X']
    {\abs{q(S(X)) ~-~q(S(X'))}}\,,
    \end{split}
\end{equation}
where in the last line, we use the structure of $\A_2$. Using the definition of $q$, we have:
\begin{align*} 
    \Prma[r,X,X'\sim \cD_{\theta}^{\otimes \ns}] {\A_2(X;r) \not = \A_2(X';r)}
    & = \Ema[X,X'] {\abs{q(S(X)) ~-~q(S(X'))}} \\
    &= \Ema[X, X']{\left|\Ema[Y \sim \cD_{\theta}^{\otimes \ns}|S(Y)=S(X)]{f(Y)}-\Ema[Y' \sim \cD_{\theta}^{\otimes \ns}|S(Y')=S(X')]{f(Y')}\right|} \\
    &\leq \Ema[X, X']{\Ema[Y \sim \cD_{\theta}^{\otimes \ns}|S(Y)=S(X), Y' \sim \cD_{\theta}^{\otimes \ns}|S(Y')=S(X')]{|f(Y)-f(Y')}|} \tag{Via triangle inequality} \\
    &= \Ema[Y,Y' \sim \cD_\theta^{\otimes \ns}]    {|f(X)-f(Y)|} \\    
    & = \Prma[Y,Y' \sim \cD_\theta^{\otimes \ns}, r \sim \Unif{[0, 1]}] {\A_1(Y; r) \neq \A_1(Y'; r)} \\
    & \leq \rho. \tag{Since $\A_1$ is $\rho$-replicable}
\end{align*}
Hence, the proof is complete. 
\end{proof}

Next, we show that in the Gaussian mean testing setting, the algorithm can be assumed to only depend on the Euclidean norm of the empirical mean of the samples.

\begin{lemma} \label{lem:gaussian-symmetry-conversion}
    Let $\A_0(X; r)$ be a $\rho$-replicable algorithm that distinguishes between $\cN(\mu, I): \mu = 0$ and $\cN(\mu, I): \|\mu\| \ge \alpha$. Then, there exists another weakly $\rho$-replicable algorithm $\A_3(\|\bar{X}\|, r)$, ony depending on $X$ through the norm of its empirical mean $\|\bar{X}\| = \left\| \frac{X_1 + \cdots + X_{\ns}}{\ns} \right\|$, that also distinguishes between $\cN(\mu, I): \mu = 0$ and $\cN(\mu, I): \|\mu\| \ge \alpha$. Moreover, $\A_3$ has the canonical form, meaning $r \sim \Unif[0, 1]$ and there exists a deterministic function $q_3(\|\bar{X}\|) \in [0, 1]$ where $\A_3(\|\bar{X}\|; r) = 1$ if $r \le q_3(\|\bar{X}\|)$ and $0$ otherwise.
\end{lemma}

\begin{proof}
    By \Cref{prop:gaussian_sufficient_statistic}, $\bar{X} = \frac{X_1 + \cdots + X_{\ns}}{\ns}$ is a sufficient statistic for the parameterized distribution $\cN(\mu, I)$. By \Cref{lem:suff_stat_invariant}, we can start with the weakly $\rho$-replicable algorithm $\A_2(\bar{X}; r)$, which distinguishes between $\mu = 0$ and $\|\mu\| \ge \alpha$ by sampling $r \sim \Unif[0, 1]$ and outputting $1$ if $r \le q(\bar{X})$, for some deterministic function $q$.

    Let $O_d$ represent the uniform (Haar) measure over $d \times d$ orthogonal matrices. Next, let $X_1, \dots, X_{\ns} \sim \cN(\mu, I)$, and let $q_3 = \Ema[H \sim O_d]{q(H(\bar{X}))}$. Then, $H(\bar{X})$ is a random vector on the sphere of radius $\|\bar{X}\|$, so $q_3$ only depends on $\|\bar{X}\|$. So, for any $X_1, \dots, X_{\ns},$ we can define $q_3(\|\bar{X}\|) = \Ema[H \sim O_d]{q(H(\bar{X}))}$.

    First, we check that the algorithm $\A_3(\|\bar{X}\|; r)$, which outputs $1$ if $r \le q_3(\|\bar{X}\|)$ and $0$ otherwise, is accurate. If $X_1, \dots, X_{\ns} \sim \cN(\mu, I),$ then $\bar{X} \sim \cN(\mu, \frac{I}{\ns})$. So, $H(\bar{X})$ has the distribution of $\cN(\mu, \frac{I}{\ns})$ followed by a random rotation. This is the same as first randomly rotating $\mu$ to get some $\mu'$ with $\|\mu'\| = \|\mu\|$, and then sampling from $\cN(\mu', \frac{I}{\ns})$. So, if $\beta = \|\mu\|,$ then
\begin{align*}
    \Prma[r, X \sim \cN(\mu, I)^{\otimes \ns}] {\A_3(X; r) = \accept}
    &= \Ema[X \sim \cN(\mu, I)^{\otimes \ns}, H \sim O_d] {q(H(\bar{X}))} \\
    &= \Ema[\mu': \|\mu'\| = \beta]{\Ema[X \sim \cN(\mu', I)^{\otimes \ns}]{q(\bar{X})}}.
\end{align*}
    If $\mu = 0$, then $\mu' = 0$ with probability $1$, and $\Ema[X \sim \cN(\mu', I)^{\otimes \ns}]{q(\bar{X})} \ge 1-\delta$. Thus, $\Prma[r, X \sim \cN(\mu, I)^{\otimes \ns}] {\A_3(X; r) = \accept} \ge 1-\delta$ as well. Alternatively, if $\beta = \|\mu\| \ge \alpha$, then $\|\mu'\| \ge \alpha$ with probability $1$, and $\Ema[X \sim \cN(\mu', I)^{\otimes \ns}]{q(\bar{X})} \le \delta$. Thus, $\Prma[r, X \sim \cN(\mu, I)^{\otimes \ns}] {\A_3(X; r) = \accept} \le \delta$ as well. Hence, the same accuracy bounds hold.
    
    To prove weak replicability, suppose that $X, X' \sim \cN(\mu, I)^{\otimes \ns}$. Then, 
\begin{align*} 
    \Prma[r,X,X'\sim \cN(\mu, I)^{\otimes \ns}] {\A_3(X;r) \not = \A_3(X';r)}
    & = \Ema[X,X'\sim \cN(\mu, I)^{\otimes \ns}] {\abs{q_3(\|X\|) ~-~q_3(\|X'\|)}} \\
    &= \Ema[X, X']{\left|\Ema[H \sim O_d]{q(H(\bar{X}))}-\Ema[H \sim O_d]{q(H(\bar{X'}))}\right|} \\
    &\leq \Ema[X, X', H \sim O_d]{\left|q(H(\bar{X}))-q(H(\bar{X'}))\right|} \tag{Via triangle inequality} \\
    &= \Ema[H \sim O_d]{\Ema[X, X']{\left|q(H(\bar{X}))-q(H(\bar{X'}))\right|}}.
\end{align*}
    If you fix $H$, then if $X \sim \cN(\mu, I)^{\otimes \ns}$, then $H(\bar{X})$ has the same distribution as the empirical mean of $\ns$ samples drawn from $\cN(H(\mu), I)$, by rotational symmetry of the Gaussian. Hence, we have
\begin{align} 
    \Prma[r,X,X'\sim \cN(\mu, I)^{\otimes \ns}] {\A_3(X;r) \not = \A_3(X';r)}
    &= \Ema[H \sim O_d]{\Ema[X, X']{\left|q(H(\bar{X}))-q(H(\bar{X'}))\right|}} \nonumber \\
    &= \Ema[H \sim O_d]{\Ema[X, X' \sim \cN(H(\mu), I)^{\otimes \ns}]{\left|q(\bar{X})-q(\bar{X'})\right|}} \nonumber \\
    &= \Ema[H \sim O_d]{\Prma[X, X' \sim \cN(H(\mu), I)^{\otimes \ns}, r \sim \Unif{[0, 1]}]{\A_2(X; r) \neq \A_2(X'; r)}}. \tag{Definition of $\A_2$ and $q$} \\
    &\ge \Ema[H \sim O_d]{1-\rho} = 1-\rho \tag{Weak replicability of $\A_2$}
\end{align}
    Therefore, the overall probability $\Prma[r,X,X'\sim \cN(\mu, I)^{\otimes \ns}] {\A_3(X;r) \not = \A_3(X';r)}$ is at least $1-\rho$. Hence, $\A_3$ is also $\rho$-weakly replicable.
\end{proof}

We need the following auxiliary lemma about total variation distance between norms of Gaussians.

\begin{lemma} \label{lem:norm-gaussian-tv}
    Let $\beta_1, \beta_2 \ge 0$, and $d, \ns \in \N$ be positive integers. Let $Z_1$ be the distribution in which we pick an arbitrary $\mu_1 \in \R^d$ of norm $\beta_1$, sample $z_1 \sim \cN(\mu_1, \frac{I}{s})$, and $Z_1 = \|z_1\|$. (By rotational symmetry of Gaussians, note that the choice of $\mu_1$ doesn't affect the distribution of $Z_1$.) Likewise, let $Z_2$ be the distribution in which we pick an arbitrary $\mu_2 \in \R^d$ of norm $\beta_2$, sample $z_2 \sim \cN(\mu_2, \frac{I}{\ns})$, and $Z_2 = \|z_2\|$.

    Then, if $\ns \le \max\left(\frac{c}{(\beta_1-\beta_2)^2}, \frac{c \sqrt{d}}{|\beta_1^2-\beta_2^2|}\right)$, then $\dtv(Z_1, Z_2) \le 0.5$.
\end{lemma}

\begin{proof}
    It is equivalent to look at $\dtv(Z_1^2, Z_2^2)$ since $Z_1, Z_2 \ge 0$ always. Moreover, we may assume WLOG that $\mu_1 = (\beta_1, 0, \dots, 0) \in \R^d$ and $\mu_2 = (\beta_2, 0, \dots, 0) \in \R^d$.

    Then, $Z_1^2 = (\beta_1 + \frac{z_1}{\sqrt{s}})^2 + (\frac{z_2}{\sqrt{s}})^2 + \cdots + (\frac{z_d}{\sqrt{s}})^2$, and $Z_2^2 = (\beta_2 + \frac{z_1}{\sqrt{s}})^2 + (\frac{z_2}{\sqrt{s}})^2 + \cdots + (\frac{z_d}{\sqrt{s}})^2$, where $z_1, \dots, z_d \sim \cN(0, 1)$.

    First, if $\beta_1-\beta_2 \le \frac{0.1}{\sqrt{s}}$, then the total variation distance between $\beta_1 + \frac{z}{\sqrt{s}}$ and $\beta_2 + \frac{z}{\sqrt{s}}$ is at most $0.5$ if $z \sim \cN(0, 1)$. Hence, the total variation distance between $(\beta_1 + \frac{z_1}{\sqrt{s}})^2$ and $(\beta_2 + \frac{z_1}{\sqrt{s}})^2$ is also at most $0.5$. Thus, we can couple the values $z_2, \dots, z_d$, to obtain that $\dtv(Z_1^2, Z_2^2) \le 0.5$, as long as $\beta_1-\beta_2 \le \frac{0.1}{\sqrt{s}}$, or equivalently, if $s \le \frac{0.01}{(\beta_1-\beta_2)^2}$.

    For any $d \ge 2$, note that $(\frac{z_2}{\sqrt{s}})^2 + \cdots + (\frac{z_d}{\sqrt{s}})^2 = \frac{1}{s} \cdot \chi_{d-1}^2$. By coupling $z_1$, we have that 
\begin{align*}
    \dtv(Z_1^2, Z_2^2) 
    &\le \Ema[z_1 \sim \cN(0, 1)]{\dtv\left((\beta_1+\frac{z_1}{\sqrt{s}})^2 + \frac{1}{s} \cdot \chi_{d-1}^2, (\beta_2+\frac{z_1}{\sqrt{s}})^2 + \frac{1}{s} \cdot \chi_{d-1}^2\right)} \\
    &= \Ema[z_1 \sim \cN(0, 1)]{\dtv\left(\frac{1}{s} \cdot \chi_{d-1}^2, (\beta_2^2-\beta_1^2)+(\beta_2-\beta_1) \cdot \frac{2z_1}{\sqrt{s}}+ \frac{1}{s} \cdot \chi_{d-1}^2\right)} \\
    &= \Ema[z_1 \sim \cN(0, 1)]{\underbrace{\dtv\left(\chi_{d-1}^2, (\beta_2^2-\beta_1^2) \cdot s + 2(\beta_2-\beta_1) \sqrt{s} \cdot z_1 + \chi_{d-1}^2\right)}_{T}}.
\end{align*}
    By \Cref{prop:tv-chi-square-shifted}, as long as $|\beta_2^2-\beta_1^2| \cdot s + 2 |\beta_2-\beta_1| \cdot \sqrt{s} \cdot |z_1| \le 0.001 \sqrt{d-1}$, the expression $T$ is at most $0.1$. For any positive $\beta_1, \beta_2$, $|\beta_2-\beta_1| \le \sqrt{|\beta_1^2-\beta_2^2|}$ Hence, as long as $|z_1| \le 2$ and $s \le \frac{c \cdot \sqrt{d}}{|\beta_1^2-\beta_2^2|}$ for a sufficiently small constant $c$, we have that $|\beta_1^2-\beta_2^2| \cdot s \le c \cdot \sqrt{d}$ and $2|\beta_2-\beta_1| \cdot \sqrt{s} \cdot |z_1| \le 4 \cdot \sqrt{s \cdot |\beta_2^2-\beta_1^2|} \le 4 \sqrt{c \sqrt{d}}.$ So, if $|z_1| \le 2$ and $c$ is sufficiently small, $|\beta_2^2-\beta_1^2| \cdot s + 2 |\beta_2-\beta_1| \cdot \sqrt{s} \cdot |z_1| \le 0.001 \sqrt{d-1} \le 0.001 \sqrt{d-1},$ and $T \le 0.1$. Since $T \le 1$ with probability $1$, and $|z_1| \le 2$ with at least $0.9$ probability, we have $\Ema[z_1 \sim \cN(0, 1)]{T} \le 0.9 \cdot 0.1 + 0.1 \cdot 1 \le 0.5$. Overall, this means as long as $d \ge 2$ and $\ns \le \frac{c \sqrt{d}}{|\beta_1^2-\beta_2^2|}$ for sufficiently small $c$, $\dtv(Z_1^2, Z_2^2) \le 0.5$.

    In summary, if $\ns \le \frac{0.01}{(\beta_1-\beta_2)^2},$ or $d \ge 2$ and $\ns \le \frac{c \sqrt{d}}{|\beta_1^2-\beta_2^2|}$ for sufficiently small $c$, $\dtv(Z_1^2, Z_2^2) \le 0.5$. Note that if $d = 1$, and $c \le 0.01$, then $s \le \frac{c \sqrt{d}}{|\beta_1^2-\beta_2^2|} \le \frac{0.01}{|\beta_1^2-\beta_2^2|} \le \frac{0.01}{(\beta_1-\beta_2)^2}$. Hence, it suffices for $\ns \le \max\left(\frac{c}{(\beta_1-\beta_2)^2}, \frac{c \sqrt{d}}{|\beta_1^2-\beta_2^2|}\right)$.
\end{proof}

We are now ready to prove the main lower bound.

\begin{proof}[Proof of \Cref{thm:gaussian-testing-lb}]
    First, we may assume that $\rho \le 0.001$, as otherwise the lower bound equals $\Omega\left(\frac{\sqrt{d}}{\alpha^2}\right)$, which is required even for non-replicable testers.
    
    Let $\A_3(\|\bar{X}\|; r)$ be the weakly $\rho$-replicable algorithm on $X_1, \dots, X_{\ns} \sim \cN(\mu, I)$, following \Cref{lem:gaussian-symmetry-conversion}, and let $q_3: \R_{\ge 0} \to [0, 1]$ represent the function where $\A_3\left(X; r\right) = 1$ if $r \le q_3\left(\left\|\frac{X_1 + \cdots + X_{\ns}}{\ns}\right\|\right)$ and $0$ otherwise.

    First, suppose that $\ns \le \frac{c}{\alpha^2 \rho^2}$, where we recall that $c$ is a sufficiently small constant. Let $t = \lfloor \frac{1}{300 \rho} \rfloor$, and define $\beta_i = \alpha \cdot \frac{i}{t}$ for each $i = 0, 1, \dots, t.$ By \Cref{lem:norm-gaussian-tv}, if $Z_i = \|\cN(\mu_i, \frac{I}{s})\|$ where $\|\mu_i\| = \beta_i$, then $\dtv(Z_i, Z_{i+1}) \le 0.5$ for all $0 \le i \le t-1$, since $s \le \frac{c}{(\beta_i-\beta_{i+1})^2}$. Hence, by \Cref{lem:chaining_lemma}, as $\A_3$ is a weakly $\rho$-replicable algorithm, it cannot distinguish between $Z_0$ and $Z_t$, and therefore cannot distinguish between $\ns$ samples from $\cN(0, I)$ and $\ns$ samples from $\cN(\mu, I)$ with $\|\mu\| = \|\mu_t\| = \alpha$.

    Alternatively, suppose that $\ns \le \frac{c \sqrt{d}}{\alpha^2 \rho}$. Again, let $t = \lfloor \frac{1}{300 \rho} \rfloor$, and this time define $\beta_i = \alpha \cdot \sqrt{\frac{i}{t}}$ for each $i = 0, 1, \dots, t.$ By \Cref{lem:norm-gaussian-tv}, if $Z_i = \|\cN(\mu_i, \frac{I}{s})\|$ where $\|\mu_i\| = \beta_i$, then $\dtv(Z_i, Z_{i+1}) \le 0.5$ for all $0 \le i \le t-1$, since $s \le \frac{c \sqrt{d}}{|\beta_i^2-\beta_{i+1}^2|}$. Hence, by \Cref{lem:chaining_lemma}, a weakly $\rho$-replicable algorithm cannot distinguish between $Z_0$ and $Z_t$ with $\ns$ samples, and therefore cannot distinguish between $\ns$ samples from $\cN(0, I)$ and $\ns$ samples from $\cN(\mu, I)$ with $\|\mu\| = \|\mu_t\| = \alpha$.
\end{proof}

\section{Replicable Hypothesis Selection via Testing}
\label{sec:selection}

 In hypothesis selection there is a known collection of distributions $\mathcal{H}=\{H_1,\dots, H_n\}$ all over the same domain. Given samples from an unknown distribution $P$, our goal is to output an $i\in [n]$ such that 
 \begin{align}\label{eq:HS-definition}
 d_{TV}(H_i,P)\leq C\cdot \min_{H\in\mathcal{H}}d_{TV}(H,P)+\eps.
\end{align}
 for some constant $C$ and desired small $\eps$. Without replicability, the sample complexity of this problem is well-studied. Using $O(\frac{\log n}{\eps^2})$ samples, we can achieve the above guarantee with $C=3$ \cite{devroye2001combinatorial}. Moreover, it is known that obtaining the guarantee with any constant $C<3$ requires polynomially many samples in the domain size \cite{bousquet2019optimal}. In particular if the domain is infinite, we cannot get any finite sample complexity. Our main result in this section is a replicable algorithm for hypothesis selection.

Using our improved replicable coin testing algorithm from \cref{sec:coin-testing-upper-bound} as a key subroutine, we obtain the following sample complexity bound on replicable hypothesis selection. The challenge in hypothesis-selection is that we have a huge space of potential outputs (all $n$ hypothesis), but we want to be stable in our outputs (for replicability). Our idea to get around this issue is to view hypothesis selection as an (adaptive) sequence of coin testing problems, by partitioning the hypothesis in a binary-tree fashion. At each node in the tree, we have to pick if we want to descend down to the left or right branch which corresponds to one instance of the coin testing problem. The initial node contains all the $n$ hypothesis and the final leaf nodes only contain one hypothesis.

 \begin{theorem}\label{thm:hypo_selection_main}
Let $0\leq \eps,\rho\leq 1$. There exists a $\rho$-replicable algorithm for hypothesis selection with optimal multiplicative approximation $C=3$ which succeeds with high probability in $n$ and takes samples
\begin{equation*}
    O\p{\frac{\log^5 n}{\eps^2 \rho^2}}
\end{equation*}
in the worst-case, and
\begin{equation*}
    O\p{\frac{\log^5 n}{\eps^2 \rho}}
\end{equation*}
in expectation.
\end{theorem}

\begin{remark}\label{remark:gaussian_mean_estimation}
An interesting question is if the sample complexity of Theorem \ref{thm:hypo_selection_main} can be improved. We briefly note that this sample complexity cannot be improved by more than a $\log(n)^3 \log(1/\eps)$ factor. This follows from Corollary 1.6 in \cite{hopkins2024replicability} where it is shown that replicable mean estimation for a unit covariance Gaussian in $d$ dimensions (up to additive error $\eps$) requires $\Theta(\frac{d^2}{\eps^2\rho^2})$ samples. One way to solve this mean estimation problem is to discretize the unit sphere and solve hypothesis selection (any $C = O(1)$ factor suffices in \eqref{eq:HS-definition}) among $n = (1/\eps)^{O(d)}$ different possible hypothesis. Since this is a special case of the general hypothesis selection problem, $\Omega(\frac{\log^2(n)}{\eps^2 \rho^2 \log(1/\eps)})$ samples must be necessary.
\end{remark}

\begin{proof}
To prove the theorem, we first recall how the algorithm from~\cite{devroye2001combinatorial} works. For distinct $i,j$, define $S_{ij}=\{x\in [d]\mid H_i(x)\leq H_j(x)\}$ and the \emph{semi-distance} $w_j(H_i)=|H_i(S_{ij})-P(S_{ij})|$. Also define $W_i=\max_{j\neq i} w_j(H_i)$ and $W=\min_{i}W_i$. One can show that for any $i$ such that $W_i=W$, the hypothesis $H_i$ satisfies~\cref{eq:HS-definition} with $C=3$ and $\eps=0$. Moreover, if $W_i\leq W+\eps$, then $H_i$ satisfies~\cref{eq:HS-definition} with $C=3$ and additive error $\eps$. 
Now  using $m=O(\frac{\log n}{\eps^2})$ samples, the algorithm from~\cite{devroye2001combinatorial} provides an estimate $\hat w_j(i)$ of $ w_j(i)$ satisfying that $ |\hat w_j(i)-w_j(i)|\leq \eps$ for all pairs of distinct $i,j$ with high probability in $n$. For each hypothesis $H_i$ they then compute the quantity $\hat W_i=\max_{j\neq i}\hat w_j(H_i)$ and their algorithm returns an index $i$ such that $\hat W_i$ is minimal. Since all estimated semi-distances are within $\eps$ of the true semi-distance, it follows that $W_i\leq W+\eps$, and they thus obtain the desired approximation guarantee.

For our replicable algorithm, we assume for simplicity that $n$ is a power of two. Let us define $\eps_0=\eps/\lg n$ and $\rho_0=\rho/\lg n$ and further for integers $0\leq j\leq \lg n$ and $0\leq i< 2^j$, we define 
\[
A_{i,j}=\{i 2^{n-j}+k: 1\leq k\leq 2^{n-j}\},
\]
so that $|A_{i,j}|=n/2^{j}$ and $A_{i,j}=A_{2i,j+1}\cup A_{2i+1,j+1}$ for $0\leq j<\lg n$. Further define $\mathcal{H}_{i,j}:=\{H_k: k\in A_{i,j}\}$.
Finally, set $i_0=0$.

For $1\leq j\leq \lg n$, our algorithm iteratively computes an index $i_j$ such that $0\leq i_j< 2^j$ and with high probability in $n$, the set $\mathcal{H}_{i_j,j}$ contains a hypothesis $H_i$ with $W_i\leq W+j\eps_0$. Since $|\mathcal{H}_{i_{\lg n},\lg n}|=1$, with high probability, the single hypothesis $H_i$ in $\mathcal{H}_{i_{\lg n},\lg n}$ has $W_i\leq W+\eps$ and returning this hypothesis, the desired approximation guarantee follows. 

To construct $i_j$ from $i_{j-1}$, we will reduce the subproblem to $\rho_0$-replicable coin testing.
We will consider the outcome of running the algorithm from ~\cite{devroye2001combinatorial} as a single sample from the coin.
In a single of these runs, when computing the estimate $\hat W_i$ for a given hypothesis $H_i$, we do so with respect to the full set of hypotheses $\mathcal{H}$ by taking $\hat W_i=\max_{j\in [n]}\hat w_j(H_i)$. Denote by $p$ the probability that in a single run, we return a hypothesis in $\mathcal{H}_{2i_j,j+1}$. The probability of returning a hypothesis in $\mathcal{H}_{2i_j+1,j+1}$ is thus $1-p$.
We will run our algorithm from \cref{thm:optimal-coin-testing} on the resulting coin testing problem with $p_0=1/2$, $q_0=3/4$, $\rho_0$-replicability, and failure probability $\delta= \poly(1/n)$.
As replicable coin testing is technically defined for testing $p = p_0$ or $p \geq q_0$, we will duplicate this process twice flipping the semantic meaning of heads so that the algorithm is correct with high probability when $p \leq 1/4$ or $p \geq 3/4$.
To analyze our final algorithm, we need to argue about sample complexity, approximation guarantee, and replicability.

\paragraph{Sample Complexity:} 
Let $s_j$ be the sample complexity of the coin testing algorithm at level $j$.
By \cref{thm:optimal-coin-testing}, $s_j$ is upper bounded in expectation by $\Ema[]{s_j} = O\p{\log n + \frac{1}{\rho_0}} = O\p{\frac{\log n}{\rho}}$.
For each $j=1,\dots, \lg n$, the algorithm computes $s_j$ maximum semi-distance estimators each based on $O\p{\frac{\log n}{\eps_0^2}}$ samples. Thus the total number of samples is $O\p{\frac{\log^5 n}{\eps^2 \rho}}$ in expectation.

Via \cref{prop:markov-translate} by stopping early and running the $O\p{\frac{\log n}{\eps^2}}$ sample size non-replicable algorithm to ensure correctness, the sample complexity becomes $O\p{\frac{\log^5 n}{\eps^2 \rho^2}}$ in the worst-case.

\paragraph{Approximation:}
We prove inductively that $\mathcal{H}_{i_j,j}$ contains a hypothesis $H_i$ with $W_i\leq W+j\eps_0$. This is trivially true for $j=0$, so suppose inductively that it holds for some $j$, and let us show that with high probability it also holds for  $j+1$.
Suppose first that $1/4\leq p\leq 3/4$. With high probability, the algorithm from~\cite{devroye2001combinatorial} returns a hypothesis $H_i$ with 
\[
W_i\leq \min_{j\in A_{i_j,j}} W_j+\eps_0\leq W+\eps_0(j+1),
\]
where the last step uses the inductive hypothesis. In particular, in the case $1/4\leq p\leq 3/4$, then both of $\mathcal{H}_{2i_j,j+1}$ and $\mathcal{H}_{2i_j+1,j+1}$ must contain a hypothesis $H_i$ with $W_i\leq W+(j+1)\eps_0$ and we are thus happy regardless of whether $i_{j+1}=2i_j$ or $i_{j+1}=2i_j+1$. Now assume that $p<1/4$ (the case $p>3/4$ is similar). By the guarantee of replicable coin testing (\cref{thm:optimal-coin-testing}), with high probability in $n$, $i_{j+1}=2i_j+1$ , and since $p<1/4$, $\mathcal{H}_{2i_j+1,j+1}$ does indeed contain a hypothesis $H_i$ with $W_i\leq \min_{j\in A_{i_j,j}} W(H_j)$. The claim then again follows from the inductive hypothesis.

Union bounding over $j=0,\dots, \lg n-1$, we obtain that with high probability in $n$, the single hypothesis $H_i$ in $\mathcal{H}_{i_{\lg n},\lg n}$ has $W_i\leq W+\eps$ and the desired result follows.

\paragraph{Replicability:}
 Let $X_1$ and $X_2$ be independent sets of samples, and let $\mathcal{A}$ be our algorithm.  Let $i_0^{(1)},\dots, i_{\lg n}^{(1)}$ and $i_0^{(2)},\dots, i_{\lg n}^{(2)}$ be the indices computed by $\A$ when run on samples $X_1$ and $X_2$ respectively. To bound $\Prma[X_1,X_2]{\A(X_1 , r) \neq  \A(X_2, r)}$, we bound the probability that there exists a $j$ such that $i_j^{(1)}\neq i_j^{(2)}$. If no such $j$ exists, then $\A(X_1 , r) =  \A(X_2, r)$. Denote by $E_j$ the event that $i_j^{(1)}\neq  i_j^{(2)}$. By the independence of the samples, and the $\rho_0$-replicability guaranteed by \cref{thm:optimal-coin-testing}, it follows that 
\[
\Prma{E_j\mid\bigcup_{k<j}  E_k^c}\leq \rho_0.
\]
Thus, 
\[
\Prma[X_1,X_2]{\A(X_1 , r) \neq  \A(X_2, r)} \leq \Prma{ \bigcup_{1\leq j\leq \lg n}E_j}\leq 1-(1-\rho_0)^{\lg n}\leq \rho,
\]
as desired.
\end{proof}

\section*{Acknowledgements}
The authors thank Cl\'ement Canonne for his valuable ideas and discussion.

Anders Aamand was supported by the VILLUM Foundation grant 54451. Justin Y. Chen was supported by an NSF Graduate Research Fellowship under Grant No. 17453. 

\bibliographystyle{alpha}
\bibliography{bib}

\appendix
\section{Proof of \cref{prop:basic-thresholding}}\label{sec:appendix_gaussian}
\begin{proof}
    First, we check the accuracy guarantees. If $q(\cD) \ge 5T$, then with probability at least $1-\delta$, $t := h(X) \ge 5T \cdot (1-\delta) \ge 3T,$ so the algorithm rejects since $\gamma = \frac{3T-h(X)}{T} \le 0$. Otherwise, if $q(\cD) \le T$, then with probability at least $1-\delta$, $t := h(X) \le q + \rho \cdot \max(|q|, T) \le 2T$, so the algorithm accepts since $\gamma = \frac{3T-h(X)}{T} \ge 1$.
    
    Next, we check that $\A_{h, T}$ is replicable. Assume $\rho \le 1/10$, as otherwise the claim is trivial. Suppose that we sample $X_1, \dots, X_n$ and $X_1', \dots, X_n'$ i.i.d. from $\cD$. With probability at least $1-2\delta$, both $h(X)$ and $h(X')$ are within $\rho \cdot \max(q, T)$ from $q = q(\cD)$. If $q \ge 5T$, then with probability at least $1-2\delta$, both $t := h(X)$ and $t' := h(X')$ are at least $3T$, in which case the algorithm always rejects. If $q \le 0$, then probability at least $1-2 \delta$, both $t := h(X)$ and $t' := h(X')$ are at most $\rho \cdot T \le T$, so the algorithm always accepts. Alternatively, with probability at least $1-2\delta$, both $t := h(X)$ and $t' := h(X')$ are within $5 \rho \cdot T$ of $q$, so are within $10 \rho \cdot T$ of each other. This means that $\gamma = \frac{3T - h(X)}{T}$ and $\gamma' = \frac{3T - h(X')}{T}$ are within $10 \rho$ of each other. In this case, the probability of selecting $r \sim \Unif([0, 1])$ that lies between $\gamma$ and $\gamma'$ is at most $10 \rho$. Overall, there is at most a $10 \rho + 2 \delta \le 12 \rho$ failure probability that, over the random seed $r$ and samples $X_1, \dots, X_n$, $X_1', \dots, X_n'$ from $\cD$, $\A_{h, T}$ outputs a different result on $X$ and $X'$.
\end{proof}

\end{document}